\pdfoutput=1
\documentclass[aps,prx,twocolumn,superscriptaddress,nofootinbib,longbibliography]{revtex4-2}
\usepackage{amsmath,amssymb,amsthm,mathtools}
\usepackage{graphicx,booktabs}
\usepackage{xcolor}
\usepackage[colorlinks=true,linkcolor=blue,citecolor=blue,urlcolor=blue]{hyperref}
\makeatletter\let\auto@bib@innerbib\@empty\makeatother

\theoremstyle{plain}
\newtheorem{theorem}{Theorem}
\newtheorem{lemma}{Lemma}
\newtheorem{corollary}{Corollary}
\newtheorem{proposition}{Proposition}
\newtheorem{conjecture}{Conjecture}
\theoremstyle{definition}
\newtheorem{definition}{Definition}
\theoremstyle{remark}
\newtheorem{remark}{Remark}

\newcommand{\dproj}{d_{\mathrm{proj}}}
\newcommand{\Tpre}{T_{\mathrm{pre}}}
\newcommand{\Tout}{T_{\mathrm{out}}}
\newcommand{\E}{\mathbb{E}}
\newcommand{\Z}{\mathbb{Z}}
\newcommand{\Q}{\mathbb{Q}}
\newcommand{\R}{\mathbb{R}}
\newcommand{\F}{\mathbb{F}}
\newcommand{\PU}{\mathrm{PU}(2)}
\newcommand{\SU}{\mathrm{SU}(2)}
\newcommand{\tr}{\operatorname{tr}}
\newcommand{\supp}{\operatorname{supp}}
\DeclareMathOperator{\spec}{spec}

\begin{document}

\title{A Memory--Magic Exchange Law in Streaming \texorpdfstring{Clifford+$T$}{Clifford+T} Compilation}

\author{Jinze Yang}
\affiliation{School of Physics, Xidian University, Xi'an 710071, China}

\author{Yangyang Li}
\email{yyli@xidian.edu.cn}
\affiliation{Key Laboratory of Intelligent Perception and Image Understanding of Ministry of Education of China,
Collaborative Innovation Center of Quantum Information of Shaanxi Province, Institute of Interdisciplinary Quantum Science and Technology (IIQST),
School of Artificial Intelligence, Xidian University, Xi'an 710071, China}

\author{Xiu-Hao Deng}
\email{dengxiuhao@iqasz.cn}
\affiliation{Shenzhen International Quantum Academy, Shenzhen 518048, China}
\affiliation{Shenzhen Branch, Hefei National Laboratory, Shenzhen, 518048, China}

\begin{abstract}
A phase that reaches a fault-tolerant processor in several additive pieces can be remembered until the last piece arrives, or executed on arrival: the first option is paid for in classical bits carried across a round boundary, the second in magic states committed before the phase is known. We show that in the ancilla-free coordinatewise Clifford+$T$ model realized by per-rotation synthesis pipelines the two payments are convertible at an exchange rate $\alpha$---committed $T$ gates per bit of memory forgone---that does not depend on the problem size. Unconditionally $\alpha\ge2$: it rests on a new equidistribution estimate for Clifford+$T$ words near rotation cosets, whose only external input is the Ramanujan bound of Lubotzky--Phillips--Sarnak and Parzanchevski--Sarnak, and which shows that a committed word of $T$-count $\tau$ carries at most $\tau/2+O(\log\log(1/\varepsilon))$ bits about the share it commits. An elementary determinant method, with no automorphic input, goes below the square-root barrier of that estimate uniformly over cosets and raises the asymptotic floor to $\alpha\ge\tfrac{11}5$. Splitting the sphere sections that carry the lattice points by the height of their hyperplane, and fibring the numerators over rational planes and projections, raises it to every $\alpha<\tfrac52$, as an SMT solver checks exactly on a continuum form of the resulting case analysis. This is where the method stops: at $\tfrac52$ its boxes shrink to $\varepsilon$-cubes, and a tube carrying words of $T$-count $(\tfrac52+o(1))L$ with gaps of at most $2^{o(L)}$ grid steps along its core would make $\tfrac52$ sharp. The bound is a statement about irreversible commitment rather than about the target: it holds unchanged when the aggregate phase turns out to be zero and the final unitary is the identity. An explicit fractional passthrough family attains $\alpha\to3$ whenever the mean single-rotation synthesis cost on the grid is $(3+o(1))L$, $L:=\log_2(1/\varepsilon)$, as Ross--Selinger synthesis gives for typical angles and as we measure at $\varepsilon=10^{-10}$. Closing the remaining gap is an arithmetic question: under an equidistribution conjecture whose exponent $\beta=2$ is the codimension of a rotation coset in $\SU$---with a linear additive term forced by the powers of infinite-order words, and which exhaustive enumeration to $T$-count $22$ supports on Haar-random and on arithmetic cosets with exponents to three decimals and constants matching the Haar prediction to a few per cent---the rate is exactly $\alpha=1+\beta=3$ and commitment is quantized: a committed word carries at most $O(\log L)$ bits about its coordinate unless it costs at least $2L-O(1)$ gates, so memory should be shed in whole rotations rather than low bits. All constants of the rate-two bound are explicit; at $\varepsilon=10^{-10}$ the rigorous floor is $0.78$ committed $T$ gates per bit against $3.20$ achieved by passthrough at the certified share accuracy. Side information held downstream enters through a conditional entropy: when the earlier shares are an injective function of a $\lambda$-bit seed independent of the aggregate, as in randomized compiling, the information a compiler without the seed is missing is at most $\lambda$ rather than $mL$. Probabilistic mixing of synthesized words rescales the law by one half and does not remove it: under mixing all but the top $\tfrac12L+O(\log L)$ bits of every share are free, the bits that remain are charged in total at the same rate, the rigorous floor is asymptotically $L$ committed $T$ gates per share, against $\tfrac32L$ achievable under a synthesis hypothesis and $\tfrac32L$ with fee $L$ under the same conjecture. Measurement adaptivity with a bounded number of ancillas lowers the rate but not the order, there for the sum of $T$ gates and measurements rather than for committed magic alone, while clean ancillas with a catalytic phase-gradient register let table lookups batch $O(\log\log(1/\varepsilon))$ coordinates and drive the rate to $O(1/\log\log(1/\varepsilon))$: the constant-rate law is specific to coordinatewise synthesis.
\end{abstract}

\maketitle

\section{Introduction}

Two physical resources meet at the boundary between a fault-tolerant quantum processor and its classical control electronics. On the quantum side, every non-Clifford rotation consumes magic states, and the number of $T$ gates is the dominant cost of a logical computation~\cite{BravyiKitaev,Litinski}. On the classical side, an instruction that has been lowered into the pulse stream cannot be taken back, so the only quantity that a round of the computation can pass to the next one is the state it has decided, in advance, to carry~\cite{Caune,Fowler}. These are different resources with different exchange rates against everything else, and there is no a priori reason for them to be convertible into one another.

They become convertible when a phase arrives in pieces. A Trotter step, a randomized-compiling instance, a layerwise parametrized ansatz, or a phase-polynomial emitted by a scheduler presents the angle of a given rotation as a sum of shares delivered in different rounds. Whoever executes the stream then faces a choice at every share: hold it in classical memory until the sum is known and synthesize one rotation, or synthesize the share on arrival and consume the corresponding magic now. Our earlier work~\cite{Paper1} established that the first option is charged in bits: any streaming process that is correct to projective error $\varepsilon$ on all $r$-round dispersed streams of $m$ commuting phases must carry, across the last inter-round cut, $\Omega(m\log(1/\varepsilon))$ bits of restart state or of committed output. The second option was there only estimated, through an empirical synthesis cost of about $3\log_2(1/\varepsilon)$ $T$ gates per rotation. The purpose of this paper is to make the second option a theorem and to determine the exchange rate.

The quantity traded against magic is not storage. A control processor has gigabytes of memory and the snapshot $S$ below is measured in kilobytes---at $m=10^4$ coordinates and $\log_2K=33$ the entire aggregate is $41$\,kB, against the $10^6$ magic states on the other side of the trade. What $S$ measures is \emph{deferred decision}: how much of the computation a schedule has refused to lower into the instruction stream until later rounds have been seen. The cost of not deferring is a property of the commitment and not of the target: under the canonical sharing the first $r-1$ rounds are independent of the aggregate, so every lower bound below holds conditionally on any fixed aggregate, including the zero aggregate whose target unitary is the identity (Corollary~\ref{cor:identity})---a computation that cancels completely can still have consumed $\Theta(m\log(1/\varepsilon))$ magic states before the cancellation was knowable. Deferral is not free in three situations that occur in practice. In a real-time control stack the stream is append-only and the scheduler emits with bounded lookahead~\cite{Caune,Fowler}, so a rotation that is not held must be committed; lookahead depth and snapshot size are different resources, and it is the second that is priced here. In randomized compiling and in any protocol where the shares are produced by different agents, the compiler does not hold the seed that generated the earlier shares; Theorem~\ref{thm:side} shows that what it is then charged for is the entropy it is missing---for shares that are an injective function of a $\lambda$-bit seed independent of the aggregate, at most $\lambda$ and, for two-round streams with uniform aggregate, at least $\lambda-m\log_2(Q/K)$, rather than $m\log_2K$ (Corollary~\ref{cor:side})---which is the precise boundary of the ``randomize and lower late'' rule of~\cite{Paper1}. And when the angles are produced on the quantum side---qubitization, a QROM-loaded phase table, adaptive QSVT phases---no classical snapshot exists to be traded away at all; our proof, which applies Fano's inequality to a classical stream, does not reach that case, and we state it as the natural target for this line of work (Sec.~\ref{sec:disc}). The law is specific to the synthesis mechanism in a definite way: Sec.~\ref{sec:models} shows that probabilistic mixing halves every quantity in it that is linear in $L$ and leaves the rate untouched (Theorem~\ref{thm:mixing}, Corollary~\ref{cor:mixq}), asymptotically in $L$ and vacuously at $\varepsilon=10^{-10}$ except for the rate-one bound, that measurement-adaptive synthesis on $n$ qubits keeps an $\Omega(L/n)$ floor on the number of $T$ gates plus measurements (Proposition~\ref{prop:adaptive}), that the entry fee disappears in the phase-gradient model, that with clean ancillas, a phase-gradient catalyst and table lookups batched across coordinates no constant-rate law survives, the committed cost per share dropping to $O(L/\log L)$ (Proposition~\ref{prop:batch}), and that the quasi-probabilistic case is open.

Our object is the \emph{committed magic} $\Tpre$: the number of $T$ gates executed before the last share of the stream has arrived. A process that holds every share in memory has $\Tpre=0$; a memoryless process must commit each share as it comes. The main results, stated for the ancilla-free coordinatewise Clifford+$T$ model of Sec.~\ref{sec:setting}, are as follows.
\begin{itemize}\setlength\itemsep{1pt}
\item \textbf{Theorem~\ref{thm:main1}} (unconditional). $\Tpre\ \ge\ (r-1)\,(m\log_2K-S)-O(rm\log\log(1/\varepsilon))$, where $S$ is the worst-input expected cut budget of Definition~\ref{def:model} and $K=\Theta(1/\varepsilon)$: one committed $T$ gate per bit not held, per round.
\item \textbf{Theorem~\ref{thm:tube}} (equidistribution). The number of Clifford+$T$ words of $T$-count $\le k$ within $\varepsilon$ of any coset $G_1R_z(\theta)G_2$ of a rotation subgroup is at most $C_12^k\varepsilon^2+C_2(k+1)2^{k/2}$ with $C_1\le408$ and $C_2\le232$.
\item \textbf{Theorem~\ref{thm:rate2}} (rate two). Consequently $\Tpre\ge2(r-1)(m\log_2K-S-mA_2)$ up to $O(\delta)$ corrections, where $A_2=\log_2C_2+2\log_2(2\log_2K)+O(1)$ bits per coordinate: every bit forgone costs at least two committed $T$ gates, once the first $A_2$ bits of each coordinate are given away.
\item \textbf{Theorems~\ref{thm:elem} and~\ref{thm:rate115}} (beyond the spectral bound). A determinant method with no automorphic input bounds the tube by $2^{0.4513t+o(t)}$ for every coset and every $t\le\tfrac{20}9L$, beating the square-root barrier $2^{t/2}$. Hence every bit forgone costs at least $\tfrac{11}5$ committed $T$ gates asymptotically, without (R) and uniformly over frames.
\item \textbf{Theorems~\ref{thm:dioph}, \ref{thm:rate177} and~\ref{thm:rate52}} (height dichotomy and rational projections). Splitting the sphere sections of the determinant method by the height of their hyperplane lowers the exponent to $0.4111$ for $t\le\tfrac{17}7(L+2)$, so every bit forgone costs at least $\tfrac{17}7$ committed $T$ gates asymptotically, again without (R). Fibring the numerators over rational planes and their projections, and replacing the finite certificate by a continuum form of the case analysis that an SMT solver checks exactly, gives every rate below $\tfrac52$ (Theorem~\ref{thm:rate52}). This is the limit of the method with one auxiliary hyperplane per box: $\tfrac52$ is where the boxes of its five-point determinant shrink to $\varepsilon$-cubes, and a tube with words of $T$-count $(\tfrac52+o(1))L$ at gaps of at most $2^{o(L)}$ grid steps along its core would make it sharp (Proposition~\ref{prop:dense}).
\item \textbf{Theorems~\ref{thm:clifford} and~\ref{thm:rate3seg}, Corollary~\ref{cor:typical}} (rate three where the frame is harmless). At Clifford-framed cosets the volume law~\eqref{eq:H} holds up to $2^{O(\tau/\log\tau)}$. Consequently all but a vanishing fraction of $z$-rotations need $T$-count $3L-O(L/\log L)$, unconditionally. For processes whose every committed segment is within $\varepsilon$ of a Clifford-framed $z$-rotation $CR_z(\varphi)C'$, which include per-rotation pipelines and the passthrough family, the exchange rate is at least $3-o(1)$ unconditionally.
\item \textbf{Proposition~\ref{prop:batch}} (the law needs coordinatewise synthesis). With clean ancillas and a phase-gradient catalyst, table lookups batched across coordinates commit a share for $O(L/\log L)$ $T$ gates, so no constant-rate law survives batching.
\item \textbf{Theorem~\ref{thm:frontier}} (achievability). A fractional passthrough scheme attains $\Tpre=(r-1)(m-q)\,\bar\tau_Q(\varepsilon/r)$ at $S=q\log_2Q$ for every $q$, where $\bar\tau_Q$ is the mean single-rotation synthesis cost on the grid, hence rate $\to3$ whenever $\bar\tau_Q(\varepsilon/r)=(3+o(1))\log_2K$, which Ross--Selinger synthesis gives for typical angles and which the grid samples of Sec.~\ref{sec:num} support.
\item \textbf{Conjecture~\ref{conj:H} and Theorem~\ref{thm:main2}} (rate $1+\beta$, quantization). Under a stated equidistribution conjecture---whose exponent $\beta$ is the codimension of a rotation coset in $\SU$, with a linear additive term that the powers of infinite-order words force (Proposition~\ref{prop:cluster}), and which the enumeration supports with Haar-exact constants---the rate is $1+\beta=3$, and a committed word carries at most $O(\log\log(1/\varepsilon))$ bits about its coordinate unless it costs at least $2\log_2(1/\varepsilon)-O(1)$ gates: $7$ bits at $\varepsilon=10^{-10}$, against the $33$ a coordinate carries.
\item \textbf{Theorems~\ref{thm:side} and~\ref{thm:oneshot}} (side information). With side information $Y$ available downstream, $m\log_2K$ is replaced by the conditional entropy $H(X\mid\text{other rounds},Y)$, with a one-shot version in terms of smooth entropies.
\end{itemize}
The two constants $2$ and $3$ are not the same kind of number. A single-qubit Clifford+$T$ word of $T$-count $t$ is one of $36\cdot2^t$ unitaries (Matsumoto--Amano), so a $T$ gate carries at most one bit of angle: that is the rate $1$ of Theorem~\ref{thm:main1}. What changes it is that the committed word is not free in the group but confined to a tube of radius $\varepsilon$ around a rotation coset, a curve of codimension $\beta=2$ in the three-dimensional group. If the words in that tube are distributed as volume predicts, a word of $T$-count $\tau$ can take one of $\approx2^{\tau}\varepsilon^{\beta}$ values there, so carrying $b$ bits costs $\tau\ge b+\beta L$: a marginal rate of one gate per bit on top of an entry fee of $\beta L$ gates, which over the $\log_2K\approx L$ bits of one coordinate averages to $1+\beta=3$. This is the same count as the Ross--Selinger cost $3L$ of a single rotation, seen from the counting rather than the synthesis side. What we can prove is weaker in a specific way: pointwise counting through Hecke operators controls the tube only up to an error $2^{\tau/2}$, the square root of the number of all words of that cost (Remark~\ref{rem:barrier}), leaving a profile $2^{\tau/2}$, under which carrying $b$ bits costs $\tau\ge2b$---no entry fee, but a marginal rate $1+\beta/2=2$. Rate $2$ is therefore the natural threshold of the spectral method rather than a geometric constant. It is not a threshold of the problem: an elementary determinant method, which uses only that the numerators are lattice points on a sphere in both real embeddings, controls the tube below the square root in the relevant range and gives the rate $\tfrac{11}5$ (Theorem~\ref{thm:rate115}), $\tfrac{17}7$ once the heights of the hyperplanes it produces are taken into account (Theorem~\ref{thm:rate177}), and every rate below $\tfrac52$ once the numerators are also fibred over rational projections (Theorem~\ref{thm:rate52}). Closing the remaining gap to $1+\beta$ is a problem in the geometry of numbers and in arithmetic, not a compilation one. The statements are asymptotic in $L$. Sec.~\ref{sec:num} evaluates the additive terms of the rate-two and conditional bounds at $\varepsilon=10^{-10}$ (those of Theorems~\ref{thm:rate115}, \ref{thm:rate177} and~\ref{thm:rate52} are not explicit), where the unconditional floor is $0.78$ and the conditional one $1.61$ committed $T$ gates per bit against $3.20$ achieved.

\paragraph{What is new.} Relative to~\cite{Paper1}, whose bit-level bound Theorem~\ref{thm:main1} reproves with the Matsumoto--Amano count in place of Kraft's inequality, the new unconditional results are Theorem~\ref{thm:tube}, a count of Clifford+$T$ words near a rotation coset with explicit constants, and Theorem~\ref{thm:rate2}, which turns it into the rate $2$: the counting results of~\cite{PS18} bound covering exponents in the whole group, and Selinger's~\cite[Thm.~30]{Selinger15} bounds the worst-case cost of one angle; neither controls how many cheap words lie near a coset, which is what a streaming lower bound needs. Corollary~\ref{cor:identity} and Theorems~\ref{thm:side}--\ref{thm:oneshot} are new and elementary. Theorem~\ref{thm:mixing} extends the unconditional bounds to programs that are sampled at run time; its proof is the same counting argument applied to a representative word selected by a fidelity-concentration inequality (Lemma~\ref{lem:conc}), and it locates the small-angle savings of~\cite{Bothe} between bit $\tfrac12L-\tfrac12\log_2L$, up to which they are known to be free, and bit $\tfrac12L$, from which Lemma~\ref{lem:nocheap} shows no mixed program is cheap. Theorem~\ref{thm:main2} is conditional on Conjecture~\ref{conj:H}, whose exponent is fixed by geometry and whose constants we test but do not prove; Theorem~\ref{thm:frontier} is an elementary construction whose rate depends on a synthesis hypothesis, Corollary~\ref{cor:frontier}. We are not aware of an earlier bound below the square-root barrier of the spectral method for Clifford+$T$ words (points of this $S$-arithmetic quaternion lattice) near an arbitrary rotation coset. Theorem~\ref{thm:elem} is one, and Theorem~\ref{thm:rate115} turns it into the rate $\tfrac{11}5$; Theorem~\ref{thm:dioph} sharpens it by a height dichotomy and gives the rate $\tfrac{17}7$ (Theorem~\ref{thm:rate177}), and Theorem~\ref{thm:rate52} takes the same method to every rate below $\tfrac52$, which is the exact limit of this case analysis. For points of spheres near planes and in caps see~\cite{BourgainRudnick}, and for the related question of optimal strong approximation~\cite{BKS,Sardari}. The rigorous content of the paper is therefore the set (Theorems~\ref{thm:rate2},~\ref{thm:rate52},~\ref{thm:rate3seg},~\ref{thm:frontier}):
\begin{itemize}\setlength\itemsep{0pt}
\item two gates per bit are necessary with every additive term explicit;
\item $\tfrac52-o(1)$ gates per bit are necessary asymptotically in general (Theorem~\ref{thm:rate52}), and $3-o(1)$ for rotation-segmented processes (Theorem~\ref{thm:rate3seg});
\item three gates per bit are sufficient under the synthesis hypothesis.
\end{itemize}

\section{Setting}\label{sec:setting}

The model is the streaming compilation model of~\cite{Paper1}. So that the present paper can be read on its own, Appendix~\ref{app:model} restates the four items we use---the charged one-pass process and its snapshots, the dispersed family, the closed form of the projective distance, and the block compiler---as Definitions~\ref{def:model} and~\ref{def:family}, Lemma~\ref{lem:dproj} and Proposition~\ref{prop:block}, with proofs. Nothing below depends on any other result of~\cite{Paper1}; in particular Theorem~\ref{thm:main1} reproves its bit-level bound.

\paragraph{Family.} Fix a modulus $Q\ge3$, step $\Delta=2\pi/Q$, alphabet size $K=Q-1$, and $m$ $\F_2$-independent $Z$-characters $P_1,\dots,P_m$ on $n\ge m$ qubits; let $V$ be a fixed public Clifford with $VZ_jV^\dagger=P_j$. The aggregate is $x\in[K]^m$, $[K]=\{0,\dots,K-1\}$, and the target is
\begin{equation}
U_x=V\Big(\bigotimes_{j\le m}R_z(\Delta x_j)\otimes I\Big)V^\dagger,\qquad R_z(\theta)=e^{-i\theta Z/2}.
\end{equation}
An $r$-round dispersed stream is a round-major sequence of shares $a^{(1)},\dots,a^{(r)}\in\Z_Q^m$ with $\sum_ta^{(t)}\equiv x\pmod Q$. The \emph{canonical sharing} $\mathcal D$ takes $a^{(1)},\dots,a^{(r-1)}$ i.i.d.\ uniform on $\Z_Q^m$ and $a^{(r)}=x-\sum_{t<r}a^{(t)}$. The accuracy satisfies $0<\varepsilon<\sin(\pi/2Q)$, so that distinct grid rotations are separated by $2\sin(\pi/2Q)>2\varepsilon$ in the projective distance $\dproj(A,B)=\inf_\phi\|A-e^{i\phi}B\|$ and nearest-grid decoding is unambiguous. The \emph{tuned} modulus $Q_\varepsilon:=\lfloor\pi/(2\arcsin2\varepsilon)\rfloor$ is the largest for which the separation is at least $4\varepsilon$; it satisfies $\log_2K_\varepsilon=L-\log_2(4/\pi)+o(1)=L-0.35+o(1)$, where throughout
\begin{equation}
L:=\log_2(1/\varepsilon).
\end{equation}

\paragraph{Output model.} The \emph{coordinatewise ancilla-free Clifford+$T$ model} (CW) is the one realized by per-rotation synthesis pipelines, in which each generator of the phase polynomial is handed to a single-qubit synthesizer and no ancilla or cross-coordinate cancellation is used~\cite{RossSelinger,Paper1}: the write-once output is a list of records $(j,g)$ with $j\in[m]$ and $g\in\{H,S,T\}$, the implemented unitary is $V(\bigotimes_jW_j)V^\dagger$ where $W_j$ is the ordered product of the gates recorded on coordinate $j$, and the $T$-count is the number of $T$ records. We say the output is \emph{correct} if $\dproj(W_j,R_z(\Delta x_j))\le\varepsilon$ for every $j$. By Lemma~\ref{lem:loc}, correctness of the whole unitary---$\dproj(V(\bigotimes_jW_j)V^\dagger,U_x)\le\varepsilon$, the criterion of~\cite{Paper1}---implies coordinatewise correctness, so every lower bound below also holds under the total criterion. The converse is not asserted: a construction with coordinate error $\varepsilon$ need only have total error at most $m\varepsilon$, and for a total error budget $\varepsilon_{\mathrm{tot}}$ the share-by-share construction of Sec.~\ref{sec:frontier} can use $\varepsilon_{\mathrm{tot}}/(mr)$ per share.

\paragraph{Process and committed magic.} The process is a compiler in the charged one-pass model of Definition~\ref{def:model}: it reads the stream once, may use arbitrary RAM and a random tape whose state is part of every restart snapshot, and $S:=\max_{\text{input}}\E_\rho\max_cB_{\mathrm{cut}}(c)$ is the worst-input expected snapshot size. It is required to be correct with probability $\ge1-\delta$ on every $x$ and every valid stream, $\delta<\tfrac12$. Let $T_t$ be the number of $T$ records emitted while reading round $t$, $\Tpre:=\sum_{t<r}T_t$ the \emph{committed magic}, and $\Tout=\Tpre+T_r$. Bars denote expectation over the tape, the canonical sharing and a uniform aggregate. We write $L_\delta(m,K):=(1-\delta)m\log_2K-h_2(\delta)$, $\ell_\delta:=(1-\delta)\log_2K-h_2(\delta)$, and $\rho_m(\bar T):=m\log_2(e(1+\bar T/m))+2\log_2(1+\bar T)+1$.

\section{Three lemmas}

\begin{lemma}[Localization]\label{lem:loc}
For unitaries $A=\bigotimes_jA_j$, $B=\bigotimes_jB_j$ on the same tensor factors,
$\max_j\dproj(A_j,B_j)\le\dproj(A,B)\le\sum_j\dproj(A_j,B_j)$. For $2\times2$ unitaries, $\dproj(A,B)=\sqrt{2-|\tr(B^\dagger A)|}$.
\end{lemma}

\begin{lemma}[Matsumoto--Amano count]\label{lem:MA}
The number of single-qubit Clifford+$T$ unitaries modulo phase of minimal $T$-count exactly $t$ is $n_1(0)=24$ and $n_1(t)=36\cdot2^t$ for $t\ge1$; hence $\sum_{s\le t}n_1(s)=72\cdot2^t-48$, and the number of $m$-tuples with total minimal $T$-count at most $T$ is $N^{\le}_m(T)\le36^m2^T\binom{T+m}{m}$.
\end{lemma}

\begin{lemma}[Entropy versus $T$-count]\label{lem:ent}
Let $M=(\overline W_1,\dots,\overline W_m)$ be a random $m$-tuple of single-qubit Clifford+$T$ unitaries and $T\ge\sum_jt_{\min}(\overline W_j)$ a random integer. Then $H(M)\le\bar T+m\log_236+\rho_m(\bar T)$.
\end{lemma}
Proofs are in Appendix~\ref{app:lemmas}. Lemma~\ref{lem:MA} is the source of the exchange rate: a $T$ gate carries at most $\log_2 2=1$ bit.

\section{Committed magic: the unconditional bound}

\begin{theorem}[Committed magic]\label{thm:main1}
For every process as in Sec.~\ref{sec:setting} and every round $t\in[r-1]$,
\begin{equation}\label{eq:main1}
S+\bar T_t+m\log_236+\rho_m(\bar T_t)\ \ge\ L_\delta(m,K).
\end{equation}
Consequently
\begin{multline}
\bar\Tpre\ \ge\ (r-1)\big[L_\delta(m,K)-S-m\log_236\big]\\-(r-1)\,\rho_m\!\Big(\frac{\bar\Tpre}{r-1}\Big),
\end{multline}
and for the tuned family and $\delta=0$, $\Tpre\ge(r-1)(m\log_2K_\varepsilon-S)-O(rm\log\log(1/\varepsilon))$.
\end{theorem}

\begin{proof}
Fix $t$. Let $A:=a^{(t)}$, $\mathrm{Past}:=(a^{(s)})_{s<t}$, $\mathrm{Fut}:=(a^{(s)})_{s>t}$ (which includes round $r$), and $Z:=(\mathrm{Past},\mathrm{Fut},\rho)$.

\emph{(i) $H(A\mid Z)=m\log_2K$.} Under $\mathcal D$ the joint law of $(a^{(1..r-1)},x)$ is $Q^{-m(r-1)}P_X(x)$ with $a^{(r)}$ determined. Conditioning on $(\mathrm{Past},\mathrm{Fut})$ fixes $c:=\sum_{s\ne t}a^{(s)}$ and $x=A+c$, so the conditional law of $A$ is proportional to $P_X(A+c)$, i.e.\ uniform on the translate $[K]^m-c$. The tape is independent of the stream.

\emph{(ii) Message and Fano.} Let $M_t:=(\overline W^{(t)}_1,\dots,\overline W^{(t)}_m)$ be the tuple of unitaries (modulo phase) implemented by the gate segments emitted on the respective coordinates during round $t$, and $\sigma_t$ the restart snapshot at the cut after round $t$. Given $\rho$ the process is deterministic; the segments emitted before round $t$ are a function of $(\mathrm{Past},\rho)$, and those emitted after the cut are a function of $(\sigma_t,\mathrm{Fut})$ by the restart semantics of Definition~\ref{def:model}\,(ii). Hence the final unitary on coordinate $j$, $U^{<t}_j\overline W^{(t)}_jU^{>t}_j$, is a function of $(M_t,\sigma_t,Z)$. On success it is within $\varepsilon$ of $R_z(\Delta x_j)$, and nearest-grid decoding recovers $x_j$ because neighbouring grid rotations are more than $2\varepsilon$ apart; thus $A=x-c$ is recovered with error probability $\le\delta$. Fano's inequality gives $H(A\mid M_t,\sigma_t,Z)\le h_2(\delta)+\delta\log_2(K^m-1)$, i.e.\ $I(A;M_t,\sigma_t\mid Z)\ge L_\delta(m,K)$.

\emph{(iii) Message size.} $I(A;M_t,\sigma_t\mid Z)\le H(\sigma_t)+H(M_t)$. The snapshot serialization is prefix-free (Definition~\ref{def:model}\,(iii)), so $H(\sigma_t)\le\E|\sigma_t|\le S$, the worst-input expectation dominating the average over inputs. The number $T_t$ of $T$ records emitted in round $t$ is at least $\sum_jt_{\min}(\overline W^{(t)}_j)$, so Lemma~\ref{lem:ent} gives $H(M_t)\le\bar T_t+m\log_236+\rho_m(\bar T_t)$. This proves~\eqref{eq:main1}; summing over $t$ and using concavity of $\rho_m$ proves the rest.
\end{proof}

\begin{remark}[Multiple passes]
For $p$ forward passes, the alternating-simulation transcript---all $2p-1$ crossing snapshots together with the unitary tuple committed while scanning rounds $<r$ in any pass, cut between rounds $<r$ and round $r$ as in~\cite{Paper1}---determines the output, and the same argument gives $\Sigma_p+\bar T_A+m\log_236+\rho_m(\bar T_A)\ge L_\delta$ with $\Sigma_p\le(2p-1)S$. The per-round splitting, and with it the factor $r-1$, is a one-pass phenomenon: a process that rereads round $t$ after seeing round $r$ need not commit anything during round $t$.
\end{remark}

\begin{corollary}[Commitment before cancellation is knowable]\label{cor:identity}
Under the canonical sharing $\mathcal D$ the rounds $a^{(1)},\dots,a^{(r-1)}$ are i.i.d.\ uniform and independent of the aggregate $x$, and $\Tpre$ is a function of $(a^{(1)},\dots,a^{(r-1)},\rho)$ alone. Hence for every fixed aggregate $x\in[K]^m$,
\[
\E\big[\Tpre\,\big|\,X=x\big]=\bar\Tpre ,
\]
and the lower bounds of Theorems~\ref{thm:main1}, \ref{thm:rate2}, \ref{thm:rate115}, \ref{thm:rate177}, \ref{thm:rate52}, \ref{thm:rate3seg} and~\ref{thm:main2} hold after conditioning on $X=0$, for which $U_x=I$. A process that is required to be correct on every valid stream, and is not promised in advance that the aggregate vanishes, commits the same expected magic whether or not the phases cancel.
\end{corollary}
\begin{proof}
Independence is item (i) of the proof of Theorem~\ref{thm:main1}; $\Tpre$ counts records emitted while reading rounds $<r$, which by Definition~\ref{def:model}(ii) depend only on those rounds and on $\rho$.
\end{proof}
The corollary is a restatement of the sharing distribution, not an independent strengthening of the bound. A compiler that is told the promise $X=0$ can emit nothing; the statement concerns the resources a universally correct process has already spent on a stream that turns out to cancel, and it separates committed cost from the intrinsic synthesis cost of the completed target, which for the identity is zero.

Theorem~\ref{thm:main1} is the bit-level theorem of~\cite{Paper1} with Lemma~\ref{lem:MA} in place of Kraft's inequality; its content is that the counting rate $\log_2 2=1$ is the worst case. The rest of the paper is about the fact that the true rate is larger.

\section{Words near rotation cosets}\label{sec:tube}

\paragraph{Arithmetic structure.} Single-qubit Clifford+$T$ modulo phase is the full $\{(\sqrt2)\}$-arithmetic group $\Gamma$ of the maximal order
\begin{equation}
\mathcal O=\mathcal O_K\oplus\mathcal O_K\tfrac{1+i}{\sqrt2}\oplus\mathcal O_K\tfrac{1+j}{\sqrt2}\oplus\mathcal O_K\tfrac{1+i+j+k}{2}\label{eq:order}
\end{equation}
of the definite quaternion algebra $D=(-1,-1)_{\Q(\sqrt2)}$, where $\mathcal O_K=\Z[\sqrt2]$~\cite{KMM13,PS18}. $D$ splits at the prime $P=(\sqrt2)$; the completion $K_P=\Q_2(\sqrt2)$ has residue field $\F_2$, and $\Gamma$ acts on the $3$-regular Bruhat--Tits tree $X_P$ with vertex stabilizer $U=\mathcal O^\times/\mathcal O_K^\times\cong S_4$, the single-qubit Clifford group modulo phase, of order $24$~\cite[\S4.1.3]{PS18}. \emph{$T$-count equals tree distance:} the sphere of radius $t$ has $3\cdot2^{t-1}$ vertices, and $24\cdot3\cdot2^{t-1}=36\cdot2^t$ is exactly Lemma~\ref{lem:MA}. Let $\Lambda_k\subset\Gamma$ be the ball of $T$-count $\le k$, $N_k=|\Lambda_k|=72\cdot2^k-48$, and for a left-$U$-invariant $h\in L^2(\PU)$ let $T_jh(g):=\sum_{\gamma U:\,d(\gamma v_0,v_0)=j}h(\gamma^{-1}g)$ be the Hecke operator of level $j$.

\paragraph{Input (R).} On $L^2_0(U\backslash\PU)$ one has $\|T_j\|\le(j+1)2^{j/2}$~\cite[Eq.~(3.10)]{PS18}. This is the Ramanujan bound~\cite{LPS} for the Clifford+$T$ lattice, an $S$-arithmetic subgroup of the definite quaternion algebra over $\Q(\sqrt2)$ split at the prime above $2$. Its proof in~\cite{PS18} rests on the Ramanujan property of the automorphic representations that occur in $L^2(\Gamma\backslash G)$; for the forms that arise here that property is a theorem, by Deligne's bound in the elliptic modular case~\cite{Deligne} with the local--global compatibility of~\cite{HarrisTaylor}, transferred by Jacquet--Langlands to the quaternionic side~\cite{Blasius}. We use only the resulting bound, and refer to~\cite[Sec.~3]{PS18} for the dictionary between the spherical harmonics of Appendix~\ref{app:tube} and the weights of the corresponding Hilbert modular forms. The non-backtracking version~\cite[Eqs.~(5.2)--(5.3)]{PS18} gives the same bound directly for the Matsumoto--Amano syllables $HT$ and $SHT$.

\begin{theorem}[Tube bound]\label{thm:tube}
For all $G_1,G_2\in\SU$, all $\varepsilon\in(0,1/8]$ and all $k\ge0$, with
$T_\varepsilon(C):=\{W:\exists\theta,\ \dproj(W,G_1R_z(\theta)G_2)\le\varepsilon\}$,
\begin{equation}\label{eq:tube}
\#\big(\Lambda_k\cap T_\varepsilon(C)\big)\ \le\ C_1\,2^k\varepsilon^2+C_2\,(k+1)\,2^{k/2},
\end{equation}
with the absolute constants $C_1=72\cdot4\sqrt2\le408$ and $C_2=24(1-2^{-1/2})^{-1}2^{3/2}\le232$.
\end{theorem}
The proof (Appendix~\ref{app:tube}) reduces the tube to a spherical cap via the Hopf fibration whose fibre is the coset, majorizes the cap by a Poisson kernel of width $2\varepsilon$---whose spherical-harmonic coefficients are $r^\ell$ with $1-r=2\varepsilon$, so that both constants can be read off---and bounds the count by the Hecke operator through a rank-one sup-norm estimate.

\begin{remark}[The square-root barrier]\label{rem:barrier}
Any majorant of a cap of radius $\varepsilon$ on $S^2$ has spectral mass up to degree $\gtrsim1/\varepsilon$, and the Hecke bound (R) is applied to each degree separately, so the error term $(k+1)2^{k/2}$ in~\eqref{eq:tube} cannot be reduced by this method. It exceeds the main term $C_12^k\varepsilon^2$ whenever $k<4L+2\log_2(k+1)-2\log_2(C_1/C_2)$, which for $k\ge1$ includes the whole range $k\le4L$, and in particular the range $k\in[2L,3L]$ in which synthesized rotations live. The true count in this range (Sec.~\ref{sec:num}) is the main term; proving that is the content of Conjecture~\ref{conj:H}. The barrier belongs to the method: for $k\le\tfrac{20}9L$, Theorem~\ref{thm:elem} goes below it by geometry of numbers, for $k\le\tfrac{17}7(L+2)$ Theorem~\ref{thm:dioph} does so with exponent $0.411$, and Theorem~\ref{thm:rate52} does so up to $T$-count $(\tfrac52-o(1))L$ with exponents below $\tfrac25$.
\end{remark}

\begin{corollary}[Alphabet profile]\label{cor:profile}
In the setting of Theorem~\ref{thm:main1}, conditional on $(\sigma_t=s,Z=z)$ and on the correctness of coordinate $j$, the unitary $\overline W^{(t)}_j$ committed on coordinate $j$ during round $t$ lies in $T_\varepsilon(C_{s,z,j})$, the tube around the coset with
\[G_1=(U^{<t}_j)^\dagger,\qquad G_2=(U^{>t}_j)^\dagger,\]
where $U^{<t}_j$ is the segment emitted before round $t$ and $U^{>t}_j$, the segment emitted after the cut, is a function of $(s,\mathrm{Fut})$; the grid shift $R_z(\Delta c_j)$ is absorbed into $\theta$, whose admissible values form the full grid because the grid is invariant under translation by $\Delta c_j$. Hence, with $n(\tau):=C_12^\tau\varepsilon^2+C_2(\tau+1)2^{\tau/2}$,
\[\#\{\overline W^{(t)}_j:t_{\min}\le\tau\}\le n(\tau),\]
and $\log_2n(\tau)\le\tau/2+\log_2(\tau+1)+1+\log_2C_2$ for every $\tau\le4L$.
\end{corollary}

\begin{theorem}[Rate two]\label{thm:rate2}
Under (R), put $\tau_*:=\lceil2\log_2K\rceil$ and
\begin{equation}\label{eq:A2}
A_2:=1+\log_2 Z,\qquad Z:=\sum_{\tau\le\tau_*}n(\tau)2^{-\tau/2},
\end{equation}
so that $A_2=\log_2C_2+2\log_2(2\log_2K)+O(1)=O(\log\log(1/\varepsilon))$. Then for every round $t$,
\begin{equation}
\bar T_t\ \ge\ 2\big[(1-2\delta)m\log_2K-2m\,h_2(\delta)-S-mA_2\big],
\end{equation}
hence $\bar\Tpre\ge2(r-1)[(1-2\delta)m\log_2K-2mh_2(\delta)-S-mA_2]$.
\end{theorem}
The proof (Appendix~\ref{app:rate}) shows that, given the snapshot and the other rounds, a committed word of $T$-count $\tau$ carries at most $\tau/2+A_2$ bits about its coordinate, and that words of cost above $\tau_*$ cannot help because a coordinate carries at most $\log_2K\le\tau_*/2$ bits. Theorem~\ref{thm:rate2} bounds a rate; it contains no statement about quantization, because the profile $(\tau+1)2^{\tau/2}$ is not $O(1)$ at $\tau\le2L$. Nor is $A_2$ negligible in practice: at $\varepsilon=10^{-10}$ it is $20.0$ bits out of the $\log_2K=32.9$ bits a coordinate carries, and the rate-two bound is the larger of the two bounds with explicit constants only where fewer than $3.7$ bits per coordinate are retained, i.e.\ near the memoryless end of the frontier (Sec.~\ref{sec:num}).

\paragraph{Beyond the spectral bound.} The square-root remainder of Theorem~\ref{thm:tube} is a limit of the spectral method, not of the problem. A determinant method in the style of Bombieri--Pila and Heath-Brown~\cite{BombieriPila,HeathBrown}, adapted as in the cap lemmas of Bourgain--Rudnick~\cite{BourgainRudnick}, which uses no automorphic input and no arithmetic property of the frame, gives a better profile in the whole range of $T$-counts that the streaming bound needs.

\begin{theorem}[Elementary tube bound]\label{thm:elem}
There is a function $\eta(t)=O(t/\log t)$ such that, for all $G_1,G_2\in\SU$, all $\varepsilon\in(0,1/8]$ and all $t\le\tfrac{20}9L$, with $C=G_1R_zG_2$,
\begin{equation}\label{eq:elem}
\#\{W:\ t_{\min}(W)\le t,\ W\in T_\varepsilon(C)\}\ \le\ 2^{\,0.4513\,t+\eta(t)} .
\end{equation}
\end{theorem}
In this range the bound of Theorem~\ref{thm:tube} is dominated by its remainder $2^{t/2}$, and~\eqref{eq:elem} improves the exponent from $\tfrac12$ to $0.4513$ uniformly over all cosets, arithmetic or not. The proof (Appendix~\ref{app:elem}) lifts each word to its quaternion numerator $w$, a point of a lattice in $\R^4\times\R^4$ lying on a sphere of radius $R'\asymp2^{t/4}$ in both real embeddings $\sigma_1,\sigma_2$ of $\Q(\sqrt2)$; the tube is the $\varepsilon$-neighbourhood of a great circle in the first embedding and imposes nothing in the second. The tube is cut into boxes. In each box, the affine determinant of any five points lies in $\tfrac1{16}\Z[\sqrt2]$ and the product of its two conjugates is $<2^{-8}$. It must therefore vanish, so the points of the box lie on a round $2$-sphere. The part of that sphere inside the box is covered by cells of an anisotropic grid, a second determinant puts the points of each cell on a circle, and the divisor bound in a CM extension of $\Q(\sqrt2)$ leaves $2^{O(t/\log t)}$ points on each circle. The worst case is a small sphere that the core circle meets almost tangentially. The exponent $0.4513$ comes from one explicit choice of box shape and an exact rational certificate over all sphere radii.

\begin{theorem}[Rate eleven fifths]\label{thm:rate115}
For every process as in Sec.~\ref{sec:setting} and every round $t<r$,
\begin{multline}\label{eq:rate115}
\bar T_t\ \ge\ \tfrac{11}5\big[(1-2\delta)m\log_2K-2m\,h_2(\delta)\\-S-mA_E\big],\qquad A_E=O(L/\log L),
\end{multline}
once $L\ge L_0$ ($L_0=110$ suffices for $\tau_*\le\tfrac{20}9L$ below); hence $\bar\Tpre\ge\tfrac{11}5(r-1)\big(m\log_2K-S-O(mL/\log L)\big)-O(\delta rmL)$.
\end{theorem}
\begin{proof}
Follow the proof of Theorem~\ref{thm:rate2}, replacing the profile of Corollary~\ref{cor:profile} by~\eqref{eq:elem}, which holds for the coset of Corollary~\ref{cor:profile} whatever its frame, and applying Lemma~\ref{lem:gibbs} with $\lambda=\tfrac5{11}$ and $\tau_*=\lceil\tfrac{11}5\log_2K\rceil$. For $L\ge L_0$, $\tau_*\le\tfrac{20}9L$, since $\log_2K\le L+\log_2(\pi/2)$. Then $Z_\lambda\le\sum_{\tau\le\tau_*}2^{(0.4513-5/11)\tau+\eta(\tau)}\le(\tau_*+1)2^{\max_{\tau\le\tau_*}\eta(\tau)}$ because $0.4513<\tfrac5{11}$, so $A_E:=1+\log_2Z_\lambda=O(L/\log L)$, and every step that produced the factor $2$ now produces $\tfrac{11}5$.
\end{proof}
Theorem~\ref{thm:rate115} does not use (R), and it raises the rate; it applies to every round and every frame. Optimising the box shape gives more. With $(a_0,a,b)=(\tfrac{179}{100},\tfrac{179}{100},\tfrac{737}{500})$ the same lemmas and an exact certificate give the exponent $0.4478<\tfrac{100}{223}$ in~\eqref{eq:elem} for $t\le\tfrac{400}{179}L$, hence the rate $2.23$ for $L\ge534$ (App.~\ref{app:elem}). Numerically, about $2.234$ is the limit of these lemmas. The price is the additive term. $A_E$ is controlled by the worst-case divisor bound, and the theorem only applies for $L\ge L_0$, far beyond $\varepsilon=10^{-10}$. Theorem~\ref{thm:rate2} therefore remains the bound we evaluate in Sec.~\ref{sec:num}. Theorem~\ref{thm:rate115} is a statement about the rate.

The loss in Theorem~\ref{thm:elem} is concentrated in the second step of its proof, on small sphere sections $Y=S^3\cap H$ that the core meets almost tangentially. Part of it can be removed by using the arithmetic of the hyperplane $H$. For a primitive normal $n$ of $H$ put $h(n):=|\sigma_1n|\,|\sigma_2n|$, the product of its lengths in the two real embeddings. This \emph{height} is invariant under units, and the points of $\mathcal O$ parallel to $H$ form a lattice of covolume $\asymp h(n)$. The two regimes are handled differently.
\begin{itemize}
\item \emph{High height.} A section of large height carries few points in each cell of the grid. The coplanarity step becomes a fibration by a short vector of the dual lattice, and the cell volume that forces coplanarity grows by the factor $h(n)$.
\item \emph{Low height.} Sections of small height are few near the core. Their normals lie in a thin convex body around the core plane, and they are counted by successive minima (Henk's bound~\cite{Henk}) together with the fact that a nonzero wedge of lattice vectors has height at least one. When exactly two minima are small, the normals lie in a rank-two $\Q(\sqrt2)$-plane $\mathcal V$. Then either a sector count in $\mathcal V$ applies, or all numerators in the tube fibre into circles over $\mathcal V^\perp$ and are few.
\end{itemize}

\begin{theorem}[Height dichotomy]\label{thm:dioph}
There is a function $\eta(t)=O(t/\log t)$ such that, for all $G_1,G_2\in\SU$, all $\varepsilon\in(0,1/8]$ and all $t\le\tfrac{17}7(L+2)$, with $C=G_1R_zG_2$,
\begin{equation}\label{eq:dioph}
\#\{W:\ t_{\min}(W)\le t,\ W\in T_\varepsilon(C)\}\ \le\ 2^{\,\kappa_Dt+\eta(t)},
\end{equation}
where $\kappa_D:=\tfrac{11183}{27200}=0.41114<\tfrac7{17}$.
\end{theorem}

\begin{theorem}[Rate seventeen sevenths]\label{thm:rate177}
For every process as in Sec.~\ref{sec:setting} and every round $t<r$,
\begin{multline}\label{eq:rate177}
\bar T_t\ \ge\ \tfrac{17}7\big[(1-2\delta)m\log_2K-2m\,h_2(\delta)\\-S-mA_D\big],\qquad A_D=O(L/\log L);
\end{multline}
hence $\bar\Tpre\ge\tfrac{17}7(r-1)\big(m\log_2K-S-O(mL/\log L)\big)-O(\delta rmL)$.
\end{theorem}
\begin{proof}
As for Theorem~\ref{thm:rate115}, with~\eqref{eq:dioph} in place of~\eqref{eq:elem}, $\lambda=\tfrac7{17}$ and $\tau_*=\lceil\tfrac{17}7\log_2K\rceil$. Since $\log_2K\le L+\log_2(\pi/2)$, we have $\tau_*\le\tfrac{17}7(L+2)$, and $\kappa_D<\tfrac7{17}$ gives $Z_\lambda\le(\tau_*+1)2^{\max_{\tau\le\tau_*}\eta(\tau)}$, so $A_D:=1+\log_2Z_\lambda=O(L/\log L)$.
\end{proof}
The proof of Theorem~\ref{thm:dioph} (Appendix~\ref{app:dioph}) uses the box shape $a_0=a=\tfrac{28}{17}$, $b=\tfrac{157}{100}$ and an exact rational certificate over sphere radii, heights and crossing depths. Its margin is small: $\tfrac7{17}-\kappa_D\approx6\cdot10^{-4}$ per $T$ gate. The implied constants in $\eta(t)$ and $A_D$, and the $L$ beyond which the bound is non-trivial, are therefore astronomically large, and like Theorem~\ref{thm:rate115} the result is a statement about the rate only. Numerically this set of lemmas stops near $2.435$, so $\tfrac{17}7$ is essentially their limit. It is not the limit of the method.

\paragraph{Rational projections and the rate five halves.} The remaining loss sits in Lemma~\ref{lem:F5}, where the normals of the low sections span a rational plane or a rational hyperplane. Three further ingredients remove it (Appendix~\ref{app:proj}).
\begin{itemize}
\item \emph{Annular fibration.} If the normals span a rank-two $\Q(\sqrt2)$-plane $\mathcal V$, the numerators in the tube fibre into circles over $\mathcal V^\perp$, and the fibre values lie near an ellipse, the image of the core. Covering that ellipse by thin rectangles and counting the projected lattice in each gives a bound in terms of the height of $\mathcal V$ and of its larger principal angle with the core plane (Lemma~\ref{lem:G1}).
\item \emph{Projections of a box.} If the normal of a box's hyperplane lies in a rational plane, or is orthogonal to a rational vector, projecting the numerators of the box onto that plane, or along that vector, puts them on a rational line or plane of controlled covolume (Lemma~\ref{lem:G3}).
\item \emph{Rank one.} A configuration of normals of rank one contributes a bounded number of normals (Lemma~\ref{lem:G2}).
\end{itemize}
An exact rational certificate of the resulting case analysis gives the rate $\tfrac{249}{100}$. Finite certificates lose the width of their discretisation near the limit. To remove that loss we pass to a continuum form of the case analysis, in which every exponent is piecewise linear in the parameters, and check it exactly on a whole interval of box shapes with the SMT solver z3~\cite{Z3} (Appendix~\ref{app:rate52}).

\begin{theorem}[Rate five halves]\label{thm:rate52}
For every $\alpha<\tfrac52$ there is $L_0(\alpha)$ such that, for every process as in Sec.~\ref{sec:setting}, every round $t<r$ and every $L\ge L_0(\alpha)$,
\begin{multline}\label{eq:rate52}
\bar T_t\ \ge\ \alpha\big[(1-2\delta)m\log_2K-2m\,h_2(\delta)\\-S-mA\big],\qquad A=O_\alpha(1).
\end{multline}
Behind it is the following tube bound. For every $0<\mu\le\tfrac1{10}$ there is $\eta_\mu(t)=o(t)$ such that, for all $G_1,G_2\in\SU$, all $\varepsilon\in(0,1/8]$ and all $t\le4(L+2)/(\tfrac85+\mu)$, with $C=G_1R_zG_2$,
\begin{equation}\label{eq:rate52tube}
\#\{W:\ t_{\min}(W)\le t,\ W\in T_\varepsilon(C)\}\ \le\ 2^{(\frac25-\frac\mu{11})t+\eta_\mu(t)}.
\end{equation}
\end{theorem}
The exponent $\tfrac25-\tfrac\mu{11}$ is below $1/\alpha=\tfrac14(\tfrac85+\mu)$ by a fixed amount, which is why the additive term is $O_\alpha(1)$. Near the limit it is also the sharp exponent of this case analysis: the same computation with $\tfrac25-\nu\mu$ in place of $\tfrac25-\tfrac\mu{11}$ is satisfiable, i.e.\ fails, for every $\nu>\tfrac1{11}$. So the method approaches $\tfrac52$ but does not reach it.

\begin{proposition}[Dense tubes would make $\tfrac52$ sharp]\label{prop:dense}
Suppose that for $\varepsilon$ in a sequence tending to $0$ there are:
\begin{itemize}\setlength\itemsep{0pt}
\item a unitary $G\in\SU$;
\item words $W_u$ of $T$-count at most $(\tfrac52+o(1))L$, for $u$ in a set $U\subset\Z_{Q_\varepsilon}$, with $\dproj(W_u,R_z(\Delta u)G)\le\varepsilon/2$;
\item the set $U$ has maximal gap at most $Q_\varepsilon2^{-(1-o(1))L}$.
\end{itemize}
Then for $r=2$ and $Q=Q_\varepsilon$ (gaps measured cyclically in $\Z_{Q_\varepsilon}$) there are deterministic zero-error CW processes with $S\le m\big(\log_2K-(1-o(1))L\big)$ and $\Tpre\le m(\tfrac52+o(1))L$. Hence Theorem~\ref{thm:rate52} cannot be improved along that sequence.
\end{proposition}
\begin{proof}
For each coordinate pick $u\in U$ with $\rho:=a^{(1)}-u\in[0,\gamma)$, where $\gamma$ is the maximal gap. In round~1 emit $W_u$ and keep only $\rho$ in the snapshot. In round~2 emit a word within $\varepsilon/2$ of $G^{-1}R_z(\Delta(\rho+a^{(2)}))$. The product is within $\varepsilon$ of $R_z(\Delta x)$ by unitary invariance of $\dproj$ and the triangle inequality, and only $W_u$ counts towards $\Tpre$.
\end{proof}

\begin{remark}[What $\tfrac52$ is]\label{rem:what52}
In the notation of Appendix~\ref{app:elem} the tube radius is $\delta=\varepsilon R'$ (not the error probability of~\eqref{eq:rate52}), so $R'/\delta=R'^{a_0}$ with $\alpha=4/a_0$. The target count $2^{t/\alpha}$ at $t=\alpha L$ is therefore the number of $\delta$-cubes along the core, and a rate is a statement about $R'^{o(1)}$ words per $\delta$-cube. The five-point determinant of Lemma~\ref{lem:E1} forces the points of a $\delta$-cube onto a hyperplane exactly when $a_0>\tfrac85$, i.e.\ $\alpha<\tfrac52$. Beyond $\tfrac52$, most $\delta$-cubes must be shown to be empty at $T$-count $\tfrac52L$, where volume predicts $2^{L/2}$ words against $2^L$ cubes. Proposition~\ref{prop:dense} shows that some sparsity is necessary. A randomised version of its process, with a shared random offset in the snapshot, replaces the maximal gap by a length-weighted mean log-gap, so the exact condition lies between \emph{no dense tube} and \emph{sparse by a power}.

Balls behave differently from tubes, which is one reason why Conjecture~\ref{conj:H} is stated for tubes. Heuristically, near rational points $g/|g|$, $g\in\mathcal O$ with $N(\mathrm{nrd}\,g)\asymp2^{t(4/\alpha-1)/3}$, there are sphere sections $\{\mathrm{nrd}(w)=n_t,\ \langle w,g\rangle=s\}$ within distance $2^{-t/\alpha}$ of $g/|g|$ that carry $2^{t(2/3-5/(3\alpha))+o(t)}$ words each, with large constants ($126$ against about $5$ in the instance below). A bound obtained by covering the tube with balls of radius $2^{-t/\alpha}$ therefore cannot, heuristically, give more than $\tfrac52$. An exact instance: for $g=((3+2\sqrt2)+i+j-k)/2$, the $\dproj$-ball of radius $0.0025$ around $g/|g|$, whose Haar expectation is $16$ words, contains $126$ words of $T$-count $25$, all on the section $\langle w,g\rangle=(357+256\sqrt2)/2$. Heuristically a tube meets few such sections, and they are compatible with Conjecture~\ref{conj:H}.

Heuristically, auxiliary surfaces of higher degree do not help either: a rational quadric through the points of a box longer than a $\delta$-cube can be a thin tube of radius at most $\delta$ around the core, and on such a surface the second level recovers, in a model computation, exactly the degree-one box count $R'^{(8-2a_0)/3}$; we do not prove this.
\end{remark}
\section{An achievable tradeoff family}\label{sec:frontier}

\begin{theorem}[Fractional passthrough]\label{thm:frontier}
For every $q\in\{0,\dots,m\}$ and every single-rotation synthesizer of cost $\tau(\varepsilon)$ per rotation to projective error $\varepsilon$, there is a deterministic zero-error one-pass CW process with
\begin{align}
S&\le\lceil q\log_2Q\rceil+O(\log(mQr)),\notag\\
T_t&\le(m-q)\,\tau(\varepsilon/r)\quad(t<r),\notag\\
\Tout&\le\big(m+(r-1)(m-q)\big)\,\tau(\varepsilon/r).
\end{align}
\end{theorem}
\begin{proof}
Keep running sums modulo $Q$ of $q$ chosen coordinates (the block compiler of Proposition~\ref{prop:block} with $p=1$); emit every share of the remaining coordinates as it arrives, synthesized to $\varepsilon/r$ (a zero share costs nothing); in round $r$ emit the aggregate rotation of each stored coordinate and the last share of each committed one. The coordinatewise error is at most $\varepsilon$ by the right inequality of Lemma~\ref{lem:loc} and unitary invariance of $\dproj$.
\end{proof}

\begin{corollary}\label{cor:frontier}
Write $\tau(u;\eta)$ for the cost of the synthesis word selected for the angle $2\pi u/Q$ at accuracy $\eta$ and $\bar\tau_Q(\eta):=Q^{-1}\sum_{u\in\Z_Q}\tau(u;\eta)$. Under the canonical sharing every pre-final share is uniform on $\Z_Q$, so the process of Theorem~\ref{thm:frontier} has the exact expected committed magic
\begin{equation}\label{eq:achexp}
\bar\Tpre^{\mathrm{ach}}=(r-1)(m-q)\,\bar\tau_Q(\varepsilon/r),
\end{equation}
and along the family $\bar\Tpre^{\mathrm{ach}}/[(r-1)(m\log_2K-S)]\le\bar\tau_Q(\varepsilon/r)/\log_2K\cdot(1+o(1))$. If the mean single-rotation cost on the tuned grid satisfies $\bar\tau_Q(\varepsilon/r)=(3+o(1))\log_2K$---which Ross--Selinger synthesis gives for a typical angle, $\tau=3L+O(\log L)$~\cite{RossSelinger}, and which Sec.~\ref{sec:num} measures directly on the grid at $\varepsilon=10^{-10}$---the family has asymptotic slope $3$. The worst case is $4L+O(1)$, optimal to an additive constant by the lower bound $4L-9$ of~\cite[Thm.~30]{Selinger15}. The slope is within a factor $\to3$ of Theorem~\ref{thm:main1} and $\to3/2$ of Theorem~\ref{thm:rate2}; under Conjecture~\ref{conj:H} it matches Theorem~\ref{thm:main2} to leading order in $L$. The family is achievable, not proven optimal: nothing below excludes a process strictly beneath it.
\end{corollary}
The endpoint gap $\log_2Q$ versus $\log_2K$ is the difference between storing a share, which lives in $\Z_Q$, and the entropy of an aggregate, which lives in $[K]^m$ (Proposition~\ref{prop:block}).

\begin{remark}[The first bits of a coordinate are cheap]\label{rem:firstbit}
Fractional passthrough is not optimal for the first few bits, and the reason is the exact alphabet of Lemma~\ref{lem:identity}. Suppose $4\mid Q$. On any coordinate and in any round a process may commit the Clifford $S^{k}=T^{2k}$ with $k:=\mathrm{round}(4u/Q)\in\Z_4$, where $u$ is the share, at no $T$ cost, and carry in the snapshot only the residual $u-kQ/4$, an integer with at most $Q/4+1$ values: two bits per coordinate are shed for no magic at all. This is the $\log_236$ of Theorem~\ref{thm:main1}, of which $\log_224$ bits are the Cliffords. If $8\mid Q$ a third bit follows for half a $T$ gate: commit $T^{j}$, $j:=\mathrm{round}(8u/Q)\in\Z_8$, at cost $j\bmod2$, i.e.\ $\tfrac12$ on average over a uniform share, and carry $u-jQ/8$, which has at most $Q/8+1$ values. The tradeoff family therefore starts at $(2(r-1)m\text{ bits},\,0)$ and passes through $(3(r-1)m\text{ bits},\,(r-1)m/2)$, and only then continues with the passthrough slope; the three bits are exactly the allowance $\log_2(c_0+1)=3.17$ of the $c_0$ term, and by Lemma~\ref{lem:identity} the eight $T$-powers are the only words of cost at most $2L-8$ that can be used this way. The construction needs the divisibility, and in general the residual $u-jQ/g$ with $g:=\gcd(Q,8)$ and a half-open rounding takes exactly $Q/g$ values, so the saving is $\log_2g$ bits: three for $8\mid Q$, two for $Q\equiv4\pmod 8$, one for $Q\equiv2\pmod 4$, and none for odd $Q$, which is the case of the tuned $Q_\varepsilon$ at $L=5$ ($Q=25$). Rounding the residual is not free either, because the tuned grid spacing $2\pi/Q_\varepsilon\approx8\varepsilon$ is already of the order of the error budget. The tuned $Q_\varepsilon$ is in general not a multiple of $8$; the points of Fig.~\ref{fig:frontier} that use this construction are taken on the admissible grid $Q=8\lfloor Q_\varepsilon/8\rfloor$, which changes $\log_2K$ by $\log_2\frac{Q_\varepsilon-1}{Q'-1}\le\log_2\big(1+\tfrac7{Q'-1}\big)<13\cdot2^{-L}$ bits per coordinate ($0.061$ at $L=5$, $0.060$ at $L=6$). Under Conjecture~\ref{conj:H} the low-cost alphabet is larger---$N_0=O(L)$ words of cost $\le\tau_0$, the powers of a cheap infinite-order word among them (Proposition~\ref{prop:cluster})---and Theorem~\ref{thm:main2} bounds what it can buy by $\log_2N_0=O(\log L)$ bits per coordinate per round, $7$ bits at $\varepsilon=10^{-10}$. The rate statements $\alpha\ge2$ and $\alpha\to3$ describe the remaining $\log_2K-O(\log L)$ bits of each coordinate, which is where all but $O(m\log L)$ of the tradeoff lies.
\end{remark}

\section{Rate three and the quantum of commitment}\label{sec:conj}

\begin{conjecture}[H]\label{conj:H}
There are constants $c_0,c_1,c$ such that for all $G_1,G_2\in\SU$, all $\varepsilon\in(0,1/8]$, all moduli $3\le Q\le Q_\varepsilon$ (i.e.\ $\sin(\pi/2Q)\ge2\varepsilon$, the separation of the tuned family) and all $\tau\ge0$,
\begin{multline}\label{eq:H}
\#\big\{W:\ t_{\min}(W)\le\tau,\ \exists d\in\Z_Q:\\ \dproj\big(W,G_1R_z(\tfrac{2\pi d}{Q})G_2\big)\le\varepsilon\big\}\\ \le\ c_0+c_1\tau+c\,2^\tau\varepsilon^{\beta},\qquad\beta=2 .
\end{multline}
\end{conjecture}
The three terms have three different origins, and the constants are chosen for the domain stated. The exponent $\beta$ is the codimension of a rotation coset in the three-dimensional group, and Haar measure predicts a cumulative coefficient $96\varepsilon Q/\pi\le24$ for $Q\le Q_\varepsilon$ in~\eqref{eq:H}---twice the shell constant $c_{\mathrm{grid}}\le12$ of Sec.~\ref{sec:num}, since $\sum_{s\le\tau}2^s\approx2^{\tau+1}$; the restriction $Q\le Q_\varepsilon$ is part of the statement because the accuracy condition alone admits grids up to twice as fine, on which the Haar coefficient doubles. We do not fix the constants, and we do not claim that the enumeration determines them. The count on the left of~\eqref{eq:H} is a step function of $\varepsilon$, and its steps are not small: at the identity coset and at the axis cosets of cheap words the words enter in whole orbits of a symmetry of the coset, all at the same distance, so a step can be a few hundred words wide. At the onset $2^\tau\varepsilon^2\approx1$ the right side of~\eqref{eq:H} is only $c_0+c_1\tau+c$, so every such step forces $c$ above it. The largest we know is $232$ words at the identity coset with $\varepsilon=0.024532$, $Q=Q_\varepsilon=32$ and $\tau=11$, which requires $c\ge163.9$ if $(c_0,c_1)=(8,2)$ are kept (Sec.~\ref{sec:num}); dyadic accuracies fall between the steps and miss it entirely. A finite scan therefore bounds $c$ from below and never from above. Every quantity derived from~\eqref{eq:H} below is stated with $c$ as a parameter; the numbers of Table~\ref{tab:numbers} are given for $c\in\{48,256\}$ as a sensitivity check only, and they move by less than $3\%$ over that range because $c$ enters solely through $\log_2c$. The constant $c_0\le8$ counts the words lying exactly on an arithmetic coset, which are the powers of $T$ in the $(G_1,G_2)$ frame (Lemma~\ref{lem:arith}). The linear term is forced: a coset can contain, or lie within $\varepsilon$ of, every power of an infinite-order word (Proposition~\ref{prop:cluster}), and the cheapest such word, $HT$, has $t_{\min}((HT)^{j})=j$ because $(HT)^j$ is its own Matsumoto--Amano normal form, so $c_1\ge2$ from the powers $j$ and $-j$; we take $c_1=2$. The words on or near a rotation coset thus consist of a volume-law population and a single sparse cluster, and the cluster is what a uniform statement must pay for. By Remark~\ref{rem:barrier} the conjecture lies beyond (R): in the parametrization of Appendix~\ref{app:lemmas} it asserts that among the $\xi=tt^\dagger\in\Z[\sqrt2]$ in a thin strip, those with $2^k-\xi$ also a relative norm have density $2^{-k+o(k)}$---a shifted-norm sieve at the ramified prime $(1-\omega)$.

\begin{lemma}[Words on an arithmetic coset]\label{lem:arith}
For $G_1,G_2\in\Gamma$ the Clifford+$T$ unitaries lying exactly on $\{G_1R_z(\theta)G_2\}$ are the eight words $G_1T^jG_2$.
\end{lemma}
\begin{proof}
$R_z(\theta)=G_1^{-1}WG_2^{-1}\in\Gamma$ is diagonal; in the parametrization of Lemma~\ref{lem:identity} this is the case $t=0$, whence $u=\omega^a(1-\omega)^{2k}$ up to a unit and the unitary is $\mathrm{diag}(\omega^a,\omega^b)$ modulo phase.
\end{proof}

\begin{proposition}[The linear term cannot be removed]\label{prop:cluster}
No constants $c_0,c$ satisfy~\eqref{eq:H} with $c_1=0$, even when $Q=Q_\varepsilon$ is tuned and even when $G_1,G_2\in\Gamma$.
\end{proposition}
\begin{proof}
Let $V:=e^{-i\pi/4}HTHT$, of $T$-count $2$, with rotation angle $\theta$, $\cos(\theta/2)=(2+\sqrt2)/4$. Since $2\cos(\theta/2)$ is a root of $2y^2-4y+1$ and not an algebraic integer, $V$ has infinite order, and $I,V,\dots,V^N$ are distinct with $t_{\min}(V^j)\le2j$. Write $V=DR_z(\theta)D^\dagger$. Let $x:=\theta/2\pi$ and let $p/Q$ run over the convergents of $x$, $|x-p/Q|<Q^{-2}$; put $\varepsilon_Q:=\tfrac12\sin(\pi/2Q)$, so that $Q_{\varepsilon_Q}=Q$, and $d_j:=jp\bmod Q$. Then $\dproj(V^j,DR_z(2\pi d_j/Q)D^\dagger)\le\tfrac j2|\theta-2\pi p/Q|<\pi N/Q^2$, which is below $\varepsilon_Q/2\ge1/(4Q)$ once $Q>4\pi N$. For $G_1,G_2\in\Gamma$: $\Gamma$ is dense in $\PU$, so there is $G_1\in\Gamma$ with $D=G_1E$, $\dproj(E,I)\le\varepsilon_Q/4$; then $\dproj(V^j,G_1R_z(j\theta)G_1^{-1})\le2\dproj(E,I)\le\varepsilon_Q/2$, and with $G_2:=G_1^{-1}$ all $N+1$ powers are within $\varepsilon_Q$ of a grid rotation on the arithmetic coset $G_1R_zG_1^{-1}$. With $\tau=2N$ the left side of~\eqref{eq:H} is $\ge N+1$ while, along the convergents, $c2^{2N}\varepsilon_Q^2\to0$; so $N+1\le c_0$ for every $N$, which is absurd. (At $p/Q=1746/10009$, $N=9$ and the frame $D$ itself, a direct check gives ten words against $8+12\cdot2^{18}\varepsilon_Q^2=8.02$.)
\end{proof}

Conjecture~\ref{conj:H} is the quantitative form of an observation of Selinger~\cite[\S9]{Selinger15}: the elementary count $36\cdot2^\tau\ge c\,\varepsilon^{-3}$, which gives $3L-O(1)$ for a generic element of the three-dimensional group, ``is not a priori clear'' for $z$-rotations, since these form a one-dimensional submanifold; \cite[Thm.~30]{Selinger15} then proves the worst-case bound $4L-9$ for a specific angle. What Conjecture~\ref{conj:H} adds is a bound for \emph{every} coset and \emph{every} cost, with the constants, which is what a counting argument against a stream needs. It is also falsifiable in a useful direction. The $\varepsilon$-balls of distinct grid rotations are disjoint, so \eqref{eq:H} bounds the number of grid rotations reachable at cost $\tau$ by $c_0+c_1\tau+c2^\tau\varepsilon^2$; covering the $K_\varepsilon\approx\pi/4\varepsilon$ rotations of the tuned family then forces
\begin{equation}\label{eq:Hcons}
\max_d\tau_{\min}(d)\ \ge\ 3L-\log_2(4c/\pi)-o(1)\ =\ 3L-5.93-o(1),
\end{equation}
so Conjecture~\ref{conj:H} cannot coexist with an ancilla-free synthesis of grid rotations below the Ross--Selinger exponent. The numerical constant is shown for the illustrative $c=48$; the statement itself is $3L-\log_2(4c/\pi)-o(1)$ for whatever $c$ the conjecture carries. The enumeration gives $\max_d\tau_{\min}=10,16,18$ against $3L=12,15,18$ at $L=4,5,6$.

\begin{theorem}[Rate three; quantization]\label{thm:main2}
Assume Conjecture~\ref{conj:H}, $Q\le Q_\varepsilon$ and $\varepsilon\le(2c)^{-1/2}$. Put $k:=\log_2K$, $b:=\beta L-\log_2c-1$, $\tau_0:=\lfloor\beta L-\log_2c\rfloor$, $N_0:=c_0+c_1\tau_0+1$, $a_0:=1+\log_2N_0$, $\kappa:=1+b/k$ and $D_t:=(1-2\delta)mk-2mh_2(\delta)-S$. For every round $t<r$,
\begin{equation}\label{eq:rate3}
\bar T_t\ \ge\ \kappa\,(D_t-ma_0)-3m\log_2(1+\bar T_t/m),
\end{equation}
and, writing $\bar{\mathcal A}^{>}_t$ for the expected number of coordinates that are correct and whose word committed in round $t$ has minimal $T$-count greater than $\tau_0$,
\begin{equation}\label{eq:active}
\bar{\mathcal A}^{>}_t\ \ge\ \frac{D_t-ma_0}{k}.
\end{equation}
The accuracy restriction is what makes $b\ge0$, equivalently $\tau_0\ge1$, which Lemma~\ref{lem:cost} requires; and it cannot be dropped: $b<0$ gives $\kappa<0$ and a right-hand side of~\eqref{eq:rate3} that grows with $S$, which the process that stores every share and commits nothing before the last round violates. At $\varepsilon=10^{-10}$ it is satisfied with room to spare for every $c$ below $5\times10^{19}$; it bites only at coarse accuracy, where~\eqref{eq:rate3} was vacuous anyway. With $\beta=2$, $c_1=2$ and $c_0=8$, $\kappa\to1+\beta=3$ for every fixed $c$, so $\bar\Tpre\ge3(r-1)(m\log_2K-S)(1-o(1))-O(rm\log L)$; a committed word of cost at most $\tau_0=\lfloor2L-\log_2c\rfloor$ carries at most $\log_2N_0=\log_2(2\tau_0+9)=O(\log L)$ bits about its coordinate, and every bit beyond that allowance is carried by a word of cost at least $\tau_0+1$, i.e.\ $2L-\log_2c-O(1)$ gates. At $\varepsilon=10^{-10}$ the threshold is $\tau_0=60$ and the allowance $\log_2N_0=7.0$ bits for $c=48$, and $\tau_0=58$ with $7.0$ bits for $c=256$, against the $32.9$ bits a coordinate carries.
\end{theorem}
The proof is in Appendix~\ref{app:rate}. The structural content is that memory cannot be shed a few bits at a time within a coordinate. The low-cost alphabet of cost $\le\tau_0\approx2L$ holds at most $N_0=O(L)$ words---the eight $T$-powers of Remark~\ref{rem:firstbit}, which are real, and the powers of a cheap infinite-order word, which Proposition~\ref{prop:cluster} shows must be allowed for: a process may align its prefix with the axis of $HT$ and then encode an integer $j$ by emitting $(HT)^j$, which is $\log_2j$ bits for $j$ gates---so it buys $O(\log L)$ bits, and a coordinate that carries more must commit a word of cost $\ge2L-O(1)$; \eqref{eq:active} counts how many coordinates per round must do so. The theorem does not say that every coordinate receiving a $T$ gate has paid $2L$ gates---a process may mix cheap words with rare expensive ones---but the expected cost of every bit beyond the $O(\log L)$ allowance is $\kappa$ gates, and the optimal way to pay it is to commit whole coordinates, which is what fractional passthrough does. The convergence $\kappa\to3$ is slow: at $\varepsilon=10^{-10}$, $\kappa=2.82$ for $c=48$ and $2.75$ for $c=256$, rising only to $2.93$ and $2.90$ at $\varepsilon=2^{-80}$, so the number that a resource estimate at a chemically relevant accuracy should use is $\kappa$, not $3$---and the choice of $c$ matters far less than the finite accuracy does.

\begin{lemma}[Identity coset]\label{lem:identity}
If $2\varepsilon^24^k<1$, i.e.\ $k<L-\tfrac12$, the only Clifford+$T$ unitaries with denominator exponent $\le k$ within $\varepsilon$ of some $R_z(\theta)$ are the eight powers of $T$.
\end{lemma}
Lemma~\ref{lem:identity} (Appendix~\ref{app:lemmas}) is an elementary instance of the fee. A word of minimal $T$-count $\tau$ has denominator exponent $k\le\tfrac\tau2+3$ (Appendix~\ref{app:clifford}), so at the identity coset every word of $T$-count $\le2L-8$ within $\varepsilon$ of a $z$-rotation is a power of $T$: the fee is $2L-O(1)$, unconditionally. The enumeration puts it at $2L+O(1)$: the least $T$-count $\ge2$ at which some word comes within $\varepsilon$ of a grid rotation is $10,10,12,17,18$ for $L=4,\dots,8$, i.e.\ between $2L$ and $2L+3$ (within $\varepsilon$ of some $z$-rotation, not necessarily a grid one: $7,10,11$ at $L=4,5,6$).

Beyond the fee, the whole of Conjecture~\ref{conj:H} holds at the cosets whose frame is Clifford, up to a subexponential factor. The reason is that there the off-diagonal entry of a word is an arithmetic coordinate, and the divisor bound controls the rest.

\begin{theorem}[Volume law at Clifford cosets]\label{thm:clifford}
There is $\eta(\tau)=O(\tau/\log\tau)$ such that for all Cliffords $C_1,C_2$, all $\varepsilon\in(0,1/8]$ and all $\tau\ge1$, with $C=C_1R_zC_2$,
\begin{equation}\label{eq:clifford}
\#\{W:\ t_{\min}(W)\le\tau,\ W\in T_\varepsilon(C)\}\le2^{\eta(\tau)}\big(\tau+1+\varepsilon^22^\tau\big).
\end{equation}
For $G_1,G_2\in\Gamma$ with $t_{\min}(G_1)+t_{\min}(G_2)\le h$ the same holds with $\tau$ replaced by $\tau+h$.
\end{theorem}
The proof is in Appendix~\ref{app:clifford}. Up to $2^{\eta(\tau)}$, \eqref{eq:clifford} is the tube form of~\eqref{eq:H} with $\beta=2$, all constants absorbed in $2^\eta$, and it holds for every $\tau$, not only above the onset. Two consequences follow at once.

\begin{corollary}[Typical $z$-rotations need $3L$]\label{cor:typical}
For every $\varepsilon\in(0,1/8]$ and every $s\ge0$, the set of $\theta\in[0,4\pi)$ for which $R_z(\theta)$ has an $\varepsilon$-approximation of $T$-count $\le3L-s$ has measure $\le C\,2^{\eta(3L)}\big(L2^{-L}+2^{-s}\big)$. The same bound holds for the fraction of grid angles $2\pi d/Q$, $d\in\Z_Q$, for the tuned modulus $Q=Q_\varepsilon$ and for any modulus $Q\ge c'/\varepsilon$ with $\sin(\pi/2Q)>\varepsilon$.
\end{corollary}
\begin{proof}
A word within $\varepsilon$ of the torus is within $\varepsilon$ of $R_z(\theta)$ only for $\theta$ in a set of measure $\le C\varepsilon$, so the measure is $\le C\varepsilon$ times~\eqref{eq:clifford} at $\tau=3L-s$. For the grid, distinct grid rotations have disjoint $\varepsilon$-balls, so the number of grid angles reached is at most~\eqref{eq:clifford}, and there are $Q\ge c'/\varepsilon$ of them.
\end{proof}
In words: all but a vanishing fraction of $z$-rotations need $T$-count $3L-O(L/\log L)$. By Borel--Cantelli along $\varepsilon=2^{-j}$ this holds, for almost every $\theta$, for all small $\varepsilon$. That is the lower half of~\cite[Conj.~8.10]{RossSelinger} for almost every angle, whereas the conjecture is stated for every angle with $\tan(\theta/2)\notin\Q(\sqrt2)$. The analogue for Haar-almost every element of $\SU$ is proved in~\cite{SMA}. This is the typical-case lower bound that Ross and Selinger~\cite{RossSelinger} call information-theoretic and support heuristically, and that Selinger~\cite[\S9]{Selinger15} notes is ``not a priori clear'' for $z$-rotations. It also makes~\eqref{eq:Hcons} unconditional up to $O(L/\log L)$: the passthrough family cannot be improved by a better single-rotation synthesizer, only by a better way of committing.

The second consequence concerns the processes that pipelines actually run. Call a CW process \emph{$\vartheta$-rotation-segmented} if, on every input and tape, for every coordinate $j$ and every round $t<r$, the unitary $\overline W^{(t)}_j$ of the segment it emits on $j$ in round $t$ is within $\vartheta$ of $C\,R_z(\varphi)\,C'$ for some angle $\varphi$ and some Cliffords $C,C'$. Per-rotation synthesis of each share is of this form, and so is the fractional-passthrough family of Theorem~\ref{thm:frontier}, with $\vartheta=\varepsilon/r$. The constraint bears on the committed word itself, so the frames of Corollary~\ref{cor:profile}, which are the source of the gap between Theorems~\ref{thm:rate52} and~\ref{thm:main2}, play no role.

\begin{theorem}[Rate three for rotation-segmented processes]\label{thm:rate3seg}
For every $\vartheta$-rotation-segmented process as in Sec.~\ref{sec:setting} with $\vartheta\le1/8$, and every round $t<r$,
\begin{multline}\label{eq:rate3seg}
\bar T_t\ \ge\ \alpha_\vartheta\big[(1-2\delta)m\log_2K-2m\,h_2(\delta)\\-S-mA_\vartheta\big],\qquad\alpha_\vartheta:=1+\frac{2\log_2(1/\vartheta)}{\log_2K},
\end{multline}
with $A_\vartheta=O(\Lambda/\log\Lambda)$, $\Lambda:=\alpha_\vartheta\log_2K$. In particular, a $\vartheta$-rotation-segmented process with $\vartheta\le\varepsilon$ is $\varepsilon$-rotation-segmented, and $\alpha_\varepsilon\ge3-O(1/L)$, so $\bar\Tpre\ge3(r-1)\big(m\log_2K-S-O(mL/\log L)\big)-O(\delta rmL)$.
\end{theorem}
\begin{proof}
Follow the proof of Theorem~\ref{thm:rate2}. The alphabet of $\overline W^{(t)}_j$ lies in the union of the $576$ tubes of radius $\vartheta$ around the Clifford-framed tori, so by Theorem~\ref{thm:clifford} its profile is $n(\tau)\le576\cdot2^{\eta(\tau)}(\tau+1+\vartheta^22^\tau)$, whatever the context. Apply Lemma~\ref{lem:gibbs} with $\lambda=1/\alpha_\vartheta$ and $\tau_*=\lceil\alpha_\vartheta\log_2K\rceil$. For $\tau\le\alpha_\vartheta\log_2K$ one has $\tau-2\log_2(1/\vartheta)\le\tau/\alpha_\vartheta$, which is the definition of $\alpha_\vartheta$. So $n(\tau)2^{-\lambda\tau}\le576\cdot2^{\eta(\tau)+1}(\tau+1)$, and $Z_\lambda\le C(\tau_*+1)^22^{\max\eta}$. Hence $A_\vartheta:=1+\log_2Z_\lambda=O(\Lambda/\log\Lambda)$, and every step that produced the factor $2$ now produces $\alpha_\vartheta$. Finally $\alpha_\varepsilon\ge1+2L/(L+\log_2\frac\pi2)$, because $\log_2K<L+\log_2\frac\pi2$. If $D:=(1-2\delta)m\log_2K-2mh_2(\delta)-S-mA_\varepsilon\ge0$, then $\alpha_\varepsilon D\ge3D-O(m)$; otherwise the claim is trivial.
\end{proof}
Theorem~\ref{thm:rate3seg} is unconditional and matches the passthrough family in its rate. For shares committed at accuracy $\varepsilon/r$ the lower bound is $1+2\log_2(r/\varepsilon)/\log_2K$, and passthrough achieves $\bar\tau_Q(\varepsilon/r)/\log_2K$, which is $3+O(\log r/L)$ under the synthesis hypothesis of Corollary~\ref{cor:frontier}. In the model where each committed piece is within $\varepsilon$ of a Clifford-framed $z$-rotation, the exchange rate is therefore $3$ to leading order in $L$: unconditionally from below, and from above under the synthesis hypothesis. The restriction builds in the entry fee: by Lemma~\ref{lem:identity} and $k\le\tfrac\tau2+3$, a segment of cost $\le2L-8$ is exactly some $CT^jC'$, one of at most $24^2\cdot8$ unitaries. Theorems~\ref{thm:clifford} and~\ref{thm:rate3seg} and Corollary~\ref{cor:typical} are asymptotic. Their $\eta$ is the worst-case divisor bound, and at $\varepsilon=10^{-10}$ the largest divisor count of an ideal of $\Z[\sqrt2]$ of norm $\le16\cdot2^{99}$ is $2^{22.8}$; with the constants of the proof $A_\varepsilon>\log_2K$, so~\eqref{eq:rate3seg} is vacuous there. Beating it requires committing words that are \emph{not} near any Clifford-framed rotation and whose frames, created by the prefix and suffix, carry the rest of the information. Theorem~\ref{thm:rate52} bounds how much such frames can help, and Conjecture~\ref{conj:H} asserts that they cannot help at all.

\section{Side information}\label{sec:side}

\begin{theorem}[Side information]\label{thm:side}
Let $X\sim P_X$ on $[K]^m$, let the sharing have any joint law with $X$, and let side information $Y$ be available to the process from the start. Then for every round $t$,
\begin{multline}
S+\bar T_t+m\log_236+\rho_m(\bar T_t)\\
\ge\ H\big(X\mid\mathrm{Past},\mathrm{Fut},Y\big)-h_2(\delta)-\delta\log_2(K^m-1).
\end{multline}
\end{theorem}
\begin{proof}
Theorem~\ref{thm:main1} with $Z:=(\mathrm{Past},\mathrm{Fut},Y,\rho)$; the tape is independent of $(\text{stream},Y)$, and $a^{(t)}\leftrightarrow X$ is a bijection given the other rounds.
\end{proof}
\begin{corollary}\label{cor:side}
(a) Uniform $X$, canonical sharing, trivial $Y$ recovers Theorem~\ref{thm:main1}.
(b) If $Y$ determines the earlier shares (the process holds the sharing or twirl seed), the right side is $0$: no memory and no committed magic. This is the boundary of the design rule ``randomize and lower late'' of~\cite{Paper1}.
(c) Suppose the shares of rounds $1,\dots,r-1$ are a deterministic injective function of a $\lambda$-bit uniform seed $V$, independent of $X$ and unknown to the process, and $Y$ is trivial. Given the other rounds, $A^{(t)}$ is a function of $V$, so
\[H\big(X\mid\mathrm{Past},\mathrm{Fut}\big)=H\big(A^{(t)}\mid\mathrm{Past},\mathrm{Fut}\big)\le H(V)=\lambda\]
and the bound of Theorem~\ref{thm:side} degrades from $m\log_2K$ to $\lambda$: this is the direction that matters in practice, since in randomized compiling $\lambda$ is a handful of bits, not $m\log_2K$. The matching lower bound needs two rounds and a uniform aggregate, and then holds: for $r=2$ and $X$ uniform on $[K]^m$, $H(a^{(2)})\le m\log_2Q$ and $H(a^{(2)}\mid A^{(1)})=H(X)=m\log_2K$, so $I(A^{(1)};a^{(2)})\le m\log_2(Q/K)$ and
\[H\big(A^{(1)}\mid a^{(2)}\big)\ \ge\ \lambda-m\log_2(Q/K),\]
where $m\log_2(Q/K)=O(m/Q)$ is negligible for the tuned family. This is the precise form of~\cite[Rem.~9]{Paper1}: the lower bound is governed by the entropy the process is missing, not by the size of the aggregate.

Both restrictions are needed. With $r=3$, take $a^{(1)}:=V$, $a^{(2)}:=-V$ and $a^{(3)}:=X$, where $V\in\Z_Q^m$ is the injective image of the $\lambda$-bit seed. The first two rounds are still an injective function of the seed, but for $t=1$ the conditioning already contains $a^{(3)}=X$, so $H(X\mid\mathrm{Past},\mathrm{Fut})=0$ and nothing is charged. A seed that is spent and then cancelled in a later round buys no lower bound, whatever its entropy; only the upper bound $\lambda$ survives for general $r$.
(d) For non-product sources Theorems~\ref{thm:rate2} and~\ref{thm:main2} hold with $m\ell_\delta$ replaced by $\sum_jH(A_j\mid Z)-m(h_2(\delta)+\delta\log_2K)$ and $S$ by $S+\mathrm{TC}(A\mid Z)$, since $\sum_jI(A_j;\sigma\mid Z)\le I(A;\sigma\mid Z)+\mathrm{TC}(A\mid Z)$.
\end{corollary}

\begin{theorem}[One-shot]\label{thm:oneshot}
Let the process be deterministic with worst-case resources $S_{\max}$, $T_{\max}$, correct with probability $\ge1-\delta$ over $(X,\text{sharing},Y)$, and $r=2$. With the set-smoothed entropy $H_0^\delta(U\mid VY):=\min_{\Omega:\Pr[\Omega]\ge1-\delta}\log_2\max_{v,y}|\{u:(u,v,y)\in\Omega\}|$ of Renner--Wolf~\cite{RennerWolf},
\begin{multline}
S_{\max}+T_{\max}+m\log_272+m\log_2\!\big(e(1+T_{\max}/m)\big)\\
\ge\ H_0^\delta(U\mid V,Y).
\end{multline}
\end{theorem}
The Shannon form (Theorem~\ref{thm:side}) applies to expected resources and the one-shot form to worst-case resources; the two are not to be mixed. We state the bound in terms of $H_0^\delta$ only. It can be relaxed further to a smooth max-entropy, but not at the same smoothing parameter: truncating to an event of probability $1-\delta$ and renormalizing moves the distribution by $\sqrt\delta$ and not $\delta$ in purified distance, so the resulting statement is a bound by $H^{\sqrt\delta}_{\max}$, with a second-order term that is negative rather than positive~\cite{TomamichelHayashi}. Since nothing below uses it, we do not carry that version.

\section{Numerical evidence}\label{sec:num}

We enumerated all single-qubit Clifford+$T$ unitaries of $T$-count $\le22$ ($3.0\times10^8$ words, of which $1.5\times10^8$ lie in the top shell $t=22$) through the Matsumoto--Amano normal form, with the $24$ right Cliffords verified by closure.

\paragraph{Conjecture~\ref{conj:H} (Fig.~\ref{fig:tube}).} For each $t\le22$, $\varepsilon\in\{2^{-4},\dots,2^{-8}\}$ with tuned $Q_\varepsilon\in\{12,25,50,100,201\}$, and three coset pairs (identity and two Haar-random $(G_1,G_2)$), we counted words within $\varepsilon$ of some grid rotation $G_1R_z(2\pi d/Q)G_2$ and words within $\varepsilon$ of the whole coset. A least-squares fit over the $136$ cells with count $\ge30$ gives
\begin{equation}
\log_2N=1.001\,t-2.007\,L+\text{const},
\end{equation}
i.e.\ exponents $(1,2)$ to three decimals. The constants are not fitted. The grid constant is $c_{\mathrm{grid}}=36\cdot\tfrac13\cdot8\varepsilon/(2\pi/Q_\varepsilon)=48\varepsilon Q_\varepsilon/\pi$, which equals $12$ only up to the rounding of $Q_\varepsilon$ and is $11.46,11.94,11.94,11.94,12.00$ for $\varepsilon=2^{-4},\dots,2^{-8}$; the tube constant is $36$ for every $\varepsilon$. At $t=22$, pooling the two Haar pairs, $N/(2^t\varepsilon^2)=11.47,11.96,11.99,12.16,12.14$ for the grid and $35.98,35.98,36.02,36.20,36.43$ for the tube, agreeing column by column to $0.1,0.2,0.5,1.8,1.2$ per cent and to $0.1,0.1,0.1,0.6,1.2$ per cent respectively. Over the whole table, restricted to the $30$ Haar-pair cells whose \emph{expected} count exceeds $3000$, $N$ divided by its exact prediction lies in $[0.98,1.04]$ for the grid and $[0.98,1.01]$ for the tube; the cut must be on the expected and not the observed count, since selecting cells by observed count conditions on the upper Poisson tail. Using $12$ in place of $c_{\mathrm{grid}}$ instead leaves the $\varepsilon=2^{-4}$ column low by a systematic $4.5\%$ at every $t$ --- at $t=22$ that is $20$ standard deviations in the single column of Table~\ref{tab:counts}, and manifestly not Poisson noise. What the three extra $T$-counts buy is the small-$\varepsilon$ end of the table, where the expected counts were previously of order one: at $t=19$ the $\varepsilon=2^{-8}$ cell held $77$ words against a prediction of $96$ ($20\%$ low, and Poisson-noisy), whereas at $t=22$ it holds $786$ against $768$ ($2\%$ high). Under Haar measure the squared off-diagonal entry $|b|^2$ is uniform on $[0,1]$, so the tube $\{|b|^2\le\varepsilon^2\}$ has measure $\varepsilon^2$ and $36\cdot2^t\varepsilon^2$ of the $36\cdot2^t$ words of $T$-count $t$ are expected in it; the grid acceptance region $\{s+\varphi^2/4\le\varepsilon^2\}$ in the (transverse, angular) coordinates occupies one third of the tube per angular window of width $8\varepsilon$, and the windows are spaced by $2\pi/Q_\varepsilon$, giving $c_{\mathrm{grid}}2^t\varepsilon^2$ with $c_{\mathrm{grid}}=48\varepsilon Q_\varepsilon/\pi$ as above. The words are therefore Haar-equidistributed near rotation cosets to within statistical error wherever the expected count is at least of order one. At the identity coset the count is exactly $0$ for $2\le t\le9$ at $L=4$, with onset at $t=10=2L+2$: that coset is \emph{below} Haar at low cost, which is the fee.

\begin{figure*}[t]
\centering\includegraphics[width=0.92\textwidth]{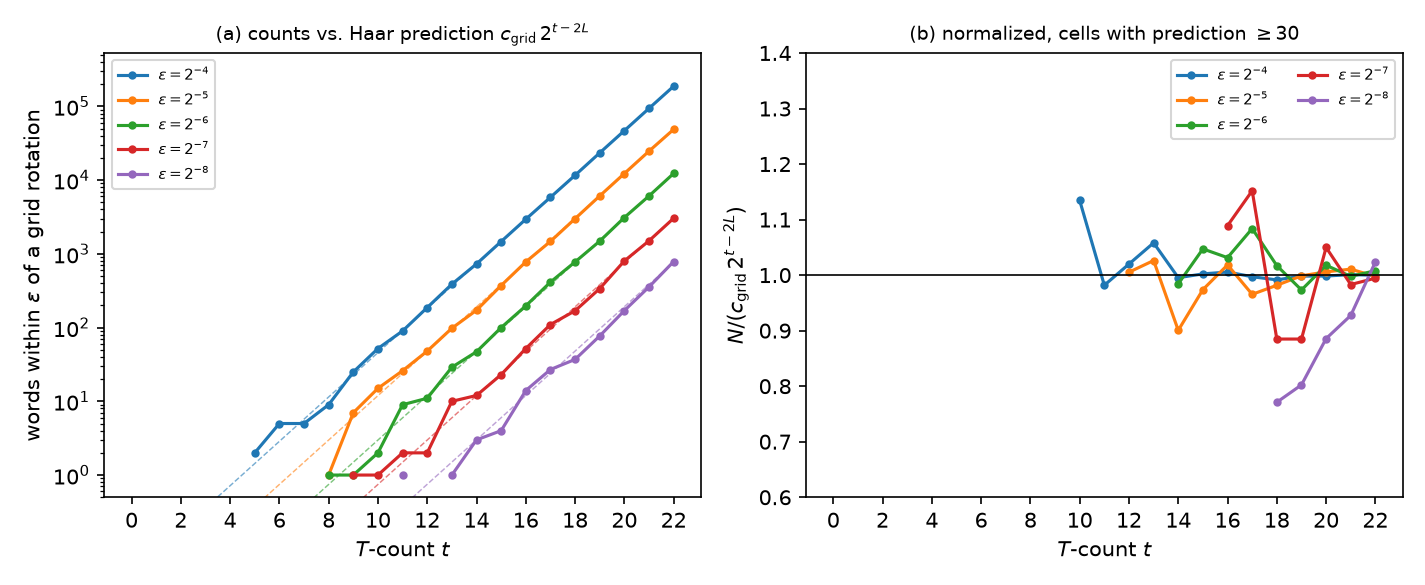}
\caption{Number of Clifford+$T$ words of $T$-count exactly $t$ within $\varepsilon$ of a grid rotation on a Haar-random coset $G_1R_z(\theta)G_2$. (a) Counts (markers) against the exact Haar prediction $c_{\mathrm{grid}}2^{t-2L}$ with $c_{\mathrm{grid}}=48\varepsilon Q_\varepsilon/\pi$ (dashed). (b) The same counts divided by that prediction, plotted where the prediction is at least $30$; every $\varepsilon$ has reached the Haar value by $t=22$. Exhaustive enumeration to $t=22$ ($3.0\times10^8$ words, $1.5\times10^8$ in the top shell).}
\label{fig:tube}
\end{figure*}

\begin{table}[t]
\centering\scriptsize
\caption{Grid counts on Haar coset pair 0; exact Haar prediction $c_{\mathrm{grid}}2^{t-2L}$, $c_{\mathrm{grid}}=48\varepsilon Q_\varepsilon/\pi$, in parentheses.}
\label{tab:counts}
\setlength{\tabcolsep}{1pt}\begin{tabular}{@{}c ccccc@{}}
\toprule
$t$ & $\varepsilon=2^{-4}$ & $2^{-5}$ & $2^{-6}$ & $2^{-7}$ & $2^{-8}$\\
\midrule
12 & 187 (183) & 48 (48) & 11 (12) & 2 (3.0) & 0 (0.7)\\
15 & 1470 (1467) & 372 (382) & 100 (95) & 23 (24) & 4 (6.0)\\
19 & 23426 (23468) & 6104 (6112) & 1487 (1528) & 338 (382) & 77 (96)\\
22 & 187707 (187747) & 48772 (48892) & 12321 (12223) & 3039 (3056) & 786 (768)\\
\bottomrule
\end{tabular}
\end{table}

\paragraph{Arithmetic cosets.} The cosets that Corollary~\ref{cor:profile} hands to Theorems~\ref{thm:rate2} and~\ref{thm:main2} are not Haar-random: $G_1$ and $G_2$ are themselves Clifford+$T$ words, and Conjecture~\ref{conj:H} is asserted uniformly in $(G_1,G_2)$, such frames included. We therefore repeated the whole count on four arithmetic pairs---random Matsumoto--Amano words of $T$-count $5$, $10$ and $15$, and a pair of Cliffords ($T$-count $0$)---with the three pairs of the previous paragraph recomputed as a regression ($345$ cells, all identical). The exponents do not move: a fit over the $180$ arithmetic cells with count $\ge30$ gives $\log_2N=1.006\,t-2.003\,L$ for the grid count and $1.002\,t-2.009\,L$ for the tube, against $1.001\,t-2.000\,L$ and $0.997\,t-2.003\,L$ on the Haar pairs. Nor does the constant, for pairs of $T$-count $10$ and $15$: over all cells with expected count $\ge100$ the ratio $N/(c_{\mathrm{grid}}2^{t-2L})$ lies in $[0.87,1.07]$, which is the spread of the Haar pairs themselves ($[0.87,1.08]$). Cosets close to the identity behave differently, and exactly as they must. A Clifford pair reproduces the identity-coset counts \emph{cell by cell}: for $G_1,G_2$ Clifford, $W\mapsto G_1^\dagger WG_2^\dagger$ is a $T$-count-preserving bijection of $\Lambda_k$ carrying the grid onto itself, so Clifford cosets are isometric copies of the identity coset and inherit its fee. There the ratio ranges over $[0.34,1.75]$---suppressed below the onset $t\approx2L$, overshooting for a few $T$-counts above it---and the $T$-count-$5$ pair interpolates, $[0.73,1.24]$. All $70$ cells at $t=22$ lie in $[0.86,1.04]$. Uniformity in $(G_1,G_2)$ thus costs at least a factor of two in the constant $c$; since $\kappa$ depends on $c$ only through $\log_2c$, the cost of that factor is small---$\kappa=2.82$ against the $2.85$ of the cumulative Haar value $24$ at $\varepsilon=10^{-10}$---which is what makes it affordable to leave $c$ unspecified. The additive terms are a separate matter: an arithmetic coset contains the exact word $G_1G_2$ at $\theta=0$, so it has grid hits at $T$-count $\le t_1+t_2$ no matter how small $\varepsilon$ is, and a frame whose axis is aligned with that of an infinite-order word carries the powers of that word (Proposition~\ref{prop:cluster}). Table~\ref{tab:lowcost} counts, on each frame and each $\varepsilon$, the words of cost $\le\tau_{24}=\lfloor2L-\log_224\rfloor$ within $\varepsilon$ of a grid rotation, against the $c_0+c_1\tau_{24}$ that Conjecture~\ref{conj:H} allows. On the identity and Clifford frames the count is the number of $T$-powers that fall within $\varepsilon$ of the grid; on the arithmetic and Haar-random frames it is at most three. The last row is a frame built for Proposition~\ref{prop:cluster}: $G_1\in\Gamma$ of $T$-count $16$ is the word of cost $\le16$ whose image $G_1\hat z$ of the $z$-axis is closest to the rotation axis of $HT$ (angle $1.0\times10^{-2}$), and $G_2=G_1^{-1}$, so that the powers $(HT)^{\pm j}$ lie within about $10^{-2}$ of the coset. Where that alignment is finer than $\varepsilon$ ($L\le6$) the cluster is visible---seven words at $L=4$, i.e.\ all of $(HT)^{\pm j}$ with $j\le\tau_{24}=3$, and seven of the fifteen candidates at $L=6$, about the fraction of the coset that the grid balls cover at the tuned modulus; where it is coarser ($L\ge7$) the frame behaves like an arithmetic one. Every entry is below $c_0+c_1\tau_{24}$, and the linear term is a worst case over frames rather than the population of a generic one. The shell counts of the previous paragraphs come from \texttt{scripts/enum\_arith.py} and \texttt{scripts/arith\_analysis.py}; Table~\ref{tab:lowcost}, the exact-optimal costs $\tau_{\min}(d)$ below, and Fig.~\ref{fig:frontier} from \texttt{lowcost\_counts.py}, \texttt{thm\_numbers.py} and \texttt{make\_fig2.py}.
\begin{table}[t]
\centering\footnotesize
\caption{Words of cost $\le\tau_{24}:=\lfloor2L-\log_224\rfloor$ within $\varepsilon$ of a grid rotation of the tuned family, on the frames indicated, against the allowance $c_0+c_1\tau_{24}$ of Conjecture~\ref{conj:H} with $(c_0,c_1)=(8,2)$. The threshold $\tau_{24}$ is one unit above the $\tau_0$ of Theorem~\ref{thm:main2} at $c=48$ and equal to it at $c=24$, so the table is a more demanding check than Theorem~\ref{thm:main2} requires. On the axis-aligned frame the powers $(HT)^{\pm j}$, $j\le\tau_{24}$, lie within $10^{-2}$ of the coset, which is below $\varepsilon$ at $L\le6$ and above it at $L\ge7$. Computed by \texttt{scripts/lowcost\_counts.py}.}
\label{tab:lowcost}
\setlength{\tabcolsep}{4pt}\begin{tabular}{@{}l ccccc@{}}
\toprule
$L$ & $4$ & $5$ & $6$ & $7$ & $8$\\
$\tau_{24}$ & $3$ & $5$ & $7$ & $9$ & $11$\\
\midrule
identity & $4$ & $3$ & $2$ & $4$ & $3$\\
Clifford pair & $4$ & $3$ & $2$ & $4$ & $3$\\
arithmetic, $T$-count $5$ & $0$ & $0$ & $0$ & $0$ & $3$\\
arithmetic, $T$-count $10$ & $1$ & $1$ & $0$ & $1$ & $1$\\
arithmetic, $T$-count $15$ & $3$ & $0$ & $1$ & $1$ & $0$\\
Haar-random pair A & $0$ & $0$ & $0$ & $0$ & $1$\\
Haar-random pair B & $1$ & $1$ & $1$ & $0$ & $0$\\
axis-aligned (Prop.~\ref{prop:cluster}) & $7$ & $3$ & $7$ & $3$ & $3$\\
\midrule
$c_0+c_1\tau_{24}$ & $14$ & $18$ & $22$ & $26$ & $30$\\
\bottomrule
\end{tabular}
\end{table}

\paragraph{Axis cosets and the size of $c$.} The frames on which~\eqref{eq:H} is tightest are the cosets $DR_zD^\dagger$ of the rotation axis of a cheap word $V$, with $D$ the eigenbasis of $V$. On the axis of $HT$ at $\varepsilon=2^{-5}$, $Q=25$, the cumulative grid count at $\tau=11$ is $137$---shell counts $1,0,0,2,0,0,2,4,8,16,32,72$---against the allowance $8+22+48\cdot2^{11}\varepsilon^2=126$ that $c=48$ would give, so a uniform constant must satisfy $c\ge53.5$ if $(c_0,c_1)=(8,2)$ are kept. The mechanism is visible in the enumeration. Since $V$ lies on the coset and $\dproj$ is bi-invariant, left multiplication by $V$ maps the $\varepsilon$-tube of the coset bijectively onto itself, and it changes the Matsumoto--Amano cost by at most one, $t_{\min}(VW)\le t_{\min}(W)+1$. It does not raise the cost of every word: $V=HT$ and $W=V^{-1}$ each cost one $T$ while $VW=I$ costs none. The consequence is therefore one about cumulative counts and not about shells---the tube words of cost $\le\tau$ contain $V$ times those of cost $\le\tau-1$---and that containment alone is no stronger than monotonicity of the cumulative count. What carries the explanation is the enumeration itself, shell by shell: at $t=11$, $57$ of the $156$ tube words of that shell are $V$ times a tube word of shell $10$, against a Haar prediction of $72$ words for the shell. Near the onset of the $2^\tau\varepsilon^2$ term the shell counts run at $2.2$--$2.4$ times the Haar value ($t=8$--$10$); the fraction of genuinely new words then falls below Haar ($0.78$ at $t=15$), the shell ratio returns to $1.07$ by $t=17$, and the required constant falls from $53.5$ at $\tau=11$ to $28.5$ at $\tau=17$. The same frame at $\varepsilon=2^{-6}$ shows no overshoot at all up to $t=17$ (shell ratio $\le1.09$, required $c\le27.8$). The overshoot is therefore a finite-cost effect of the orbit of a cheap word inside the tube, not a change of exponent. The identity coset shows the same effect in a sharper form, and at a non-dyadic accuracy. Its tube is invariant under left and right multiplication by powers of $T$ and under conjugation by $X$, so words enter the count in whole orbits of those symmetries, all at the same distance. At $\varepsilon=0.024532$---just above the distance $0.0245310$ of one such orbit---and $Q=Q_\varepsilon=32$, the cumulative count at $\tau=11$ is $232$, in shells $4,4,0,\dots,0,64,160$ and at only two distinct distances $0.022965$ and $0.024531$, against the $8+22+c\cdot1.23$ that~\eqref{eq:H} allows: this forces $c\ge163.9$, and it rules out both $48$ and $96$. At the dyadic $\varepsilon=2^{-5}$ the same coset needs only $c\le24.8$, so the step lies between the dyadic accuracies that our scans sample. This is why we leave $c$ unspecified, and why the value $53.5$ below is not to be read as an estimate of it: a finite scan bounds $c$ from below, never from above. Two scans locate it. The first takes $25$ frames---identity, Clifford and Haar pairs, the exact axes of fourteen words of cost $\le3$, the $HT$ axis perturbed by $0.01$ to $0.1$ rad, and four angular offsets of the grid---at $\varepsilon\in\{2^{-3},\dots,2^{-6}\}$ and $\tau\le15$: $100$ cells, $1600$ triples. The second takes $272$ frames, among them all $219$ distinct rotation axes of words of cost $\le3$ and forty axes of cost-$4$ and cost-$5$ words, at $\varepsilon\in\{2^{-3},\dots,2^{-7}\}$ and $\tau\le14$: $1360$ cells, $20\,400$ triples. The maximum required $c$ is $53.5$ in both, and the second produces no new counterexample. Of its $17$ cells above $48$, sixteen are the orbit of the $HT$ axis under the Clifford group---the twelve directions of that axis together with four cost-$2$ and cost-$3$ words sharing it, all at $\varepsilon=2^{-5}$, $\tau=11$, with identical shell counts---and the seventeenth is the same axis perturbed by $0.003$ rad, at $49.5$. Everything further away is smaller: axes of other cheap words need at most $43.5$, a perturbation of $0.1$ rad restores the Haar value, and Haar frames give $23.0$--$25.4$ at every accuracy. The overshoot is also confined to one accuracy: the per-accuracy maxima are $29.5$, $24.8$, $\mathbf{53.5}$, $29.0$ and $11.0$ for $\varepsilon=2^{-3},\dots,2^{-7}$. In neither scan is the linear term $(8,2)$ ever breached below onset ($0$ of $1360$ cells). Profiling four further axes to $\tau=17$ at $\varepsilon=2^{-6}$ and $2^{-7}$ gives peaks between $12.6$ and $29.5$, so the $HT$ axis at $L=5$ is the only frame in the dyadic scan above twice the Haar value, and at $L=6$ and $7$ the same and neighbouring axes stay below $30$. That scan is nevertheless blind to the identity-coset step above, which is larger and sits at a non-dyadic $\varepsilon$: sampling accuracies on the dyadic ladder is what hides the worst steps. Not covered: pairs $(G_1,G_2)$ not of the form $(D,D^\dagger V)$, axes of words of cost $\ge6$, $L\ge8$, and non-dyadic accuracies other than the one exhibited. Computed by \texttt{scripts/conj1\_scan.py}, \texttt{scripts/scan2.py}, \texttt{scripts/ht\_profile.py} and \texttt{scripts/astra1.py} (the identity-coset step); \texttt{scripts/scan2\_report.py} recomputes the scan totals from the stored scan.

\paragraph{Exact optimal synthesis and the frontier (Fig.~\ref{fig:frontier}).} The same enumeration yields, for $L\le6$, the minimal $T$-count $\tau_{\min}(d)$ of a word within $\varepsilon$ of every grid rotation $R_z(2\pi d/Q)$; this is the exact optimum, independent of factoring. At $L=5$ ($Q=25$) the mean synthesis cost over the $24$ nonzero grid angles is $11.83$ and the largest is $16$ ($3L=15$); at $L=6$ ($Q=50$) they are $15.59$ and $18$ ($3L=18$). Fractional passthrough with canonical sharing ($m=64$, $r=2$) then has $S=q\log_2Q$ and $\Tpre=(m-q)\E_u\tau_{\min}(u)$, in which $u$ is uniform on all of $\Z_Q$ and the zero share is free, so the mean that sets the slope is $\E_u\tau_{\min}=11.36$ and $15.28$ --- the $d\ne0$ means above, times $(Q-1)/Q$ --- giving $\E\tau/\log_2K=2.48$ ($L=5$) and $2.72$ ($L=6$). Both are computed at word accuracy $\varepsilon$, so they calibrate a single rotation and are \emph{not} the $r=2$ point of Theorem~\ref{thm:frontier}, which certifies a committed share only at $\varepsilon/2$. Repeating the enumeration at that budget resolves every grid point, at $T$-count $18$ for $L=5$ and $22$ for $L=6$, and gives $\E_u\tau_{\min}=15.28$ and $18.72$---slopes $3.333$ and $3.334$, worst cases $18$ and $22$. The two budgets approach the asymptotic $3$ from opposite sides: the calibration at $\varepsilon$ rises through $2.478$, $2.721$ and $3.113$, the last at $\varepsilon=10^{-10}$ below, while the certified $\varepsilon/2$ family falls through $3.333$, $3.334$ and $3.205$. Fig.~\ref{fig:frontier}(a) plots both, and it is the second sequence that Theorem~\ref{thm:frontier} bounds and that may be compared with panel (b). Computed by \texttt{scripts/frontier\_eps\_half.py}, whose $\varepsilon$ column reproduces \texttt{data/taumin\_exact.json} grid point by grid point. A partial-commitment variant that commits the low $\ell$ bits of a share as a rotation and stores the rest---cost $\E_u\tau_{\min}(u\bmod2^\ell)$ per coordinate for $\ell$ bits shed---costs $4.8,3.2,2.9,2.7$ $T$ gates per bit shed for $\ell=1,\dots,4$ at $L=5$, and $8.5,6.1,4.6,3.7,3.0$ for $\ell=1,\dots,5$ at $L=6$: strictly worse than passthrough and worsening as $\ell$ decreases, which is the quantum of commitment seen at exact optimality. At $L\le6$ the additive terms swamp the main terms: at $L=6$ they are $m\log_236=5.17m$ for Theorem~\ref{thm:main1} and $mA_2=15.4m$ for Theorem~\ref{thm:rate2}, against $m\log_2K=5.6m$, so both bounds are vacuous, and so is Theorem~\ref{thm:main2}, whose additive term $\kappa(1+\log_2N_0)$ already exceeds $\kappa\log_2K$ at $L=6$. Fig.~\ref{fig:frontier}(a) therefore tests slopes, not intercepts.

\paragraph{Chemistry scale (Fig.~\ref{fig:frontier}(b)).} Exhaustive enumeration stops at $L\le6$, so at a chemically relevant accuracy we replace $\tau_{\min}$ by Ross--Selinger synthesis (\texttt{gridsynth}, whose $\varepsilon$ is the up-to-phase operator norm and so is exactly $\dproj$). At $\varepsilon=10^{-10}$ the tuned modulus is $Q_\varepsilon=7\,853\,981\,633$ and $\log_2K_\varepsilon=32.871$. Under canonical sharing the share $u$ is uniform on $\Z_Q$ and the frontier depends on it only through $\E_u\tau(u)$, so we sample rather than synthesize all $Q_\varepsilon$ rotations: $20\,000$ uniform $u$ give $\E_u\tau=102.32\pm0.02$ (s.d.\ $2.16$, range $88$--$113$) against $3L=99.66$, hence $\E\tau/\log_2K=3.113\pm0.001$; this is the single-rotation cost at accuracy $\varepsilon$, a calibration. The certified full-stream point of Theorem~\ref{thm:frontier} with $r=2$ synthesizes each committed share to $\varepsilon/2=5\times10^{-11}$; the same $20\,000$ angles synthesized at that accuracy give $\E_u\tau=105.34$, i.e.\ $3.205$ per bit, three gates more per share as the $3\log_2r$ of Ross--Selinger scaling predicts, and it is this point that Table~\ref{tab:numbers}, Table~\ref{tab:models} and Fig.~\ref{fig:frontier}(b) use. Both numbers exceed $3$ slightly, as expected --- $\tau_{\mathrm{RS}}=3L+O(\log L)$, and Corollary~\ref{cor:frontier} asserts only the limit. With $m=10^4$ and $r=2$ this is $\Tpre=1.053\times10^6$ committed $T$ gates at $S=0$; both axes are linear in $m$, so $m=10^5$ is the same line scaled tenfold. Partial commitment of the low $\ell$ bits, with the committed piece synthesized at $\varepsilon$, costs $61,46,35,28,22,19,16,14$ $T$ gates per bit shed for $\ell=1,\dots,8$ (at $\varepsilon/2$ each entry rises by about three per share): the quantum of commitment persists at scale, at between $4.5$ and $20$ times the calibration slope $3.113$ at the same accuracy, because in this model a rotation by any nonzero grid angle, however small, costs a full synthesis --- the small angles $2\pi k/Q$ with $0<k<2^\ell$ here average $110$ to $122$ $T$ gates, slightly \emph{above} the generic $102$. The one exception is the exact commitment of Remark~\ref{rem:firstbit}, which on the admissible grid $Q=8\lfloor Q_\varepsilon/8\rfloor$ sheds two bits per coordinate for no magic and a third for $\tfrac12$ a $T$ gate; those points, and not the $\ell=1$ cross, are the true left end of the family. That the fee is a property of the deterministic-unitary model and not of fault-tolerant synthesis in general is the subject of Sec.~\ref{sec:models}.

\paragraph{The bounds at chemistry scale (Table~\ref{tab:numbers}).} With $C_1,C_2$ explicit, the lower bounds can be evaluated rather than compared as rates. At $\varepsilon=10^{-10}$ and $S=0$, Theorem~\ref{thm:main1} gives $\bar T_t\ge21.75\,m$, Theorem~\ref{thm:rate2} gives $25.72\,m$ (with $A_2=20.0$ bits of the $32.9$ available), Theorem~\ref{thm:main2} gives $52.9\,m$, and passthrough achieves $105.34\,m$ at the certified share accuracy: the conditional bound is within a factor $2$ of achievability, the unconditional ones within a factor $4.1$ (Theorem~\ref{thm:rate2}) and $4.8$ (Theorem~\ref{thm:main1}). The rate-two bound pays its additive term twice, so it exceeds the rate-one bound only for $S<3.7\,m$ bits---$11\%$ of the range of $S$, at the memoryless end---and only because $\varepsilon$ is small enough to begin with: at $S=0$ the two agree at $L=28.7$, and below that accuracy the rate-one bound is the stronger of the two. Fig.~\ref{fig:frontier}(b) plots these bounds themselves, not their leading-order slopes: because the bounds of Theorems~\ref{thm:main1} and~\ref{thm:main2} contain a logarithm of $\bar T_t$, their finite-precision derivatives $-\partial\bar T_t/\partial S$ at $S=0$ are $0.94$ and $2.61$ rather than the leading coefficients $1$ and $\kappa=2.82$. All four numbers are per coordinate and per round, are computed by \texttt{scripts/thm\_numbers.py}, and scale linearly in $m$.

\begin{table}[t]
\centering\footnotesize
\caption{The bounds at $\varepsilon=10^{-10}$ ($L=33.22$, $Q_\varepsilon=7\,853\,981\,633$, $\log_2K=32.871$), per coordinate per round, at $S=0$. ``Coeff.'' is the leading coefficient of $-S$ in the bound (for the bounds of Theorems~\ref{thm:main1} and~\ref{thm:main2} the finite-precision derivative $-\partial\bar T_t/\partial S$ at $S=0$ is $0.94$ and $2.61$, because of their logarithmic terms); ``per bit'' is $\bar T_t/(m\log_2K)$ at $S=0$. The achievable row is the fractional-passthrough family with each committed share synthesized at $\varepsilon/2$, the accuracy Theorem~\ref{thm:frontier} requires for $r=2$; the single-rotation calibration at $\varepsilon$ is $102.32$. The Theorem~\ref{thm:main2} row is conditional on Conjecture~\ref{conj:H}, whose constant $c$ we do not fix; the range given is $c=256$ to $c=48$, two illustrative values spanning the largest requirement we know ($c\ge163.9$, Sec.~\ref{sec:num}). It is a sensitivity check, not an interval estimate.}
\label{tab:numbers}
\setlength{\tabcolsep}{4pt}\begin{tabular}{@{}l ccc@{}}
\toprule
& $\bar T_t/m$ & per bit & coeff.\\
\midrule
Thm.~\ref{thm:main1} (unconditional) & $21.75$ & $0.66$ & $1$\\
Thm.~\ref{thm:rate2} (given (R)) & $25.72$ & $0.78$ & $2$\\
Thm.~\ref{thm:main2} (under Conj.~\ref{conj:H}) & $51.3$--$52.9$ & $1.56$--$1.61$ & $\kappa=2.75$--$2.82$\\
achievable (Thm.~\ref{thm:frontier}) & $105.34$ & $3.20$ & $3.205$\\
\quad(\texttt{gridsynth} at $\varepsilon/2$) & & &\\
\bottomrule
\end{tabular}
\end{table}

\begin{figure*}[t]
\centering\includegraphics[width=0.92\textwidth]{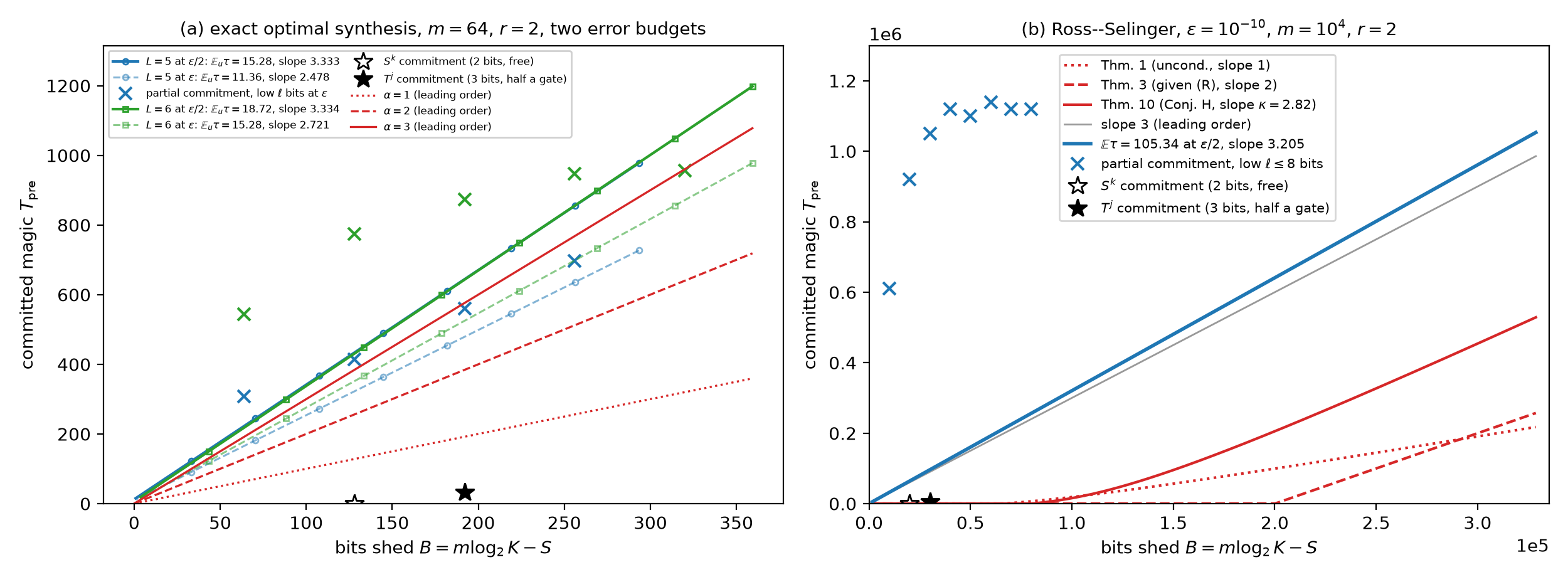}
\caption{The fractional-passthrough tradeoff family ($r=2$) against the lower bounds. (a) Exact minimal $T$-counts from the enumeration, $m=64$, $L=5,6$---an optimum no synthesis algorithm can beat---against the leading-order slopes $\alpha=1,2,3$; at $L\le6$ the additive terms swamp the main terms, all three bounds being vacuous, so only slopes can be compared there. Solid: the committed piece synthesized at $\varepsilon/2$, the accuracy Theorem~\ref{thm:frontier} requires for $r=2$, which is the family the theorem bounds. Dashed: the same optima at word accuracy $\varepsilon$, a single-rotation calibration and not a point of the family. The two budgets straddle $\alpha=3$---slopes $3.333$ and $3.334$ above it, $2.478$ and $2.721$ below---and the $\varepsilon/2$ slopes at $L=5$ and $6$ coincide to three decimals. (b) Ross--Selinger synthesis at $\varepsilon=10^{-10}$, $m=10^4$, from $20\,000$ sampled grid rotations, against those three bounds as stated, intercepts included (Table~\ref{tab:numbers}; the rate-three line uses the constants of Theorem~\ref{thm:main2} with $c=48$, intercept $52.9\,m$---at $c=256$ the line shifts down by $1.6\,m$, about the width of the stroke); the achievable line is the $\varepsilon/2$ point $\E\tau=105.34$; the grey line is the leading-order slope $3$, for comparison. The rate-two bound (dashed) overtakes the rate-one bound (dotted) only near the right-hand end, $S<3.7\,m$ bits. Crosses: partial commitment of the low $\ell$ bits of each share, committed piece synthesized at $\varepsilon$ in both panels. Star: the exact commitment of Remark~\ref{rem:firstbit} on the grid $Q=8\lfloor Q_\varepsilon/8\rfloor$, three bits per coordinate for half a $T$ gate (the first two of them free). Its abscissa is $2-\log_2\frac{Q}{Q-1}$ and $3-\log_2\frac{Q}{Q-1}$ rather than exactly $2$ and $3$ ($2.9386$ for the third bit at $Q=24$); points taken on different moduli are all plotted as $\log_2(Q-1)-S/m$.}
\label{fig:frontier}
\end{figure*}

\section{Model dependence}\label{sec:models}

\begin{table*}[t]
\centering\small
\caption{Exchange rate $\alpha$ (committed $T$ gates per bit of memory forgone) in several synthesis models.}
\label{tab:models}
\begin{tabular}{@{}p{3.6cm} p{5.2cm} p{4.6cm} p{3.4cm}@{}}
\toprule
\raggedright Model & \raggedright Rigorous lower bound & \raggedright Achievable & \raggedright Granularity\tabularnewline
\midrule
\raggedright CW ancilla-free Clifford+$T$ & \raggedright $1$ (Thm.~\ref{thm:main1}); $2$ with explicit constants (Thm.~\ref{thm:rate2}, via the theorem (R)); $\tfrac{11}5$ asymptotically and without (R) (Thm.~\ref{thm:rate115}); $\tfrac{17}7$ by a height dichotomy (Thm.~\ref{thm:rate177}) and every rate $<\tfrac52$ with rational projections (Thm.~\ref{thm:rate52}), asymptotically and without (R); $3-o(1)$ for rotation-segmented processes (Thm.~\ref{thm:rate3seg}); $1+\beta\to3$ under Conj.~\ref{conj:H}, any $c$ (Thm.~\ref{thm:main2}) & \raggedright $\to3$ (Thm.~\ref{thm:frontier}, under the synthesis hypothesis of Cor.~\ref{cor:frontier}); $3.205$ by \texttt{gridsynth} at the certified share accuracy $\varepsilon/2$, $\varepsilon=10^{-10}$; exact $3.333$ ($L=5$) and $3.334$ ($L=6$) at $\varepsilon/2$, against calibrations $2.478$ and $2.721$ at $\varepsilon$ (Sec.~\ref{sec:num}) & \raggedright one rotation (fee $2L-\log_2c-O(1)$) under Conj.~\ref{conj:H}, past the first $O(\log L)$ bits of a coordinate (Rem.~\ref{rem:firstbit}, Prop.~\ref{prop:cluster})\tabularnewline
\raggedright CW $+$ probabilistic mixing (Def.~\ref{def:mixed}; \cite{Campbell17,Hastings17,KLMPP,Bothe}) & \raggedright $0$ on the low bits; $1-\gamma$ per coarse bit \emph{in total}, beyond a $\log_236+O(\log L)$ allowance (Thm.~\ref{thm:mixing}(a)); $2(1-\gamma)$ given (R) (Thm.~\ref{thm:mixing}(b)); $1+\beta\to3$ under Conj.~\ref{conj:H}, any $c$ (Cor.~\ref{cor:mixq}); floor $L-O(\log L)$ per share, $\tfrac32L$ under Conj.~\ref{conj:H}. All vacuous at $\varepsilon=10^{-10}$ except Thm.~\ref{thm:mixing}(a) & \raggedright $O(1)$ on the low $\tfrac12L-\tfrac12\log_2L$ bits; $\tfrac32L$ for a share that commits more, under the mixed-synthesis hypothesis~\eqref{eq:mixhyp}---four-word mixing \cite[Thm.~2]{Campbell17} at word accuracy $\sqrt{\varepsilon/(5r)}$, estimated $57.5$ at $\varepsilon=10^{-10}$, $r=2$ by RS scaling and not certified; $1.52\log_2(1/\delta)-0.01=50.5$ $T$ per share at $\delta=10^{-10}$ \cite{KLMPP,Bothe} & \raggedright bits below $\tfrac12L-\tfrac12\log_2L-O(1)$ free (Cor.~\ref{cor:mixstop}(i), \cite{Bothe}); bits at $\tfrac12L-O(1)$ cheap for no program (Lem.~\ref{lem:nocheap}); the top $\log_2\gcd(Q,8)$ bits free at any position; fee $L-O(1)$ under Conj.~\ref{conj:H} (Cor.~\ref{cor:mixq})\tabularnewline
\raggedright CW $+$ RUS/fallback on $n$ qubits (measurement-adaptive, unitary branches; Prop.~\ref{prop:adaptive}) & \raggedright $(1-\gamma)/(2n+2)$ per high bit, in $T$ gates plus measurements; floor $L/(4n+4)$ per share & \raggedright $\tau\approx1.15L$ \cite{BRS}; $0.53\log_2(1/\delta)+4.86=22.5$ at $\delta=10^{-10}$ with mixed fallback \cite{KLMPP,Bothe} & \raggedright open: the token count of Prop.~\ref{prop:adaptive} does not see the fee\tabularnewline
\raggedright Quasi-probabilistic mixtures \cite{Bothe} & \raggedright not covered: a signed mixture is not a channel, and the resource is sample overhead & \raggedright $\tilde O(\theta^2/(\lambda-1))$ at small angles \cite{Bothe} & \raggedright open\tabularnewline
\raggedright Phase gradient + adder, $Q=2^b$ & \raggedright none for general prefixes & \raggedright $\approx4$ per coordinate, unbatched~\cite{Gidney}: commit the high $b-\ell$ bits of a share into the top gradient bits, remember the low $\ell$ bits, one $b$-bit addition in the suffix; batched, see next row & \raggedright one bit, no fee: the catalytic gradient state pays the fee once\tabularnewline
\raggedright General $n$-qubit Clifford+$T$, clean ancillas & \raggedright $(m\log_2K-S-\log_2|\mathcal C_n|-O(n))/(2n+1)$ from $\#\le|\mathcal C_n|(2\cdot4^n)^t$; degenerate for $n\sim m$ & \raggedright $O(1/\log L)$ per bit with a phase-gradient catalyst (Prop.~\ref{prop:batch}); $80$ (total error) and $64$ (coordinatewise) $T$ per share at $\varepsilon=10^{-10}$, $m=10^4$, $r=2$ & \raggedright no constant-rate law\tabularnewline
\raggedright $Q\mid8$ & \raggedright vacuous: $\log_2K\le\log_27<\log_236$, so Theorem~\ref{thm:main1} asserts nothing & \raggedright $\le1$ $T$ gate per share: every grid rotation is a power of $T$, exactly & \raggedright one bit, no fee. This case is \emph{not} excluded by $\varepsilon<\sin(\pi/2Q)$, which at $Q=8$ permits every $\varepsilon\le1/8$; it is the $c_0$ of Conj.~\ref{conj:H} filling the whole grid\tabularnewline
\bottomrule
\end{tabular}
\end{table*}

Table~\ref{tab:models} summarizes what changes with the synthesis model. Three entries deserve comment. In the phase-gradient model the exchange is bit-granular: adding the high bits of a share into the top bits of the gradient register is an exact operation whose cost is linear in the number of bits, so the fee of the CW model disappears; the catalytic gradient state is, in this sense, a device that has paid the fee once for all rotations. In the mixed-channel models now used in resource estimation the fee also disappears, but only at small angles and for a different reason: a rotation by $\theta$ may be realized as a probabilistic or quasi-probabilistic mixture of Clifford+$T$ channels in which the identity, at zero cost, carries most of the weight, giving an average cost $\tilde O(\theta^2/\delta)$ instead of $O(\log(1/\delta))$~\cite{Bothe}. The low bits of a share are exactly the small angles, so a mixing implementation can shed them cheaply, and the quantization of Theorem~\ref{thm:main2} does not survive in that range. The saving stops there and not later: Sec.~\ref{sec:mixing} extends the counting argument to channels and shows that mixing rescales the law by one half rather than removing it, the free bits being exactly the ones the rescaled bound does not charge. In the fully general $n$-qubit model even the premise fails: a $T$ gate on $n$ qubits can be any of $2\cdot4^n$ Pauli rotations, and with ancillas a tensor product of many single-qubit rotations can be synthesized for far less than the sum of their costs, the optimal $T$-count for an arbitrary diagonal $n$-qubit unitary being $\Theta(\sqrt{2^n\log(1/\varepsilon)}+\log(1/\varepsilon))$~\cite{GKW}. No constant-rate law survives ancilla-assisted synthesis across coordinates.

Two cases have to be separated. If ancilla qubits may survive a cut, a share can be copied into them with $X$ gates, which are Clifford, and read back in round $r$; then $\Tpre=0$ at $S=0$. Such qubits are memory and must be charged as such. The meaningful model therefore has \emph{clean} ancillas, prepared in $|0\rangle$ and returned to $|0\rangle$ within a round, and possibly a catalytic register that does not depend on the stream.

\begin{proposition}[Batched commitment]\label{prop:batch}
Allow in every round clean ancillas, the measurement-based uncomputation of~\cite{Gidney}, and a catalytic phase-gradient register of $b:=\lceil L+\log_2(2\pi mr)\rceil$ qubits. Then for every $g\ge1$ there is a one-pass process with total projective error $\le\varepsilon$ on every valid stream, $S=O(\log(mQr))$, and
\[
T_t\ \le\ \frac mg\big(8(2^g-1)+4(b-1)\big)\qquad(t<r).
\]
With $2^g\asymp L/\log L$ this is $T_t\le(4+o(1))\,mL/\log_2L$ for $\log(mr)=o(L)$, so the exchange rate is $O(1/\log L)=O(1/\log\log(1/\varepsilon))$.
\end{proposition}
\begin{proof}
Split the coordinates into blocks of $g$. Up to a global phase, the round-$t$ rotation of a block acts on its basis states $|x\rangle$, $x\in\{0,1\}^g$, as $e^{\mathrm i\Delta\phi(x)}$ with $\phi(x)=\sum_{j}a^{(t)}_jx_j$; here we conjugate by the public Clifford $V$ as in Sec.~\ref{sec:setting}. Let $\tilde y_j$ be $2^ba^{(t)}_j/Q$ rounded to an integer, and $y(x):=\sum_j\tilde y_jx_j\bmod2^b$. Then $|2\pi y(x)/2^b-\Delta\phi(x)|\le\pi g2^{-b}$.

The process applies three steps.
\begin{enumerate}\setlength\itemsep{0pt}
\item A unary-iteration lookup~\cite{Babbush} writes $y(x)$ into a clean $b$-qubit register with $4(2^g-1)$ $T$ gates. The data are written by CNOTs, so the cost does not depend on $b$.
\item An addition of that register into the phase-gradient register~\cite{Gidney} costs $4(b-1)$ $T$ gates and applies the phase $e^{2\pi\mathrm iy(x)/2^b}$.
\item The lookup is uncomputed at the same cost.
\end{enumerate}
All ancillas return to $|0\rangle$ and the gradient register is unchanged. The error per block is $\le\pi g2^{-b}$ in operator norm, and over $m/g$ blocks and $r$ rounds it is $\le\pi mr2^{-b}\le\varepsilon/2$. Every share is committed on arrival, so the snapshot holds only a program counter. Blocks number $\lceil m/g\rceil$.
\end{proof}
Proposition~\ref{prop:batch} is an instance of known ancilla-assisted diagonal synthesis. $\bigotimes_{j\le g}R_z$ is a $g$-qubit diagonal, of $T$-count $O(\sqrt{2^g\log(1/\varepsilon)}+\log(1/\varepsilon))$~\cite{LKS,GKW}, hence $O(L/\log L)$ per rotation at $2^g\asymp L$. Phase-gradient addition of looked-up angles is standard~\cite{Gidney,Sanders}. What breaks the law is batching across coordinates: with $g=1$ the construction costs $4b+4\approx4L$ per share. The gradient register is prepared once, before round $1$, with $O(b\log(b/\varepsilon))$ $T$ gates; since it is returned unchanged, its preparation error enters once, not per use. The asymptotics are slow: at $L=50$ the cost is still about $2.8$ times $4L/\log_2L$, and at $\varepsilon=10^{-10}$ batching saves about $40\%$, not an order of magnitude.
At $\varepsilon=10^{-10}$, $m=10^4$ and $r=2$ the best block size is $g=4$. The cost is $80$ $T$ gates per share at total error $\varepsilon$ ($b=51$). Since $e^{2\pi\mathrm iy(x)/2^b}=\prod_je^{2\pi\mathrm i\tilde y_jx_j/2^b}$ is a product of single-qubit phases, a coordinate accumulates $\dproj\le\pi r2^{-b-1}$. So under the coordinatewise criterion, for which passthrough is certified, $b=\lceil L+\log_2(\pi r/2)\rceil=35$ suffices, and the cost is $64$ ($1.95$ per bit) against the $105.34$ of ancilla-free passthrough. Asymptotically the ratio tends to zero like $1/\log L$. The lower-bound side is open: counting circuits on $n$ qubits gives only the degenerate bound of Table~\ref{tab:models}, and magic monotones bound the magic of the final unitary, which is $\Theta(m)$, not the magic committed before the information arrived~\cite{BCHK,GKW}. The constant-rate law of this paper is therefore specific to coordinatewise synthesis, with ancillas bounded per coordinate as in Proposition~\ref{prop:adaptive}, and that is where the lower bounds above live.

\subsection{The law under mixing}\label{sec:mixing}

The counting argument of Secs.~\ref{sec:tube}--\ref{sec:conj} concerns deterministic words. Mixed-channel synthesis~\cite{KLMPP,Bothe} and the mixing constructions of Campbell~\cite{Campbell17} and Hastings~\cite{Hastings17} replace the committed word by a program that is sampled at run time; the channel implemented is then closer to the target than any word in its support, and it is this gap that lets a small angle be shed for almost nothing~\cite{Bothe}. We show that the gap is exactly a square root: what mixing buys is that a committed word need only be accurate to $\sqrt\varepsilon$, so that a share is charged for the top half of its bits instead of all of them, and the rate at which those bits are charged is unchanged.

\setcounter{definition}{2}
\begin{definition}[Mixed-unitary CW process]\label{def:mixed}
A \emph{mixed-unitary CW process} is a charged one-pass process (Definition~\ref{def:model}) whose output tape holds, for each coordinate $j\in[m]$ and round $t\in[r]$, a \emph{program} $\pi^{(t)}_j$: a finitely supported probability distribution on single-qubit Clifford+$T$ words. At execution each program is sampled once, independently of every other program and of the tape $\rho$, and the sampled word is applied. The operation implemented on coordinate $j$ is the channel
\begin{equation}
\mathcal E_j=\mathcal E^{(r)}_j\circ\cdots\circ\mathcal E^{(1)}_j,\quad
\mathcal E^{(t)}_j:=\sum_{W}\pi^{(t)}_j(W)\,\mathcal U_W,
\end{equation}
and the process is \emph{correct} on a stream if $\|\mathcal E_j-\mathcal U_{R_z(\Delta x_j)}\|_\diamond\le2\varepsilon$ for every $j$. The charge $S$ is as in Definition~\ref{def:model}. The committed magic is the expected number of $T$ gates \emph{executed} while reading round $t$,
\begin{equation}
T_t:=\sum_j\E_{W\sim\pi^{(t)}_j}t(W),\qquad\Tpre:=\sum_{t<r}T_t,
\end{equation}
with $t(W)$ the number of $T$ records of the word as written. Finally $n_f$ denotes the largest, over $j$ and $t<r$, of $\log_2N^{<t}_j+\log_2N^{>t}_j$, where $N^{<t}_j$ and $N^{>t}_j$ are the numbers of distinct words in the supports of the prefix programs $\pi^{(1)}_j,\dots,\pi^{(t-1)}_j$ and the suffix programs $\pi^{(t+1)}_j,\dots,\pi^{(r)}_j$. It is a parameter of the process that Definition~\ref{def:mixed} does not bound; the results below do not use it, and it appears only in the sharper bounded-support variant of Remark~\ref{rem:mixconst}.
\end{definition}
\setcounter{definition}{0}

Three remarks on the definition. (i) The criterion contains the CW criterion: for unitaries $\|\mathcal U-\mathcal V\|_\diamond=2\sin(\Theta/2)=2\dproj(U,V)\sqrt{1-\dproj^2/4}\le2\dproj(U,V)$ with $\Theta$ as in Lemma~\ref{lem:dproj} \cite[Thm.~3.55]{Watrous}, so a deterministic CW process correct to $\dproj\le\varepsilon$ is a mixed-unitary process correct in the above sense, with one-point programs and $n_f=0$. (ii) Independence of the samples across rounds is not a restriction but the statement of the model: a seed shared by the executor across the cut is side information in the sense of Theorem~\ref{thm:side}, and Corollary~\ref{cor:side}(b)--(c) prices it. (iii) Measurement-adaptive protocols (repeat-until-success, fallback~\cite{BRS,KLMPP}) fall outside Definition~\ref{def:mixed}, because the operation realized on a branch is not an ancilla-free Clifford+$T$ word; Proposition~\ref{prop:adaptive} extends the bound to them at a weaker rate. Quasi-probabilistic mixtures~\cite{Bothe} remain outside: a signed mixture is not a channel, its correctness is a statement about an estimator, and its resource is sample overhead.

The two elementary facts we need are the following; proofs are in Appendix~\ref{app:mixing}.

\begin{lemma}[Fidelity from diamond error]\label{lem:fid}
For a qubit channel $\mathcal E$ and a unitary $V$ let $F_e(\mathcal E,V):=\langle\Phi_V|J(\mathcal E)|\Phi_V\rangle$ be the entanglement fidelity, $J$ the normalized Choi state and $|\Phi_V\rangle=(V\otimes I)|\Phi\rangle$. Then $1-F_e(\mathcal E,V)\le\tfrac12\|\mathcal E-\mathcal U_V\|_\diamond$, and for a mixed-unitary channel $\mathcal E=\sum_ip_i\mathcal U_{X_i}$ one has $F_e(\mathcal E,V)=\sum_ip_i|\tr(V^\dagger X_i)/2|^2$.
\end{lemma}

\begin{lemma}[Concentration: window and tube]\label{lem:conc}
Let $X_1,\dots,X_N\in\SU$, let $p$ be a distribution, $V\in\SU$, $f_i:=|\tr(V^\dagger X_i)/2|^2$, and suppose $\sum_ip_if_i\ge1-\eta^2$ with $\eta\le\tfrac12$.

(i) \emph{(window)} For every $\psi\in[-\pi,\pi]$,
\begin{equation}\label{eq:window}
\sum_ip_i\Big|\tfrac12\tr\big((VR_z(\psi))^\dagger X_i\big)\Big|^2\le\cos^2\tfrac\psi2+\eta^2+\eta|\sin\psi|,
\end{equation}
and the right side is $<1-\eta^2$ whenever $\sin(|\psi|/2)>(1+\sqrt3)\eta$.

(ii) \emph{(tube)} If $f_i\ge1-\eta^2$ then $\dproj(X_i,V)\le\sqrt2\,\eta$.
\end{lemma}

\begin{theorem}[The law under mixing]\label{thm:mixing}
Let the process be mixed-unitary (Definition~\ref{def:mixed}), correct with probability $\ge1-\delta$ over $\rho$ on every valid stream, and fix $\gamma\in(0,1)$; put $\eta:=\sqrt{\varepsilon/\gamma}$,
\[
\begin{gathered}
w:=\Big\lfloor\tfrac Q\pi\arcsin\big((1+\sqrt3)\eta\big)\Big\rfloor,\qquad
\varepsilon':=36\sqrt2\,\eta,\\
w':=\Big\lfloor\tfrac{2Q}\pi\arcsin\varepsilon'\Big\rfloor,
\end{gathered}
\]
$k_\gamma:=\log_2K-\log_2(2w+1)$ and $k'_\gamma:=\log_2K-\log_2(2w'+1)-1$, the \emph{coarse bits} of a coordinate. Then for every round $t<r$:

(a) \emph{(rate one, unconditional)} If $\gamma\ge4(1+\sqrt3)^2\varepsilon$, so that $(1+\sqrt3)\eta\le\tfrac12$,
\begin{multline}\label{eq:mixone}
S+\frac{\bar T_t}{1-\gamma}+m\log_236+\rho_m\!\Big(\frac{\bar T_t}{1-\gamma}\Big)\\ \ge\ (1-\delta)\,m\,k_\gamma-h_2(\delta).
\end{multline}

(b) \emph{(rate two, given (R))} If $\varepsilon'\le\tfrac18$, then with $K':=\lfloor Q/(2w'+1)\rfloor$, $\tau_*':=\lceil2\log_2K'\rceil$, $n_{\varepsilon'}(\tau):=C_12^\tau\varepsilon'^2+C_2(\tau+1)2^{\tau/2}$ and $A_2':=1+\log_2\sum_{\tau\le\tau_*'}n_{\varepsilon'}(\tau)2^{-\tau/2}$,
\begin{multline}\label{eq:mixtwo}
\bar T_t\ \ge\ 2(1-\gamma)\big[(1-\delta)\,m\,(k'_\gamma-1)-\delta\,m\log_2K'\\-2m\,h_2(\delta)-S-m\,A_2'\big].
\end{multline}
No assumption is made on the size of the supports of the programs: the prefix and suffix programs of a coordinate may mix over arbitrarily many words.
\end{theorem}
For the tuned family $k_\gamma=\tfrac12L-0.80-\tfrac12\log_2\tfrac1\gamma+o(1)$ and $k'_\gamma=\tfrac12L-7.0-\tfrac12\log_2\tfrac1\gamma+o(1)$, the constant $7.0=\log_2(4\cdot36\sqrt2/\pi)+1$ being the price of the frame-free argument (Remark~\ref{rem:mixconst}), while $A_2'=\log_2C_2+2\log_2L+O(1)$. Hence with $\delta=0$ and $\gamma=1/L$, \eqref{eq:mixtwo} reads $\bar T_t\ge m\log_2K-2S-O(m\log L)$, i.e.
\begin{equation}\label{eq:mixfloor}
\bar\Tpre\ \ge\ (r-1)\big(m\log_2K-2S\big)-O(rm\log L):
\end{equation}
under mixing a coordinate that is not held still costs at least $L-O(\log L)$ committed $T$ gates per round---one half of the deterministic floor $2L$ of Theorem~\ref{thm:rate2}---and the rate at which the bits that remain chargeable are paid is still two per bit, the factor $2$ in front of $S$. The proof (Appendix~\ref{app:mixing}) is the argument of Theorems~\ref{thm:main1} and~\ref{thm:rate2} applied not to the committed word but to a \emph{representative}: the cheapest word in the support of $\pi^{(t)}_j$ whose composite channel still has fidelity $\ge1-\eta^2$ with the target. Markov's inequality gives such a word mass $\ge1-\gamma$, which is where the factor $1-\gamma$ comes from; Lemma~\ref{lem:conc}(i) shows that it list-decodes the share to a window of half-width $w$. For (b) the representative must in addition be placed in a tube around a coset whose frame does not depend on the share, which Lemma~\ref{lem:cluster} does without bounding the supports, at the price of the radius $36\sqrt2\eta$ in place of $\sqrt2\eta$.

\begin{corollary}[Where mixing stops paying]\label{cor:mixstop}
Write the shares in binary with $\log_2Q$ bits.

(i) \emph{(the low bits are cheap, with a logarithm)} For a rotation $R_z(\varphi)$ implemented as a probabilistic mixture of Clifford+$T$ channels at diamond error $\delta$, Ref.~\cite{Bothe} gives the expected cost $T(\varphi,\delta)=\min\{T_{\mathrm{small}}(\varphi/2,\delta),\,1.52\log_2(1/\delta)-0.01\}$ \cite[Eqs.~(9)--(12), in the convention $R_Z(\theta)=e^{i\theta Z}$]{Bothe}, where $T_{\mathrm{small}}$ is the cost of a mixture whose under-rotation is the identity and the second entry is the angle-independent mixed-diagonal cost of~\cite{KLMPP}. Along the ray $\theta=\Theta(\sqrt\delta)$ that formula gives $T_{\mathrm{small}}=\Theta\big((\theta^2/\delta)\log_2(1/\delta)\big)$, so the logarithm is not a constant there. It is a bound along that ray only, and not a cost law on the region $\theta^2\le\delta$; two things go wrong if it is read as one. At $\theta=\delta/4$ the identity channel is already within $2\sin\theta\le\delta/2$ of the target in diamond norm, so zero $T$ gates suffice while $\theta^2\le\delta$ still holds, and no positive lower bound of that form can hold there. And at $\theta=\delta L_\delta$, writing $L_\delta:=\log_2(1/\delta)$, reading the expression as a law would predict $\Theta(\delta L_\delta^3)$, whereas the formula of~\cite{Bothe} gives $\Theta(\delta L_\delta^2\log L_\delta)$, smaller by a factor $L_\delta/\log L_\delta$. What we use below is only the resulting \emph{upper} bound. On the tuned grid the low $\ell$ bits of a share span angles $\varphi\le2\pi2^\ell/Q\approx8\varepsilon2^\ell$, so shedding them into a mixed program at $\delta=\Theta(\varepsilon)$ costs $O(\varepsilon4^\ell L)$ expected $T$ gates, which is $O(1)$ for
\begin{equation}\label{eq:freezone}
\ell\ \le\ \tfrac12L-\tfrac12\log_2L-O(1),
\end{equation}
so~\eqref{eq:freezone} is a sufficient condition for a free zone, established by a known construction. The same formula grows geometrically over the next $\tfrac12\log_2L$ bits and reaches the angle-independent cost $\Theta(L)$ at $\ell=\tfrac12L-O(1)$. At $\varepsilon=10^{-10}$ the formula gives $2.9$, $10.6$, $28.2$ and $50.5$ expected $T$ gates for $\ell=12,13,14,16$, against $\tfrac12L=16.6$ and $\tfrac12\log_2L=2.5$.

(ii) \emph{(the coarse bits are charged, in total)} By Theorem~\ref{thm:mixing}(a) the $k_\gamma=\tfrac12L-O(1)$ coarse bits of every share are paid \emph{in total} at rate $\ge1-\gamma$ committed $T$ gates per bit per round, beyond an allowance of $\log_236+\rho_m(\bar T_t)/m$ bits per coordinate, and by (b) at rate $\ge2(1-\gamma)$ beyond $A_2'+2=O(\log L)$ bits, given (R).
\end{corollary}
The allowance in (ii) is not slack in the proof: it is the information carried by Cliffords and by powers of $T$, which are free or cost one gate. On a grid with $4\mid Q$ the top bit $h$ of a share $u=hQ/2+v$ is shed for nothing---$R_z(2\pi u/Q)\sim Z^hR_z(2\pi v/Q)$ up to phase, so a process may commit the Clifford $Z^h$ and keep $v$---and with $8\mid Q$ the top three bits go for at most one $T$ gate. Theorem~\ref{thm:mixing} is therefore a statement about the total count and not about individual bit positions: every coarse bit beyond the allowance costs a committed $T$ gate, wherever in the share it sits.

The two boundaries lie $\tfrac12\log_2L+O(1)$ bits apart: \eqref{eq:freezone} frees the bits below $\tfrac12L-\tfrac12\log_2L-O(1)$ by construction, and Theorem~\ref{thm:mixing} charges those above $\tfrac12L-O(1)$. Since the slack of Theorem~\ref{thm:mixing} is itself $O(\log L)$ bits per coordinate, the two agree to within the resolution of the theorem. Which of them is tight we do not determine. Lemma~\ref{lem:nocheap} rules out extending the free zone all the way to $\tfrac12L-O(1)$, but it is stated for a fixed ratio $a=\theta/\sqrt\delta$ and gives no rate of divergence, so it does not exclude a free zone at, say, $\ell=\tfrac12L-\tfrac14\log_2L$, which would correspond to $a(L)=\Theta(L^{-1/4})$. Settling the exact boundary, and the shape of the transition between the two regimes, needs a lower bound quantitative in $a(L)$, which we do not have.

\begin{lemma}[No cheap rotation at angle $\sqrt\delta$]\label{lem:nocheap}
Fix $a>0$ and let $V_\delta:=e^{i\theta_\delta Z}$ with $\theta_\delta=a\sqrt\delta$. If $\mathcal E_\delta=\sum_jp_j\mathcal U_{U_j}$ is a mixture of single-qubit Clifford+$T$ unitaries with $\|\mathcal E_\delta-\mathcal U_{V_\delta}\|_\diamond\le\delta$, then for every integer $M$
\[
\liminf_{\delta\to0}\Pr_{j\sim p}\big[t_{\min}(U_j)>M\big]\ \ge\ \frac{a^2}{a^2+1/2},
\]
so $\liminf_{\delta\to0}\E_{j\sim p}t(U_j)\ge(M+1)a^2/(a^2+\tfrac12)$ and the expected $T$-count diverges as $\delta\to0$.
\end{lemma}
On the tuned grid a low residue of $\ell=\tfrac12L-C$ bits subtends an angle $\varphi\approx8\varepsilon2^\ell=\Theta(\sqrt\varepsilon)$, so Lemma~\ref{lem:nocheap} applies for every fixed $C$: no choice of program makes a residue at angle $\asymp\sqrt\delta$ cheap, and in particular the free zone~\eqref{eq:freezone} cannot be pushed to $\tfrac12L-O(1)$. The lemma says nothing about ratios $a$ that shrink with $L$, so it does not by itself identify the exact position of the boundary.

\begin{proposition}[Mixed passthrough]\label{prop:mixach}
Let a \emph{mixed synthesizer} return, for every grid angle and every diamond budget $\varepsilon'$, a program of expected cost $\tau_{\mathrm{mix}}(\varepsilon')$. Running the fractional-passthrough process of Theorem~\ref{thm:frontier} with programs in place of words, each committed share at diamond budget $2\varepsilon/r$, gives a mixed-unitary CW process with
\[
S\le\lceil q\log_2Q\rceil+O(\log(mQr)),\quad T_t\le(m-q)\tau_{\mathrm{mix}}(2\varepsilon/r)
\]
for $t<r$, correct to coordinatewise diamond error $2\varepsilon$, since diamond errors of composed channels add. Campbell's Theorem~2~\cite{Campbell17} provides the synthesizer: for an axial target $V$ and any $\eta<0.01$, four Clifford+$T$ words $U_1,U_2,ZU_1Z,ZU_2Z$, each obtained by one call to a deterministic synthesizer at operator-norm accuracy $\eta$, can be mixed so that $\tfrac12\|\mathcal E-\mathcal U_V\|_\diamond\le5\eta^2$, i.e.\ $\|\mathcal E-\mathcal U_V\|_\diamond\le10\eta^2$; a similar four-word construction is in~\cite{Hastings17}. Hence $\tau_{\mathrm{mix}}(\varepsilon')\le f_{\mathrm{ax}}(\sqrt{\varepsilon'/10})$ with $f_{\mathrm{ax}}$ the \emph{worst-case} cost of the deterministic synthesizer at that accuracy over axial targets. To turn this into the average $\tfrac32L$ we assume, in place of Corollary~\ref{cor:frontier},
\begin{equation}\label{eq:mixhyp}
\E\,\max\big(t(U_1),t(U_2)\big)\ \le\ 3\log_2(1/\eta)+o(L)
\end{equation}
for the two words the construction actually calls at accuracy $\eta$, the expectation being over the uniform grid share. Under~\eqref{eq:mixhyp} the committed cost per share is
\[
\begin{aligned}
\tau_{\mathrm{mix}}(2\varepsilon/r)\ &\le\ 3\log_2\sqrt{5r/\varepsilon}+o(L)\\
&=\ \tfrac32L+\tfrac32\log_2(5r)+o(L),
\end{aligned}
\]
and the mixed family has slope $\alpha_{\mathrm{mix}}\to\tfrac32$.
\end{proposition}
Hypothesis~\eqref{eq:mixhyp} is stronger than Corollary~\ref{cor:frontier} and is not implied by it. Corollary~\ref{cor:frontier} concerns the mean cost of one synthesis call at a uniform grid angle, whereas the second of Campbell's two words is a target offset from the first by an amount chosen from the first word's error direction; it need not be a grid point, and the cost that enters is the larger of the two calls rather than either mean. We do not prove~\eqref{eq:mixhyp}, and the $\tfrac32$ of this proposition is conditional on it.

Two things this does not say. First, the construction is not a pair of words with exactly antiparallel errors. Such a pair $VR_{\mathbf n}(\pm a)$, whose equal-weight mixture would have diamond error $2\sin^2(a/2)$, need not exist at all: for $V=R_z(2\pi/25)$ it would give $W_++W_-=2\cos(a/2)V$, so the ratio $\zeta_{25}^{-1}$ of the diagonal entries of $V$, a primitive $25$th root of unity of degree $20$ over $\Q$, would have to lie in the degree-$8$ field $\Q(\zeta_{16})$ that contains every entry of every $SU(2)$ representative of a Clifford+$T$ word---a contradiction at every $T$-count, and $Q=25$ is in the tuned family. Campbell's construction instead uses one over- and one under-rotation with weights chosen to cancel the first-order $Z$ component, paying the residual $X,Y$ components at second order. Second, the average cost of a single deterministic synthesis is not the cost of a mixed program. On the $20\,000$ grid angles of Fig.~\ref{fig:frontier}(b), \texttt{gridsynth} at $\eta=7.07\times10^{-6}$ gives $\E_u\tau=53.43\pm0.02$ (s.d.\ $2.20$, range $37$--$65$), which is a single-word measurement at that accuracy and nothing more. The budget $2\varepsilon/r$ requires $\eta=\sqrt{\varepsilon/(5r)}=3.16\times10^{-6}$ at $\varepsilon=10^{-10}$, $r=2$, and the four-word program costs the larger of two synthesis calls there, about $3\log_2(1/\eta)+2.7\approx57.5$ by Ross--Selinger scaling. Certifying that number means synthesizing both words at $\eta$ and verifying the channel error; we report it as an estimate. For comparison the angle-independent mixed-diagonal formula of~\cite{KLMPP,Bothe} gives $50.5$ at $\delta=10^{-10}$, and both are about half the deterministic $105.34$. Computed by \texttt{scripts/mixing\_bounds.py} and \texttt{scripts/run\_mixed\_achievable.py}.

\paragraph{The exchange curve under mixing.} Corollary~\ref{cor:mixstop} and Proposition~\ref{prop:mixach} together describe the achievable trade. Committing the low $\ell\le\tfrac12L-\tfrac12\log_2L-O(1)$ bits of a share costs $O(1)$; over the next $\tfrac12\log_2L$ bits the cost of that construction grows geometrically to $\Theta(L)$, and by $\ell=\tfrac12L-O(1)$ Lemma~\ref{lem:nocheap} shows no program at all is cheap, though it does not pin down where in this window the transition actually sits; from $\ell=\tfrac12L-O(1)$ on, committing needs a program for an angle $\gtrsim\sqrt\varepsilon$ whose cost is the angle-independent $\tfrac32L$ however many bits it carries~\cite{KLMPP,Bothe}. The \emph{achievable} curve---the cost of the constructions above, not a determination of the optimum---is thus flat at $O(1)$ on the low $\tfrac12L-\tfrac12\log_2L$ bits, climbs over the next $\tfrac12\log_2L$, and is flat again at $\tfrac32L$, except for the top $\log_2\gcd(Q,8)$ bits, which are Clifford or $T$-power information and are free wherever they sit (Corollary~\ref{cor:mixstop}(ii)). Averaged over the share this is $\alpha_{\mathrm{mix}}=\tfrac32$ per bit \emph{of the share}---not per coarse bit, the unit in which Theorem~\ref{thm:mixing} charges---but the memory worth keeping is the \emph{high} part: remembering the top $\tfrac12L+O(\log L)$ bits of a coordinate and committing the rest buys that memory for $O(1)$ committed $T$ gates, whereas remembering the low half saves nothing. Theorem~\ref{thm:mixing} bounds \emph{every} program from below by a line of the same shape---nothing charged on the low part, slope $\ge1-\gamma$ unconditionally and $\ge2(1-\gamma)$ given (R) on the coarse bits beyond the $O(\log L)$ allowance---and under Conjecture~\ref{conj:H} the bound itself acquires a jump of $L-O(1)$ at the first coarse bit carrying more than $O(\log L)$ bits, then slope $3$. Bound and construction agree at the top of the share and at $S=0$; where between the two ends the true cost lies we do not determine.

\begin{lemma}[Grid cover of a tube]\label{lem:cover}
Let $0<\varepsilon\le1/16$, let $Q_1:=\max\{Q:\sin(\pi/2Q)\ge4\varepsilon\}$ be the tuned modulus at accuracy $2\varepsilon$, and for $i=0,\dots,8$ let $\mathcal G_i$ be the set of words within $\dproj$-distance $2\varepsilon$ of some rotation $G_1R_z\big(2\pi(d+i/9)/Q_1\big)G_2$, $d\in\Z_{Q_1}$. Then $T_\varepsilon(G_1R_zG_2)\subseteq\bigcup_{i=0}^8\mathcal G_i$, and consequently, under Conjecture~\ref{conj:H}, for every $\tau\ge0$
\begin{equation}\label{eq:cover}
\begin{aligned}
&\#\{W\in T_\varepsilon(G_1R_zG_2):t_{\min}(W)\le\tau\}\\
&\qquad\le9c_0+9c_1\tau+36c\,2^\tau\varepsilon^2.
\end{aligned}
\end{equation}
\end{lemma}
Lemma~\ref{lem:cover} is what replaces a ``tube form'' of Conjecture~\ref{conj:H}: the grid form~\eqref{eq:H} is asserted for all frames and all moduli $Q\le Q_\varepsilon$, so shifting the frame by one ninth of a grid step nine times covers the tube at twice the radius. The cumulative Haar prediction for the tube coefficient is $72$, twice the shell constant $36$ of Sec.~\ref{sec:num}, against which the coefficient $36c$ of~\eqref{eq:cover}---$1728$ at the illustrative $c=48$---is loose by a factor $24$, i.e.\ by $4.6$ bits in the fee; it is rigorous given~\eqref{eq:H} and needs no further hypothesis.

\begin{corollary}[Rate three and quantization under mixing]\label{cor:mixq}
Assume Conjecture~\ref{conj:H} and let the process be mixed-unitary and correct with probability $\ge1-\delta$. Fix $\gamma$ with $\varepsilon':=36\sqrt{2\varepsilon/\gamma}\le2^{-6}$ and $\varepsilon'\le(72c)^{-1/2}$, the latter being what makes $b'\ge0$ below, and put $L':=\log_2(1/\varepsilon')$, $\tau_0':=\lfloor2L'-\log_2(36c)\rfloor$, $N_0':=9c_0+9c_1\tau_0'+1$, $a_0':=1+\log_2N_0'$, $b':=2L'-\log_2(36c)-1$, $\bar k:=\log_2K'$ with $K'$ as in Theorem~\ref{thm:mixing}, $\kappa':=1+b'/\bar k$, and
\[
D_t':=(1-\delta)m(k'_\gamma-1)-2mh_2(\delta)-2\delta m\log_2K'-S.
\]
Then for every round $t<r$,
\begin{equation}\label{eq:mixrate}
\bar T_t\ \ge\ (1-\gamma)\Big[\kappa'\big(D_t'-ma_0'\big)-3m\log_2\Big(1+\frac{\bar T_t}{(1-\gamma)m}\Big)\Big],
\end{equation}
and, writing $\bar{\mathcal A}'^{>}_t$ for the expected number of coordinates that are correct and whose representative round-$t$ word has minimal $T$-count above $\tau_0'$,
\begin{equation}\label{eq:mixactive}
\bar{\mathcal A}'^{>}_t\ \ge\ \frac{D_t'-ma_0'}{\bar k}.
\end{equation}
\end{corollary}
Since $2L'=L-\log_2(1/\gamma)-2\log_2(36\sqrt2)=L-\log_2(1/\gamma)-11.34$, one has $\tau_0'=\lfloor L-\log_2(1/\gamma)-\log_2(36c)-11.34\rfloor$ and, for every fixed $c$, $\kappa'\to1+\beta=3$; hence $\bar\Tpre\ge\tfrac32(r-1)(m\log_2K-2S)(1-o(1))-O(rm\log L)$. This meets Proposition~\ref{prop:mixach} in its leading term at $S=0$, both giving $\tfrac32(r-1)mL$, but not in its dependence on $S$: the lower bound falls with slope $3(r-1)$ in $S$, the whole-coordinate passthrough family only with $\tfrac32(r-1)$. The two rates are per different units---the bound charges $2$ per coarse bit given (R), of which a share has $\tfrac12L$, while the family trades whole shares of $\log_2Q$ bits against $\tfrac32L$ gates---and they are not to be quoted as one exchange rate $\alpha$ away from $S=0$. Closing the gap would need an achievable family that keeps the high bits of some coordinates and sheds their low bits cheaply, which we do not construct. A representative of cost at most $\tau_0'=L-O(1)$ carries at most $\log_2N_0'=O(\log L)$ bits about its coordinate: the fee is $L-\log_2(1/\gamma)-\log_2(36c)-O(1)$, one half of the deterministic $2L-\log_2c-O(1)$. The constant of Conjecture~\ref{conj:H} again enters only through $\log_2c$. At $\varepsilon=10^{-10}$ both~\eqref{eq:mixrate} and~\eqref{eq:mixactive} are vacuous, the radius $\varepsilon'=2^{-9.9}$ leaving only $L'=9.9$; \eqref{eq:mixactive} becomes informative from $L\approx45$, where it forces $24\%$ of the coordinates to commit a representative of cost above $\tau_0'=20$, and~\eqref{eq:mixrate} from $L\approx37$. Under the bounded-support hypothesis of Remark~\ref{rem:mixconst} the radius improves to $\sqrt2\eta$ and the chemistry-scale figures $\tau_0'=19$, $27.5\%$ are recovered.

\begin{proposition}[Measurement-adaptive protocols]\label{prop:adaptive}
Let a process be as in Definition~\ref{def:mixed} except that the round-$t$ program of coordinate $j$ is an adaptive circuit on $n$ qubits---the coordinate's qubit and $n-1$ ancillas prepared in $|0\rangle$ and measured or reset before the next program---built from Clifford gates, $T$ gates and single-qubit computational-basis measurements, in which each gate may depend on the outcomes obtained earlier in the same program, and such that on every outcome branch of nonzero probability the operation induced on the coordinate's qubit is unitary. Repeat-until-success and fallback synthesis~\cite{BRS,KLMPP} are of this form. Let $T_t$ and $M_t$ be the expected numbers of $T$ gates and of measurements executed in round $t$, and $|\mathcal C_n|$ the order of the $n$-qubit Clifford group modulo phase. Then, with $k_\gamma$ and $\rho_m$ as above,
\begin{multline}\label{eq:adaptive}
S+(2n+2)\frac{\bar T_t+\bar M_t}{1-\gamma}+m\log_2|\mathcal C_n|+\rho_m\Big(\frac{\bar T_t+\bar M_t}{1-\gamma}\Big)\\ \ge\ (1-\delta)mk_\gamma-h_2(\delta).
\end{multline}
\end{proposition}
Asymptotically $\bar T_t+\bar M_t\ge(1-\gamma)(m\log_2K-2S)/(4n+4)\cdot(1-o(1))-O(m\log L)$: with one ancilla ($n=2$) a share that is not held costs at least $L/12$ committed $T$ gates plus measurements per round, against $0.53\log_2(1/\delta)+4.86\approx0.53L$ achievable by mixed fallback~\cite[Eq.~(13)]{KLMPP,Bothe}. The constant is poor because~\eqref{eq:adaptive} counts transcripts by a crude token bound and not by a normal form; at $n=1$ it gives four tokens per bit where Theorem~\ref{thm:mixing}(a) gives one. What~\eqref{eq:adaptive} settles is the order: bounded-ancilla measurement adaptivity does not remove the floor, it lowers the rate by a factor that depends only on $n$.

\paragraph{What the results say together.} In the deterministic CW model the rigorous floor per share is $2L$ with explicit constants (Theorem~\ref{thm:rate2}) and $(\tfrac52-o(1))L$ asymptotically (Theorem~\ref{thm:rate52}), the conjectural floor and fee are $3L$ and $2L$ (Theorem~\ref{thm:main2}), and the achievable cost is $3L$ (Theorem~\ref{thm:frontier}). Under mixing the same three statements read $L$ (Theorem~\ref{thm:mixing}(b)), $\tfrac32L$ and $L$ (Corollary~\ref{cor:mixq}, under the same Conjecture~\ref{conj:H}) and $\tfrac32L$ (Proposition~\ref{prop:mixach} under~\eqref{eq:mixhyp}; $1.52L$ measured in~\cite{KLMPP}), all of them asymptotic in $L$---the two lower bounds vacuous at $\varepsilon=10^{-10}$, the two achievable costs conditional on a synthesis hypothesis, Corollary~\ref{cor:frontier} in one model and~\eqref{eq:mixhyp} in the other. Mixing rescales the accuracy exponent by $\tfrac12$, and with it every quantity linear in $L$, leaving the ratio rigorous-with-explicit-constants\,:\,conjectural\,:\,achievable at $2:3:3$ in both models; it does not touch the exchange rate.

Remark~\ref{rem:mixconst} in Appendix~\ref{app:mixing} evaluates the constants. Theorem~\ref{thm:mixing} is asymptotic in the same sense as Theorems~\ref{thm:main1} and~\ref{thm:rate2}, and its additive terms are larger, because $\log_236$ and $A_2'$ are now paid against $\tfrac12L$ bits instead of $L$: at $\varepsilon=10^{-10}$, $\gamma=1/4$ and $S=0$, \eqref{eq:mixone} gives $\bar T_t\ge4.1\,m$ and~\eqref{eq:mixtwo} is vacuous, becoming positive only at $L\ge56$ ($L\ge45$ under the bounded-support hypothesis).

\section{Discussion}\label{sec:disc}

In the model where rotations are synthesized coordinate by coordinate without ancillas, a bit of aggregate phase information that a schedule does not carry across a round boundary costs at least two committed magic states with explicit constants, and at least $\tfrac52-o(1)$ asymptotically ($3-o(1)$ when every committed segment is $\varepsilon$-close to a Clifford-framed rotation, Theorem~\ref{thm:rate3seg}), and three suffice along the passthrough family whenever the grid-mean synthesis cost is $(3+o(1))L$---asymptotically in $L$, and after the first $O(\log L)$ bits of each coordinate. The two constants are not the same kind of number. Three is a volume count: a rotation coset has codimension $\beta=2$ in $\SU$, so a word pays $\beta L$ gates to enter the tube around it and one gate per bit thereafter, and $1+\beta$ is also the Ross--Selinger cost of a single rotation. Two is what the spectral method yields: the Ramanujan bound for the octahedral golden gates controls the tube only up to an error $2^{\tau/2}$, the square root of the number of all words of that cost, and a profile $2^{\tau/2}$ is worth $1+\beta/2$ gates per bit. Five halves is what elementary geometry of numbers yields without the spectral method, below its square-root barrier: eleven fifths from one determinant per box, seventeen sevenths once the heights of the hyperplanes are used, and every rate below $\tfrac52$ once the numerators are fibred over rational projections. It is the exact limit of that method as set up here, with one auxiliary hyperplane per box, because $\tfrac52$ is where a box of the five-point determinant has shrunk to an $\varepsilon$-cube, and going beyond it means showing that most $\varepsilon$-cubes along a tube are empty. Two further observations, both heuristic, locate the difficulty. Arguments that cover a tube by balls cannot pass $\tfrac52$ (Remark~\ref{rem:what52}): near rational points of small height there are sphere sections of tiny radius whose words crowd into a single ball. And, as far as we can see, the square-root inputs that are standard in this area---the Ramanujan bound, or a twisted Linnik conjecture of the kind that gives optimal strong approximation on $S^3$~\cite{BKS}---stop at rate $2$ in their natural circle-method and spectral implementations. The gap between two and three is the square-root remainder of the spectral estimate, which Theorems~\ref{thm:elem}, \ref{thm:dioph} and~\ref{thm:rate52} close only in part; whether the true rate is $3$ is not decided here, and Conjecture~\ref{conj:H}---which the enumeration supports on Haar-random cosets, on arithmetic cosets, and on the Clifford cosets, where Theorem~\ref{thm:clifford} proves it up to $2^{O(\tau/\log\tau)}$---is the statement that the arithmetic words do not know the difference.

Three consequences are worth stating plainly. First, the corner $S=0$ is where the law binds without anyone choosing to be memoryless. When the rotation angles are produced on the quantum side---qubitization, a QROM-loaded phase table, adaptive QSVT phases---no classical snapshot of them exists to be traded away. If a reduction of that setting to the classical-input model existed, Theorem~\ref{thm:rate2} would bind there with nothing subtracted for memory; none is given here, and the theorem is not to be applied to it: Fano's inequality is applied to a classical stream, whereas with quantum coefficients the object crossing the cut is a state and not a message. That corner is the natural target for this line of work rather than a consequence of it. Extending the model to quantum coefficients is left to future work. Second, memory cannot be traded for magic in small units inside a coordinate, and which units are available is a property of the synthesis mechanism, not of magic. With deterministic words the trade is quantized in whole rotations past the $O(\log L)$ bits that the eight $T$-powers and the powers of a cheap infinite-order word can buy (Remark~\ref{rem:firstbit}, Proposition~\ref{prop:cluster}), and Ross--Selinger synthesis at $\varepsilon=10^{-10}$ puts the penalty for shedding low bits at between $4.5$ and $20$ times the passthrough slope. With probabilistic mixing the low $\tfrac12L-\tfrac12\log_2L-O(1)$ bits of every share are free~\cite{Bothe}, the bits at $\tfrac12L-O(1)$ are cheap for no mixed program at all (Lemma~\ref{lem:nocheap}, which bounds the free zone without locating its edge), the coarse bits are charged in total at the same rate as before (Theorem~\ref{thm:mixing}, Corollary~\ref{cor:mixstop}), and under Conjecture~\ref{conj:H} they are again quantized in whole rotations, now with fee $L-O(1)$ (Corollary~\ref{cor:mixq}): the rule ``shed whole coordinates, never low bits'' becomes ``shed the low $\tfrac12L-\tfrac12\log_2L$ bits of every coordinate, and above that shed whole coordinates''. The charge is on the total and not on bit positions: a Clifford or a power of $T$ carries the top $\log_2\gcd(Q,8)$ bits for free wherever they sit. Measurement-adaptive synthesis with a bounded number of ancillas lowers the rate by a factor that depends only on the qubit count and does not remove the floor (Proposition~\ref{prop:adaptive}); there the floor is on the token count, $T$ gates together with measurements, and we do not separate the two. Only the phase-gradient register removes the premium altogether, because there the total cost of a coordinate is split-invariant. With clean ancillas and batched table lookups it removes the law itself, at $O(L/\log L)$ committed $T$ gates per share (Proposition~\ref{prop:batch}). For ancilla-free pipelines, however, a resource estimate that reads the final $T$-count of a mixed pipeline must still book at least $L-O(\log L)$ committed $T$ gates for every share that is not held---one half of what it books for a deterministic pipeline, and not zero. Third, side information enters at the level of a conditional entropy: a process that holds the seed of the randomization faces no information obstruction, and one that misses an injective $\lambda$-bit seed independent of the aggregate is charged for at most $\lambda$ bits rather than $m\log_2K$---and, on the two-round streams with uniform aggregate where the matching lower bound holds, for at least $\lambda-m\log_2(Q/K)$ (Corollary~\ref{cor:side}(c)). That is the regime of randomized compiling, where $\lambda$ is a handful of bits. Over more rounds the upper bound is all that survives: a seed spent in one round and cancelled in a later one leaves nothing to charge for.

Two unconditional statements bound the arithmetic question from the other side: Selinger's worst-case $T$-count $4L-9$ for a specific $z$-rotation~\cite[Thm.~30]{Selinger15}, and the covering exponents of Parzanchevski and Sarnak, who show that some single-qubit unitaries need $T$-count $4L$ while all are reachable at $6L$~\cite{PS18}. The committed-magic frontier is a statement about the typical rather than the worst case, and it sits between the rigorous and the conjectural covering exponents in the same way. Theorems~\ref{thm:elem}, \ref{thm:dioph} and~\ref{thm:rate52} show that the square-root barrier can be crossed for these lattice points near a great circle of $S^3$, for $T$-counts up to $(\tfrac52-o(1))L$. Closing the gap between $\tfrac52$ and $1+\beta$ for committed words with arbitrary frames is the remaining problem, and we leave it open; by Proposition~\ref{prop:dense} it requires at least excluding tubes that carry words of $T$-count $(\tfrac52+o(1))L$ with gaps of at most $2^{o(L)}$ grid steps along their core.

\begin{acknowledgments}
J.Y.'s deepest thanks go to two teachers who are also coauthors of this
paper, Yangyang Li and Xiu-Hao Deng. Their help was immense; without them
this paper would not exist.

J.Y. thanks Ke Lin of Xidian University for opening the door to research
and for the recommendation that led him to join Prof.~Li's group.

J.Y. also thanks the teachers who were willing to offer guidance along the
way, and the senior members of Prof.~Li's group, who have been generous
with their help throughout this work.

Thanks are due, finally, to everyone who supported this paper, too many to
name individually here.

Y.L. acknowledges support from the National Natural Science Foundation of
China under Grant No.~62476209, the Key Research and Development Program of
Shaanxi under Grant No.~2024CY2-GJHX-18, and the Fundamental Research Funds
for the Central Universities under Grant No.~QTZX26097. X.-H.D.
acknowledges support from the Science, Technology and Innovation Commission of Shenzhen
Municipality (Grant Nos.~JCYJ20170412152620376 and KYTDPT20181011104202253),
the Innovation Program for Quantum Science and Technology (Grant
No.~2021ZD0301703), the Guangdong Major Project of Basic Research (Grant
No.~2025B0303000007), and the Shenzhen Science and Technology Program
(Grant No.~KQTD20200820113010023).

Generative AI tools were used to assist with manuscript drafting and
revision, code and repository maintenance, literature discovery, and
consistency checks. The authors directed their use through task-specific
prompts and independently reviewed the generated text, citations, code, and
numerical claims against the cited literature, source files, executable tests,
and the exact rational certificates reported in this work. AI output was
not treated as a scientific source or as independent proof verification. The
authors take full responsibility for the manuscript, proofs, code, data,
figures, citations, and conclusions.
\end{acknowledgments}

\section*{Author contributions}
J.Y. developed the theory and proofs, implemented the enumeration, synthesis and certificate code, and performed
the numerical computations and analyses. Y.L. supervised the research, reviewed and revised the manuscript,
and provided project administration and funding support. X.-H.D.
independently validated the theoretical derivations, numerical results, and
computational artifacts. All authors wrote the manuscript together,
discussed the results, and approved the final version.

\section*{Data availability}
{\raggedright
The exhaustive enumeration of Clifford+$T$ words to $T$-count $22$, the coset counts of Table~\ref{tab:counts} and Fig.~\ref{fig:tube}, the exact-optimal and \texttt{gridsynth} data of Fig.~\ref{fig:frontier}, and every script that produces a number or figure of this paper are available at \url{https://github.com/OIerYangJZ/memory-magic-exchange-law}. The scripts are in \texttt{scripts/}, among them \texttt{enum2.py} and \texttt{enum\_arith.py} (the enumeration and coset counts), \texttt{gridsynth\_run.py}, \texttt{fig1\_fig.py} and \texttt{make\_fig2.py} (the figures), \texttt{thm\_numbers.py} (Table~\ref{tab:numbers}), \texttt{mixing\_bounds.py} (the constants of Sec.~\ref{sec:mixing}) and \texttt{scan2.py}, \texttt{scan2\_report.py} and \texttt{astra1.py} (the frame scans of Sec.~\ref{sec:num}). The exact rational certificates of Appendices~\ref{app:elem} and~\ref{app:dioph} are \texttt{research/cert/certificate.py}, \texttt{research/cert/certificate\_223.py} and \texttt{research/dioph/cert\_dioph.py}. The exact certificate of Appendix~\ref{app:proj} is \texttt{research/annulus/cert\_proj.py}, and the continuum model and SMT check of Appendix~\ref{app:rate52} are \texttt{research/uniform52/cont\_model.py} and \texttt{research/uniform52/verify\_z3.py}, with the solver output in \texttt{research/uniform52/z3\_outputs.txt}. The example of Remark~\ref{rem:what52} is reproduced by \texttt{research/conditional3/ballpts.cpp} (\texttt{ballpts 25 5.828427 1 1 -1 0.0025}) and \texttt{research/conditional3/review/check\_sections.py}. The repository README lists, for every table, figure and theorem, the script and data file that produce it.\par}

\appendix

\section{The streaming model}\label{app:model}

This appendix states the four items of~\cite{Paper1} that Secs.~\ref{sec:setting}--\ref{sec:side} use, in the form in which they are used, with proofs, so that the present paper can be refereed on its own. The definitions are those of~\cite[Def.~9, 13, 20, 21]{Paper1} and the two statements are~\cite[Lem.~1, Prop.~10]{Paper1}; no other result of~\cite{Paper1} is used anywhere above.

\begin{definition}[Charged one-pass process]\label{def:model}
A \emph{charged one-pass process} on the stream $a^{(1)},\dots,a^{(r)}$ is an algorithm with a random tape $\rho$ and a write-once output tape such that
\begin{enumerate}\setlength\itemsep{0pt}
\item[(i)] it reads the stream once, in order, and may keep unbounded internal state while reading a round;
\item[(ii)] at each \emph{cut} $c$ between consecutive rounds it emits a \emph{snapshot} $\sigma_c$, a finite bit string, with the \emph{restart} property: everything the process writes after $c$ is a function of $\sigma_c$ and of the part of the stream after $c$. The tape $\rho$ is part of $\sigma_c$, so a process that wishes to reuse its randomness pays for it;
\item[(iii)] snapshots are serialized in a prefix-free code, whence $H(\sigma_c)\le\E|\sigma_c|$;
\item[(iv)] its \emph{charge} is $S:=\max_{\mathrm{input}}\E_\rho\max_cB_{\mathrm{cut}}(c)$, where $B_{\mathrm{cut}}(c):=|\sigma_c|$.
\end{enumerate}
For $p>1$ passes the process may re-read the stream $p$ times, and the object playing the role of $\sigma_c$ is the transcript of all $2p-1$ crossings of the cut.
\end{definition}
Condition (iii) is free: prefixing a snapshot with its own length in Elias-$\gamma$ makes any serialization prefix-free at a cost of $2\log_2(|\sigma_c|+1)+1$ bits, which is dominated by the $O(\log(mQr))$ terms of every construction below. Condition (ii) is the only substantive one, and it is what an append-only instruction stream enforces: the process may not revise what it has already emitted, and what it has not written down it must remember.

\begin{definition}[Dispersed family, canonical sharing, tuned modulus]\label{def:family}
With $Q,\Delta,K,m,V,P_j,U_x$ as in Sec.~\ref{sec:setting}, an \emph{$r$-round dispersed stream} for the aggregate $x\in[K]^m$ is any $a^{(1)},\dots,a^{(r)}\in\Z_Q^m$ with $\sum_ta^{(t)}\equiv x\pmod Q$; correctness is required on every such stream, not only on typical ones. The \emph{canonical sharing} $\mathcal D$ draws $a^{(1)},\dots,a^{(r-1)}$ i.i.d.\ uniform on $\Z_Q^m$ and sets $a^{(r)}:=x-\sum_{t<r}a^{(t)}$. The \emph{tuned} modulus at accuracy $\varepsilon$ is $Q_\varepsilon:=\lfloor\pi/(2\arcsin2\varepsilon)\rfloor$.
\end{definition}

\begin{lemma}[Tuned modulus]\label{lem:tuned}
For $0<\varepsilon\le1/8$ one has $Q_\varepsilon\ge6$, $\sin(\pi/2Q_\varepsilon)\ge2\varepsilon$, and $\log_2K_\varepsilon=L-\log_2(4/\pi)+o(1)=L-0.35+o(1)$.
\end{lemma}
\begin{proof}
$Q_\varepsilon\le\pi/(2\arcsin2\varepsilon)$ gives $\pi/2Q_\varepsilon\ge\arcsin2\varepsilon$, i.e.\ $\sin(\pi/2Q_\varepsilon)\ge2\varepsilon$, which is the accuracy condition of Sec.~\ref{sec:setting} with a factor $2$ to spare; $\varepsilon\le1/8$ gives $\arcsin2\varepsilon\le\arcsin\tfrac14<\pi/12$ and so $Q_\varepsilon\ge6$. Since $\arcsin2\varepsilon=2\varepsilon(1+O(\varepsilon^2))$, $Q_\varepsilon=\pi/(4\varepsilon)+O(1)$ and $\log_2K_\varepsilon=\log_2(Q_\varepsilon-1)=L+\log_2(\pi/4)+o(1)$.
\end{proof}

\begin{lemma}[Projective distance]\label{lem:dproj}
For unitaries $A,B$ of equal size, $\dproj(A,B)=2\sin(\Theta/4)$, where $\Theta\in[0,2\pi]$ is the length of the shortest closed arc of the unit circle containing $\spec(B^\dagger A)$.
\end{lemma}
\begin{proof}
By unitary invariance of the operator norm, $\dproj(A,B)=\inf_\phi\|B^\dagger A-e^{i\phi}I\|=\inf_\phi\max_{\lambda\in\spec(B^\dagger A)}|\lambda-e^{i\phi}|$, since $B^\dagger A$ is normal. Write $d(\phi,\lambda)\in[0,\pi]$ for angular distance, so that $|\lambda-e^{i\phi}|=2\sin(d(\phi,\lambda)/2)$, increasing in $d$. If $e^{i\phi_0}$ is the midpoint of a shortest arc, every eigenvalue is within $\Theta/2$ of it and the two ends attain that, so the maximum equals $2\sin(\Theta/4)$. Conversely, for any $\phi$ put $D:=\max_\lambda d(\phi,\lambda)\le\pi$; the spectrum then lies in the closed arc of length $2D$ centred at $e^{i\phi}$, so $\Theta\le2D$ by minimality and $\max_\lambda|\lambda-e^{i\phi}|=2\sin(D/2)\ge2\sin(\Theta/4)$.
\end{proof}

\begin{proposition}[Block compiler]\label{prop:block}
For every $q\le m$ there is a deterministic one-pass process which, on every $r$-round dispersed stream, keeps the running sums modulo $Q$ of $q$ designated coordinates in a snapshot of at most $\lceil q\log_2Q\rceil+O(\log(mQr))$ bits and, at round $r$, knows the exact aggregate $x_j$ of each designated coordinate.
\end{proposition}
\begin{proof}
The state carried across a cut is the vector of partial sums in $\Z_Q^q$, encoded as one integer in $[Q^q]$, i.e.\ $\lceil q\log_2Q\rceil$ bits; the remainder of the snapshot is the round index, the position within the round and a fixed-size program counter, $O(\log(mQr))$ bits, made prefix-free as in Definition~\ref{def:model}\,(iii). The designated set is fixed in advance and need not be stored. Each arriving share is added modulo $Q$ into its partial sum, so after round $r$ the $j$-th sum is $\sum_ta^{(t)}_j\bmod Q=x_j$.
\end{proof}
The gap between $\log_2Q$ here and the $\log_2K$ of the lower bounds is the difference between storing a share, which ranges over $\Z_Q$, and the entropy of an aggregate, which ranges over $[K]$; it is $\log_2(Q/K)=O(1/Q)$ per stored coordinate.

\section{Proofs of the lemmas}\label{app:lemmas}

\begin{proof}[Proof of Lemma~\ref{lem:loc}]
By Lemma~\ref{lem:dproj}, $\dproj(A,B)=2\sin(\Theta/4)$ where $\Theta\in[0,2\pi]$ is the length of the shortest closed arc containing $\spec(B^\dagger A)$. The spectrum of $\bigotimes_jB_j^\dagger A_j$ consists of the products of eigenvalues. Fix $j$ and two eigenvalues $\lambda,\lambda'$ of $B_j^\dagger A_j$ at circular distance $\theta_j\le\pi$; multiplying both by a fixed product $\mu$ of eigenvalues of the other factors gives two points of the joint spectrum at circular distance $\theta_j$, and any closed arc containing both has length at least $\theta_j$. Hence $\Theta\ge\theta_j$, and since $\sin$ increases on $[0,\pi/2]$, $\dproj(A,B)\ge2\sin(\theta_j/4)=\dproj(A_j,B_j)$. For the right inequality replace one factor at a time: $\|\bigotimes A_j-e^{i\sum\phi_j}\bigotimes B_j\|\le\sum_j\|A_j-e^{i\phi_j}B_j\|$, then take infima. For the closed form, with eigenphases $a,b$ of $B^\dagger A$ one has $|\tr|^2=2+2\cos\Theta$ and $2\sin(\Theta/4)=\sqrt{2-2\cos(\Theta/2)}=\sqrt{2-|\tr|}$, using $\cos(\Theta/2)\ge0$ for $\Theta\le\pi$.
\end{proof}

\begin{proof}[Proof of Lemma~\ref{lem:MA}]
Every Clifford+$T$ unitary has a unique normal form $(T|\epsilon)(HT|SHT)^*C$ with $C$ one of the $24$ Cliffords modulo phase~\cite{MA08}, and the normal form is $T$-optimal~\cite{KMM13}. Strings with exactly $t$ letters $T$ number $2^{t-1}$ (initial $T$ present) plus $2^t$ (absent), times $24$. For tuples, $\prod_jn_1(t_j)\le36^m2^{\sum_jt_j}$ and the number of $(t_j)$ with $\sum t_j\le T$ is $\binom{T+m}{m}$.
\end{proof}

\begin{proof}[Proof of Lemma~\ref{lem:ent}]
Encode $M$ by the Elias-$\gamma$ code of $T$ (at most $2\log_2(T+1)+1$ bits) followed by an index among $N^{\le}_m(T)$ tuples; $H(M)$ is at most the expected length of this prefix code. Use $\log_2\binom{T+m}{m}\le m\log_2(e(1+T/m))$, which is concave in $T$, and Jensen.
\end{proof}

\begin{proof}[Proof of Lemma~\ref{lem:identity}]
Write the unitary as $2^{-k/2}\begin{pmatrix}u&-t^\dagger\omega^l\\ t&u^\dagger\omega^l\end{pmatrix}$ with $u,t\in\Z[\omega]$, $\omega=e^{i\pi/4}$, $uu^\dagger+tt^\dagger=2^k$~\cite{KMM13}. By the closed form of Lemma~\ref{lem:loc}, closeness to a diagonal unitary forces $|t|^2\le2\varepsilon^22^k$. The Galois conjugate $\sqrt2\mapsto-\sqrt2$ of the norm equation gives $|t^\bullet|^2\le2^k$. Hence $\xi:=tt^\dagger\in\Z[\sqrt2]$ is totally non-negative with $N_{\Q(\sqrt2)/\Q}(\xi)=\xi\xi^\bullet\le2\varepsilon^24^k<1$, which forces $\xi=0$ and $t=0$. Then $uu^\dagger=2^k$; the prime above $2$ in $\Z[\omega]$ is $(1-\omega)$ with $(1-\omega)^4\sim2$, and the units of relative norm one are the $\omega^j$, so $u=\omega^j(1-\omega)^{2k}$ times a unit and the unitary is $\operatorname{diag}(\omega^a,\omega^b)$ modulo phase.
\end{proof}

\section{Proof of Theorem~\ref{thm:tube}}\label{app:tube}

\emph{Hopf reduction.} Lift $W$, $G_1$, $G_2$ to $\SU$; $\dproj$ does not see the phase. Put $W':=G_1^\dagger WG_2^\dagger=\begin{pmatrix}a&-\bar b\\ b&\bar a\end{pmatrix}$. Then $\tr(R_z(\theta)^\dagger W')=2\operatorname{Re}(e^{i\theta/2}a)$, so by Lemma~\ref{lem:loc}, $\min_\theta\dproj(W',R_z(\theta))=\sqrt{2-2|a|}$ and
\begin{multline}\label{eq:btube}
\exists\theta:\ \dproj(W,G_1R_z(\theta)G_2)\le\varepsilon\\
\iff\ |b(W')|^2\le\varepsilon^2\big(1-\tfrac{\varepsilon^2}4\big).
\end{multline}
Let $H_1:=G_1S^1G_1^{-1}$ be the maximal torus containing $G_1R_zG_1^{-1}$ and $\pi:\PU\to H_1\backslash\PU\cong S^2$ the Hopf map, with the measure on $S^2$ normalized to $1$. The function $W\mapsto|b(G_1^\dagger WG_2^\dagger)|$ is left-$H_1$-invariant, hence a function of $\pi(W)$, and its sublevel sets are the spherical caps centred at $\pi(G_1G_2)$, the point through which the fibre is the coset $C$. Since $|b|^2$ is uniform on $[0,1]$ under Haar measure,
\begin{equation}
T_\varepsilon(C)\subseteq\{W:|b(G_1^\dagger WG_2^\dagger)|^2\le\varepsilon^2\}=\pi^{-1}(\mathrm{cap}),
\end{equation}
a cap of measure $\varepsilon^2$, i.e.\ of angular radius $\theta_\varepsilon$ with $\sin^2(\theta_\varepsilon/2)=\varepsilon^2$.

\emph{Majorant.} Let $r:=1-2\varepsilon\in[\tfrac34,1)$ and let
\begin{equation}
P_r(u)=\sum_{\ell\ge0}(2\ell+1)r^\ell P_\ell(u)=\frac{1-r^2}{(1-2ru+r^2)^{3/2}}
\end{equation}
be the Poisson kernel, a positive zonal function on $S^2$ of mean $1$, decreasing in the polar angle. Because $1-2r\cos\theta_\varepsilon+r^2=(1-r)^2+4r\sin^2(\theta_\varepsilon/2)=4\varepsilon^2(1+r)$ and $1-r^2=2\varepsilon(1+r)$, its value at the rim of the cap is $P_r(\cos\theta_\varepsilon)=1/(4\varepsilon^2\sqrt{1+r})$. Set
\begin{equation}
F_+:=a_\varepsilon\,(P_r\circ\langle\cdot,x_0\rangle)\circ\pi,\qquad a_\varepsilon:=4\varepsilon^2\sqrt{1+r}\le4\sqrt2\,\varepsilon^2,
\end{equation}
with $x_0$ the centre of the cap. Then $F_+\ge0$ everywhere, $F_+\ge1$ on $\pi^{-1}(\mathrm{cap})$, $F_+$ is left-$H_1$-invariant, and $\mu(F_+)=a_\varepsilon$. No smoothing or truncation is needed: the coefficients $r^\ell$ decay geometrically at scale $\ell\sim1/\varepsilon$, which is exactly the spectral width a cap of radius $\varepsilon$ forces (Remark~\ref{rem:barrier}), and they are explicit.

\emph{Hecke reduction.} $\Lambda_k$ is closed under inversion ($T$-count is symmetric). Write $\Lambda_k=\bigsqcup_i\gamma_iU$, one left coset per vertex at distance $\le k$, and let $\bar F(g):=\frac1{24}\sum_{u\in U}F_+(ug)$, which is left-$U$-invariant. Then
\begin{align}
\sum_{\gamma\in\Lambda_k}F_+(\gamma)&=\sum_{\gamma\in\Lambda_k}F_+(\gamma^{-1})=\sum_i\sum_{u\in U}F_+(u^{-1}\gamma_i^{-1})\notag\\
&=24\sum_i\bar F(\gamma_i^{-1})=24\sum_{j\le k}(T_j\bar F)(e).
\end{align}

\emph{Isotypic decomposition.} By Peter--Weyl, $L^2(\SU)=\bigoplus_nV_n\otimes V_n^*$, and only even $n=2\ell$ occur on $\PU$; the decomposition is preserved by left and right translations, hence by $T_j$ and by left-$U$-averaging. Since $F_+$ is left-$H_1$-invariant, its $n$-th component is a rank-one matrix coefficient $F_n(g)=\langle v_0,\pi_n(g)w_n\rangle$ with $v_0$ the unique $H_1$-fixed unit vector, and it corresponds to the degree-$\ell$ spherical-harmonic component $f_\ell$ of $F_+$ viewed on $S^2$, with $\|F_n\|_2=\|f_\ell\|_{L^2(S^2)}$ and $n+1=2\ell+1$. Left translations act on the $v$ side only, so $T_j\bar F_n(g)=\langle v',\pi_n(g)w_n\rangle$ is again rank one, and rank-one coefficients satisfy $\|\phi\|_\infty\le\sqrt{n+1}\,\|\phi\|_2$. For $n\ge1$, using (R) and $\|\bar F_n\|_2\le\|F_n\|_2$,
\begin{equation}
\begin{aligned}
|(T_j\bar F_n)(e)|&\le\sqrt{n+1}\,\|T_j\bar F_n\|_2\\
&\le\sqrt{n+1}\,(j+1)2^{j/2}\,\|F_n\|_2 .
\end{aligned}
\end{equation}
The $n=0$ term contributes $24\sum_{j\le k}(\#\text{vertices at distance }j)\,\mu(F_+)=N_k\,a_\varepsilon\le72\cdot4\sqrt2\cdot2^k\varepsilon^2$, which is the main term with $C_1=72\cdot4\sqrt2\le408$.

\emph{Summation.} Two geometric series finish the proof. First $\sum_{j\le k}(j+1)2^{j/2}\le(1-2^{-1/2})^{-1}(k+1)2^{k/2}$. Second, the degree-$\ell$ component of $a_\varepsilon P_r$ is $a_\varepsilon r^\ell(2\ell+1)P_\ell$, and $\|P_\ell\|_{L^2(S^2)}=(2\ell+1)^{-1/2}$, so with $(1-r)^2=4\varepsilon^2$,
\begin{multline}
\sum_{n\ge1}\sqrt{n+1}\,\|F_n\|_2=a_\varepsilon\!\sum_{\ell\ge1}(2\ell+1)r^\ell\\
\le\ a_\varepsilon\frac{1+r}{(1-r)^2}=(1+r)^{3/2}\le2^{3/2}.
\end{multline}
Collecting, $\#(\Lambda_k\cap T_\varepsilon(C))\le\sum_{\gamma\in\Lambda_k}F_+(\gamma)\le C_12^k\varepsilon^2+C_2(k+1)2^{k/2}$ with
\begin{equation}
C_2=24\,(1-2^{-1/2})^{-1}2^{3/2}=231.7\le232 .
\end{equation}
Taking $1-r=\lambda\varepsilon$ with $\lambda\ne2$ trades the two constants against each other, $C_1=72(\lambda^2+4r)^{3/2}/(\lambda(1+r))$ and $C_2=24(1-2^{-1/2})^{-1}(\lambda^2+4r)^{3/2}/\lambda^3$; only $\log_2C_2$ enters the bounds of Sec.~\ref{sec:tube}, through $A_2$, and it varies by one bit over $\lambda\in[2,4]$. \qed

\section{Proofs of Theorems~\ref{thm:rate2} and~\ref{thm:main2}}\label{app:rate}

\begin{lemma}[Cost-weighted mutual information]\label{lem:gibbs}
Let $(A,W)$ be jointly distributed with $|\mathrm{supp}\,A|\le K$; let $t$ be a cost on the alphabet of $W$ with $\#\{w:t(w)\le\tau\}\le n(\tau)$ for every $\tau\le\tau_*$ and $n(0)\ge1$ (so $Z_\lambda\ge1$), let $\lambda\in(0,1]$, and let $\tau_*\ge\lambda^{-1}\log_2K$ be an integer. Then
\begin{equation}
I(A;W)\ \le\ \lambda\,\E[t(W)]+1+\log_2Z_\lambda,
\end{equation}
where $Z_\lambda:=\sum_{\tau\le\tau_*}n(\tau)2^{-\lambda\tau}$. The case $\lambda=\tfrac12$, with $Z:=Z_{1/2}$, is the one used in Theorem~\ref{thm:rate2}; $\lambda=\tfrac5{11}$ and $\lambda=\tfrac7{17}$ are used in Theorems~\ref{thm:rate115} and~\ref{thm:rate177}, and $\lambda=1/\alpha$ with $\alpha<\tfrac52$ in Theorem~\ref{thm:rate52}.
\end{lemma}
\begin{proof}
Let $E:=\mathbf1\{t(W)>\tau_*\}$, a function of $W$, and $q:=\Pr[E=1]$. Then $I(A;W)=I(A;W,E)\le H(E)+(1-q)I(A;W\mid E=0)+qI(A;W\mid E=1)$, and $I(A;W\mid E=1)\le H(A\mid E=1)\le\log_2K$ because $A$ takes at most $K$ values. On $E=0$ the variable $W$ is supported on $\{t\le\tau_*\}$, so the Gibbs inequality $H(W)\le\lambda\E[t]+\log_2\sum_w2^{-\lambda t(w)}$ gives
\[\begin{aligned}I(A;W\mid E=0)&\le H(W\mid E=0)\\&\le\lambda\E[t\mid E=0]+\log_2Z_\lambda,\end{aligned}\]
because $\sum_{t(w)\le\tau_*}2^{-\lambda t(w)}\le\sum_{\tau\le\tau_*}n(\tau)2^{-\lambda\tau}=Z_\lambda$ (the shell of cost exactly $\tau$ has at most $n(\tau)$ elements). Finally $\E[t]\ge(1-q)\E[t\mid E=0]+q\tau_*$ and $\log_2K\le\lambda\tau_*$, so $(1-q)\lambda\E[t\mid E=0]+q\log_2K\le\lambda\E[t]$, while $H(E)\le1$.
\end{proof}
For $\lambda=\tfrac12$ the choice $\tau_*=\lceil2\log_2K\rceil$ is the smallest admissible one and the best: $Z$ grows with $\tau_*$. With the profile of Corollary~\ref{cor:profile},
\begin{equation}
Z\le\frac{C_1\varepsilon^22^{\tau_*/2}}{1-2^{-1/2}}+\frac{C_2(\tau_*+1)(\tau_*+2)}2,
\end{equation}
in which the first term is $O(\varepsilon)$ for the tuned family; at $\varepsilon=10^{-10}$, $\tau_*=66$ and $A_2=1+\log_2Z=20.01$.

\begin{lemma}[Cost-weighted entropy]\label{lem:cost}
Let $W$ take values in an alphabet with cost $t:\mathcal A\to\mathbb N$ such that $\#\{W:t(W)\le\tau\}\le c_0+c_1\tau+c2^\tau\varepsilon^\beta$ for all $\tau$. Put $\tau_0:=\lfloor\beta L-\log_2c\rfloor$, $\mathcal A_0:=\{t\le\tau_0\}$, $N_0:=c_0+c_1\tau_0+1$ (so $|\mathcal A_0|\le N_0$) and $q:=\Pr[W\notin\mathcal A_0]$, and assume $\tau_0\ge1$, equivalently $\varepsilon\le(2c)^{-1/\beta}$. Then
\begin{multline}
H(W)\le\E[t(W)]-q(\beta L-\log_2c-1)\\
+h_2(q)+\log_2N_0+3\log_2(1+\E t).
\end{multline}
\end{lemma}
\begin{proof}
$H(W)\le h_2(q)+(1-q)\log_2N_0+qH(W\mid W\notin\mathcal A_0)$. On the complement code $t-\tau_0$ in Elias-$\gamma$ and an index among the words of cost exactly $t$, of which there are at most $c_0+c_1t+c2^t\varepsilon^\beta\le(c_0+c_1t+1)\,c2^t\varepsilon^\beta$ (for $t>\tau_0$, $c2^t\varepsilon^\beta\ge1$), and $\log_2(c_0+c_1t+1)\le\log_2N_0+\log_2(1+t-\tau_0)$ because $c_1\le N_0/\tau_0$, both steps using $\tau_0\ge1$: length $\le t-\beta L+\log_2c+\log_2N_0+3\log_2(t-\tau_0+1)+1$. Then $q\E[t\mid\notin]\le\E t$ and $q\log_2(1+a)\le\log_2(1+qa)$.
\end{proof}

\begin{proof}[Proof of Theorem~\ref{thm:rate2}]
As in Theorem~\ref{thm:main1}, $A$ is uniform on a product translate given $Z$, so the $A_j$ are conditionally independent given $Z$. Per-coordinate Fano (the decoded $\hat X_j$ is a function of $(\overline W^{(t)}_j,\sigma_t,Z)$ and errs with probability $\le\delta$): $I(A_j;\overline W_j,\sigma_t\mid Z)\ge\ell_\delta$. By the chain rule and conditional independence,
\begin{multline}\label{eq:chain}
m\ell_\delta\le\sum_jI(A_j;\sigma_t\mid Z)+\sum_jI(A_j;\overline W_j\mid\sigma_t,Z)\\
\le S+\sum_jI(A_j;\overline W_j\mid\sigma_t,Z),
\end{multline}
where $\sum_jI(A_j;\sigma_t\mid Z)\le I(A;\sigma_t\mid Z)\le H(\sigma_t)\le S$ because $I(A_j;\sigma\mid Z,A_{<j})\ge I(A_j;\sigma\mid Z)$ under conditional independence.

Fix $j$ and $(s,z)$, write $t:=t^{(t)}_j$, and let $E_j$ be the indicator that coordinate $j$ is correct, $\delta_j:=\Pr[E_j=0\mid s,z]$. Then
\begin{multline}
I(A_j;\overline W_j\mid s,z)\\
\le h_2(\delta_j)+(1-\delta_j)\,I(A_j;\overline W_j\mid E_j=1,s,z)+\delta_j\log_2K.
\end{multline}
On $E_j=1$ the alphabet of $\overline W_j$ obeys the profile of Corollary~\ref{cor:profile}, and $A_j$ takes at most $K$ values, so Lemma~\ref{lem:gibbs} with $\lambda=\tfrac12$ and $\tau_*=\lceil2\log_2K\rceil$ gives
\begin{equation}
I(A_j;\overline W_j\mid E_j=1,s,z)\ \le\ \tfrac12\E[t\mid E_j=1,s,z]+A_2 .
\end{equation}
Using $(1-\delta_j)\E[t\mid E_j=1,s,z]\le\E[t\mid s,z]$,
\begin{multline}
I(A_j;\overline W_j\mid s,z)\\
\le\tfrac12\E[t_j\mid s,z]+A_2+h_2(\delta_j)+\delta_j\log_2K.
\end{multline}
Average over $(s,z)$ (Jensen for $h_2$), sum over $j$, and use $\sum_j\bar\delta_j\le m\delta$: $m\ell_\delta\le S+\tfrac12\bar T_t+mA_2+mh_2(\delta)+\delta m\log_2K$. Rearranging with $\ell_\delta=(1-\delta)\log_2K-h_2(\delta)$ gives the claim.
\end{proof}

\begin{proof}[Proof of Theorem~\ref{thm:main2}]
Start from~\eqref{eq:chain}. Fix a coordinate $j$ and a context $(s,z)$; let $E_j$ be the event that coordinate $j$ is correct, $d_j:=\Pr[E_j=0\mid s,z]$, and $q_j:=\Pr[t_{\min}(\overline W_j)>\tau_0\mid E_j=1,s,z]$. Write $k:=\log_2K$ and $g_\delta:=h_2(\delta)+\delta k$. On $E_j=1$ the word lies, by Conjecture~\ref{conj:H} applied to the coset of Corollary~\ref{cor:profile}, in an alphabet whose cost-$\le\tau_0$ part has at most $N_0$ elements. Splitting first on $E_j$ and then on the high-cost event,
\[
I(A_j;\overline W_j\mid s,z)\le h_2(d_j)+d_jk+(1-d_j)\big[1+\log_2N_0+q_jk\big].
\]
Averaging over contexts, summing over $j$ and inserting into~\eqref{eq:chain} gives $m\ell_\delta-S\le mg_\delta+ma_0+k\bar{\mathcal A}^{>}_t$ with $\bar{\mathcal A}^{>}_t=\sum_j\overline{(1-d_j)q_j}$, which is~\eqref{eq:active} because $m\ell_\delta-mg_\delta-S=D_t$. For the second bound apply Lemma~\ref{lem:cost} to the alphabet conditional on $E_j=1$, its hypothesis $\tau_0\ge1$ being the accuracy restriction $\varepsilon\le(2c)^{-1/2}$ of the theorem:
\begin{multline*}
I(A_j;\overline W_j\mid E_j=1,s,z)\le\E[t\mid E_j=1,s,z]-bq_j\\+1+\log_2N_0+3\log_2\big(1+\E[t\mid E_j=1,s,z]\big).
\end{multline*}
Multiply by $1-d_j$, average over contexts and sum over $j$: $(1-d_j)\E[t\mid E_j=1,s,z]\le\E[t\mid s,z]$, the concavity of $x\mapsto\log_2(1+x)$ bounds the logarithmic terms by $3m\log_2(1+\bar T_t/m)$, and the high-cost terms sum to $b\bar{\mathcal A}^{>}_t$, so $D_t\le\bar T_t-b\bar{\mathcal A}^{>}_t+ma_0+3m\log_2(1+\bar T_t/m)$. Since $b\ge0$ by the hypothesis $\varepsilon\le(2c)^{-1/2}$, substituting~\eqref{eq:active} yields~\eqref{eq:rate3}. Summing over $t<r$ and using concavity of the logarithm gives the statement for $\bar\Tpre$; at $\delta=0$ and $S=0$ the additive terms are $O(m\log L)$ per round.
\end{proof}

\section{Proof of Theorem~\ref{thm:elem}}\label{app:elem}

Throughout, $C$ denotes constants depending on nothing, and \emph{exponents} are logarithms to base $R'$, where $R'\asymp2^{t/4}$ is the radius, defined below, of the sphere carrying the numerators of $T$-count $t$. Every inequality between exponents below is strict with a fixed margin, so constants are absorbed once $R'$ is large; for bounded $R'$ the claim is absorbed into $\eta$.

\paragraph{Numerators.} Let $\mathcal O$ be the maximal order of Sec.~\ref{sec:tube}, embedded in $\R^4\times\R^4$ through the two real embeddings $\sigma_1,\sigma_2$ of $\Q(\sqrt2)$ applied coordinatewise in the basis $1,i,j,k$. Its coordinates lie in $\tfrac12\Z[\sqrt2]$. A unitary of minimal $T$-count $t$ is $\pm w/|w|_{\sigma_1}$ for some $w\in\mathcal O$ with $\mathrm{nrd}(w)=(2+\sqrt2)^tu$, $u$ a totally positive unit. Totally positive units of $\Z[\sqrt2]$ are squares, so after a central rescaling, which does not change the unitary, $\mathrm{nrd}(w)=n_t$, with $n_t=2^{t/2}$ for even $t$ and $n_t=(2+\sqrt2)2^{(t-1)/2}$ for odd $t$. Hence $w$ lies on the sphere of radius $R':=\sqrt{\sigma_1(n_t)}$ in the first factor and of radius $\sqrt{\sigma_2(n_t)}$ in the second, because the algebra is definite at both real places; both radii lie between $2^{t/4}/2$ and $2\cdot2^{t/4}$. For $\mathrm{SU}(2)$ lifts $\tr(V^\dagger W)=2\langle v,w\rangle$, so $\dproj(W,V)=\min_\pm|\hat w\mp\hat v|$ with $\hat w=w/|w|$. Therefore $T_\varepsilon(G_1R_zG_2)$ corresponds exactly to the set $X_t$ of numerators whose $\sigma_1$-image lies within Euclidean distance $\varepsilon|w|$ of $|w|\cdot\mathcal C$, where $\mathcal C=G_1\{R_z(\theta)\}G_2\subset S^3$ is a great circle ($\mathcal C=-\mathcal C$), and the map $w\mapsto W$ is two-to-one. Write $\Pi\subset\R^4$ for the plane of $\mathcal C$, $\mathcal C(s)=\cos s\,e_1+\sin s\,e_2$, and put $a_0:=\tfrac95$. For $t\le\tfrac{20}9L$ we have $\varepsilon=2^{-L}\le2^{-9t/20}\le CR'^{-a_0}$. The count is monotone in $\varepsilon$, so it suffices to bound $\#X_t$ for the tube of radius $CR'^{-a_0}$, which from now on we call $\varepsilon$.

Every point of the tube is $x=R'(\cos\varphi\,\mathcal C(s)+\sin\varphi\,n)$ with $n\in\Pi^\perp$ a unit vector and $0\le\varphi\le2\arcsin(\varepsilon/2)\le2\varepsilon$ (in the first embedding, whose index $\sigma_1$ we suppress). We call $s$ its \emph{core parameter} and $R'\sin\varphi\,n\in\Pi^\perp$ its \emph{transverse offset}.

\paragraph{Step 1: boxes.} Fix exponents $a\ge\max(a_0,b)$ and $b\ge0$, and put $\rho:=R'^{-a}$, $\lambda:=R'^{-b}$. A \emph{box} $B$ is the set of tube points whose core parameter lies in an interval $I$ of length $\lambda$ and whose transverse offset lies in a square of side $\rho R'$. The tube is covered by $\ll R'^{2(a-a_0)+b}$ boxes.

\begin{lemma}\label{lem:E1}
If $2a+3b>8$ and $a_0+a\ge2b$, all points of $X_t\cap B$ lie in one affine hyperplane.
\end{lemma}
\begin{proof}
Take coordinates $(e_1,e_2)$ rotated so that $e_1=\mathcal C(s_I)$ with $s_I$ the midpoint of $I$, together with an orthonormal basis of $\Pi^\perp$. For $x\in B$, the $\Pi^\perp$ components vary by $\le2\rho R'$. The $\Pi$ component is $R'\cos\varphi\,\mathcal C(s)$, whose $e_2$ coordinate varies by $\le\lambda R'$ and whose $e_1$ coordinate varies by $\le R'(1-\cos\tfrac\lambda2)+R'|\Delta\cos\varphi|\le C\lambda^2R'+C\varepsilon\rho R'\le C\lambda^2R'$. Here $\sin\varphi$ varies by $\le\rho$ in a box, so $|\Delta\cos\varphi|\le C\varepsilon\rho$, and $\varepsilon\rho\le\lambda^2$ is the hypothesis $a_0+a\ge2b$. For five points $w_0,\dots,w_4$ the determinant of $w_1-w_0,\dots,w_4-w_0$ therefore has $|\cdot|_{\sigma_1}\le C\lambda^3\rho^2R'^4$. In the second embedding all differences have length $\le CR'$, so $|\cdot|_{\sigma_2}\le CR'^4$. The determinant lies in $\tfrac1{16}\Z[\sqrt2]$, and a nonzero element $\xi$ of that set has $|\sigma_1(\xi)\sigma_2(\xi)|\ge2^{-8}$. Since $\lambda^3\rho^2R'^8=R'^{8-3b-2a}\to0$, every determinant vanishes, and the affine span of $X_t\cap B$ has dimension $\le3$.
\end{proof}

If the affine span has dimension $\le2$, the points lie on a circle and Lemma~\ref{lem:E5} applies directly. Otherwise they lie on $Y:=S^3_{R'}\cap H$, a round $2$-sphere of some radius $r=R'^{1-g}$ in the hyperplane $H=\{\langle\nu,x\rangle=c\}$, $|\nu|=1$, $c\ge0$.

\begin{lemma}[Patches]\label{lem:E2}
Let $J\subset I$ be the set of core parameters of points of $Y\cap B$. Then $J$ is contained in a union of at most two intervals of total length $\le e_s/R'$, and every point of $Y\cap B$ has transverse offset in a disc of radius $e_u$, where $e_u:=C\min(\rho R',r)$, and
\[
e_s:=\begin{cases}\lambda R'&(r>R'/2),\\ C\min\big(\lambda R',\ \sqrt{\rho R'\,(r+\rho R')},\ r\big)&(r\le R'/2).\end{cases}
\]
\end{lemma}
\begin{proof}
The claim for $e_u$ is immediate from the definition of $B$ and $Y\subset$ ball of radius $r$. For $e_s$ the case $r>R'/2$ is trivial ($J\subset I$), so let $r\le R'/2$, whence $c=\sqrt{R'^2-r^2}\ge\tfrac{\sqrt3}2R'$. Let $R'\sin\varphi_0\,n_0$ be the centre of the transverse square of $B$, and put $p(s):=R'(\cos\varphi_0\,\mathcal C(s)+\sin\varphi_0\,n_0)$. Since $\mathcal C(s)\perp n_0$, this is a point of the sphere. If $x\in Y\cap B$ has core parameter $s$, then $p=p(s)$ satisfies $\delta:=\mathrm{dist}(p,Y)\le|x-p|\le C\rho R'$. Let $y\in Y$ be nearest to $p$. Write $\nu=\tfrac c{R'^2}y+\tfrac r{R'}\,\omega$ with $\omega\perp y$ a unit vector, which is possible because $\langle\nu,y\rangle=c$ and $|y|=R'$. Using $\langle y,p-y\rangle=-\tfrac12|p-y|^2$ for $p,y$ on the sphere,
\[
|\langle\nu,p\rangle-c|=|\langle\nu,p-y\rangle|\le\frac{\delta^2}{2R'}+\frac{r\delta}{R'}=:\delta'.
\]
Now $\langle\nu,p(s)\rangle=A\cos(s-s_0)+R'\sin\varphi_0\langle\nu,n_0\rangle$ with $A=R'\cos\varphi_0|P_\Pi\nu|\le R'$, so $J\subset\{s:|h(s)|\le\delta'\}$ with $h(s):=A\cos(s-s_0)-c'$ and $c':=c-R'\sin\varphi_0\langle\nu,n_0\rangle\ge(\tfrac{\sqrt3}2-3\varepsilon)R'\ge0.6R'$, because a transverse square meeting the tube has its centre at offset $\le(2\varepsilon+\rho)R'\le3\varepsilon R'$ and $\varepsilon\to0$. This set is empty unless $A\ge c'-\delta'\ge R'/2$, and it lies in $|s-s_0|\le\pi/3$, where $h$ is concave with $h''\le-A/2\le-R'/4$. For a concave $h$ with $h''\le-\mu$ the set $\{|h|\le\delta'\}$ is $\{h\ge-\delta'\}\setminus\{h>\delta'\}$, a union of at most two intervals. The set $\{|h'|\le\sqrt{\mu\delta'}\}$ is an interval of length $\le2\sqrt{\delta'/\mu}$, and outside it each monotone piece of $\{|h|\le\delta'\}$ has length $\le2\sqrt{\delta'/\mu}$, so the total length is $\le6\sqrt{\delta'/\mu}$. Hence $J$ lies in at most two intervals of total length $\le C\sqrt{\delta'/R'}\le C\sqrt{\rho R'(r+\rho R')}/R'$. Finally $R'\,\mathrm{diam}\,J\le Cr$, because $Y$ has diameter $2r$ and core parameters are Lipschitz on $B$.
\end{proof}

\paragraph{Step 2: cells.} Fix exponents $X_S\le X_L$ and put $x_L:=R'^{X_L}$, $x_S:=R'^{X_S}$. A \emph{cell} of $B$ is the set of points of $B$ whose core parameter lies in an interval of length $x_L/R'$ and whose transverse offset lies in a square of side $x_S$. The cells of $B$ form a grid.

\begin{lemma}[Coplanarity in a cell]\label{lem:E3}
Assume $X_L<1-g$ (i.e.\ $2x_L\le r/C$) and $2X_L-1\le X_S$. If
\begin{equation}\label{eq:E3}
X_L+X_S+\min\big(X_S,\,2X_L-(1-g)\big)+3<0,
\end{equation}
then all points of $X_t\cap Y\cap Q$ in a cell $Q$ lie on one circle.
\end{lemma}
\begin{proof}
The $\Pi$ component of a point of $Q$ is $R'\cos\varphi\,\mathcal C(s)$, with $s$ in an interval of length $x_L/R'$. In a cell $\sin\varphi$ varies by $\le x_S/R'$, so $R'\cos\varphi$ varies by $\le Cx_S$. That set therefore lies within $C(x_L^2/R'+x_S)\le Cx_S$ of a segment of length $x_L$ in $\Pi$. The $\Pi^\perp$ component lies in a square of side $x_S$. Hence $Q$ lies within $Cx_S$ of the translate $\sigma$ of that segment by the centre of the square, and $\mathrm{diam}\,Q\le Cx_L$.

Let $w_0,\dots,w_3\in X_t\cap Y\cap Q$ and let $V$ be the $3$-volume of $w_1-w_0,w_2-w_0,w_3-w_0$. Since the differences lie within $Cx_S$ of the direction of $\sigma$, $V\le Cx_Lx_S^2$. Since the four points lie on the sphere $Y$ inside a ball of radius $Cx_L\le r/2$, they are within $Cx_L^2/r$ of the tangent plane $T$ of $Y$ at $w_0$. Their projections to $T$ lie in the projection of $Q$, which is within $Cx_S$ of a segment of length $x_L$, so every triangle they form has area $\le Cx_Lx_S$. Therefore $V\le Cx_Lx_S\min(x_S,x_L^2/r)$.

In the second embedding $V\le CR'^3$. The coordinates of $(w_1-w_0)\wedge(w_2-w_0)\wedge(w_3-w_0)$ lie in $\tfrac18\Z[\sqrt2]$, and~\eqref{eq:E3} says that $V_{\sigma_1}V_{\sigma_2}\to0$. So the wedge vanishes, and any four such points are coplanar. The points of $Y\cap Q$ therefore lie in one affine plane $P$, i.e.\ on the circle $Y\cap P$.
\end{proof}

\begin{lemma}[Cells met]\label{lem:E4}
The number of cells of $B$ meeting $Y$ is $\le C R'^{N}$, where
\[
\begin{gathered}
N:=\max\big(0,\ S_0-X_L,\ U_0-X_S,\\ S_0+U_0-X_L-X_S,\ 2(U_0-X_S)\big),
\end{gathered}
\]
with $S_0:=\log_{R'}e_s$ and $U_0:=\log_{R'}e_u$.
\end{lemma}
\begin{proof}
$Y\cap B$ is semialgebraic of bounded complexity, so it has $\le C$ connected components. Translating the grid by a generic offset, which does not affect Lemma~\ref{lem:E3}, we may assume that no wall contains an open piece of $Y$. A cell meeting $Y\cap B$ either contains a whole component or contains a point of $Y$ on one of its walls. The walls are of two kinds.
\begin{itemize}\setlength\itemsep{0pt}
\item \emph{Core walls} are the sets of fixed core parameter; each lies in a $3$-space $\mathrm{span}(\mathcal C(s_i))\oplus\Pi^\perp$. By Lemma~\ref{lem:E2} only $\le C(1+e_s/x_L)$ of them meet $Y\cap B$. Each meets $Y$ in an arc of a circle whose transverse offset stays in a disc of radius $e_u$. An arc of a circle meets a planar grid of sides $w_1,w_2$ in $\le C(1+\mathrm{TV}_1/w_1+\mathrm{TV}_2/w_2)$ cells, $\mathrm{TV}_i$ being the total variation of the $i$-th coordinate, which is $\le Ce_u$ here. So it crosses $\le C(1+e_u/x_S)$ faces, which are squares of side $x_S$.
\item \emph{Transverse walls} are the sets where one transverse coordinate is fixed. Only $\le C(1+e_u/x_S)$ of them meet $Y\cap B$. On each, $Y\cap B$ consists of $\le C$ circular arcs. On each arc the core parameter stays in one interval of Lemma~\ref{lem:E2}, so it varies by $\le Ce_s/R'$, and the other transverse coordinate varies by $\le Ce_u$. So it crosses $\le C(1+e_s/x_L+e_u/x_S)$ faces, which are rectangles of sides $x_L$ and $x_S$.
\end{itemize}
Summing the two kinds gives the claim.
\end{proof}

\begin{lemma}[Points on a circle]\label{lem:E5}
A circle $S^3_{R'}\cap P$, with $P$ an affine plane spanned by points of $X_t$, contains $\le2^{O(t/\log t)}$ points of $X_t$.
\end{lemma}
\begin{proof}
Write $P=w_0+\R d_1+\R d_2$ with $w_0\in X_t$ and $d_1,d_2$ differences of points of $X_t$. Their coordinates lie in $\tfrac12\Z[\sqrt2]$ and have height $R'^{O(1)}$. By Cramer's rule every point of $X_t\cap P$ is $w_0+xd_1+yd_2$ with $(X,Y):=(\Delta x,\Delta y)\in\Z[\sqrt2]^2$, for a fixed nonzero $\Delta\in\Z[\sqrt2]$ of height $R'^{O(1)}$ (a $2\times2$ minor of $2d_1,2d_2$). After multiplying by a fixed power of $2$, the condition $\mathrm{nrd}=n_t$ becomes $aX^2+2bXY+cY^2+2eX+2fY=0$, with $a,b,c,e,f\in\Z[\sqrt2]$ of height $R'^{O(1)}$ and $D:=ac-b^2$ totally positive.

Completing the square with $U=aX+bY+e$ and $V=DY+af-be$ gives $DU^2+V^2=M$, where $M:=(af-be)^2+De^2$. If $M=0$ there is at most one point. Otherwise $\alpha=V+U\sqrt{-D}$ lies in the ring of integers $\mathcal O_E$ of the CM field $E=\Q(\sqrt2,\sqrt{-D})$ and has relative norm $M$. Elements of relative norm one in a CM extension are roots of unity, of which a quartic field has at most $12$. So the number of such $\alpha$ is at most $12$ times the number of ideals of relative norm $(M)$, which is $\le d(N_{\Q(\sqrt2)/\Q}M)^2=2^{O(\log R'/\log\log R')}$. The pair $(U,V)$ determines $(X,Y)$.
\end{proof}

\paragraph{Assembly and certificate.} By Lemmas~\ref{lem:E1}--\ref{lem:E5},
\[
\#X_t\ \le\ 2^{O(t/\log t)}\,R'^{\,2(a-a_0)+b+\sup_{g\ge0}\min N},
\]
where for each $g$ the minimum is over cell exponents satisfying the hypotheses of Lemma~\ref{lem:E3}. For $g>2.01$ the whole sphere has $r^3R'^3\to0$, so all its points are coplanar and the cost is $0$. With $S_0,U_0$ from Lemma~\ref{lem:E2},
\begin{itemize}\setlength\itemsep{0pt}
\item for $r>R'/2$: $S_0=1-b$, and then $g<\log_{R'}2$;
\item for $r\le R'/2$: $S_0=\min\big(1-b,\tfrac12(2-a-\min(g,a)),1-g\big)$;
\item $U_0=\min(1-a,1-g)$.
\end{itemize}
Take $a_0=\tfrac95$, $a=\tfrac{91}{50}$ and $b=\tfrac{73}{50}$, so that $2a+3b=8.02$ and $2(a-a_0)+b=\tfrac32$. The constraints of Lemma~\ref{lem:E3} only become more restrictive as $g$ grows, and $N$ is nonincreasing in $g$. The condition $a_0+a\ge2b$ of Lemma~\ref{lem:E1} holds ($3.62\ge2.92$). So it suffices to fix one cell $(X_L,X_S)$ for each interval $[g_i,g_{i+1}]$ of the grid $g_i=i/400$ ($0\le i\le804$), checking the constraints at $g_{i+1}$ and $N$ at $g_i$, and to treat the case $r>R'/2$ separately on $[0,\tfrac1{400}]$, which contains $[0,\log_{R'}2)$ once $R'\ge2^{400}$. An exact rational computation (\texttt{research/cert/certificate.py}) finds such cells with $\sup_g\min N\le0.305001$. Hence $\log_{R'}\#X_t\le1.805001$, and $\#X_t\le2^{\,0.4513\,t+O(t/\log t)}$. Summing over $T$-counts $\le t$ costs a factor $t+1$. This proves~\eqref{eq:elem}.

The worst radii are $g\approx1.1$, i.e.\ spheres of radius $\approx R'^{-0.1}$, at which the core meets $Y$ nearly tangentially and $\lambda R'\approx\sqrt{r\rho R'}$. With $a_0=a=\tfrac{179}{100}$ and $b=\tfrac{737}{500}$, which satisfy $2a+3b=8.002$, $a\ge b$ and $a_0+a\ge2b$, the same computation (\texttt{research/cert/certificate\_223.py}) gives $2(a-a_0)+b+\sup_g\min N\le1.791<\tfrac{400}{223}$. The worst case is now at $g\approx a$. This gives the rate $2.23$ stated after Theorem~\ref{thm:rate115}. A numerical optimisation over $(a,b)$ finds that these lemmas cannot give more than about $2.234$.

\section{Proof of Theorem~\ref{thm:dioph}}\label{app:dioph}

We keep the notation of Appendix~\ref{app:elem}; Lemmas~\ref{lem:E1}--\ref{lem:E5} and their proofs hold for any $a_0\in(0,2)$. Now put $a_0:=\tfrac{28}{17}$, $a:=a_0$, $b:=\tfrac{157}{100}$ and $E_*:=\tfrac{11183}{6800}<a_0$. For $t\le\tfrac{17}7(L+2)$ we have $\varepsilon=2^{-L}\le4\cdot2^{-7t/17}\le CR'^{-a_0}$, so as before it suffices to bound $\#X_t$ for the tube of radius $CR'^{-a_0}$, which we call $\varepsilon$. We show $\#X_t\le R'^{E_*+o(1)}$. Since $R'\le2\cdot2^{t/4}$ and $E_*/4=\kappa_D$, summing over $T$-counts $\le t$ then gives~\eqref{eq:dioph}.

The exponents satisfy $2a+3b=8.0041>8$, $a\ge\max(a_0,b)$, $a_0+a\ge2b$ and $2b\ge a$. By Step~1 of Appendix~\ref{app:elem} the tube is covered by $\ll R'^{b}$ boxes (as $a=a_0$), and by Lemma~\ref{lem:E1} the points of $X_t$ in a box $B$ lie in one affine hyperplane. If they span at most a plane, Lemma~\ref{lem:E5} bounds them by $2^{O(t/\log t)}$, and these boxes contribute $R'^{b+o(1)}\le R'^{E_*+o(1)}$ in total. Otherwise the points of $B$ span a hyperplane $H_B$, and they lie on the sphere section $Y_B:=S^3_{R'}\cap H_B$, of radius $r=R'^{1-g}$ (its \emph{class} $g$).

\paragraph{Heights.} Let $\langle\cdot,\cdot\rangle$ be the standard $K$-bilinear form on $K^4$, $K=\Q(\sqrt2)$, in the basis $1,i,j,k$, and let $\mathcal O^*:=\{x\in K^4:\langle x,\mathcal O\rangle\subset\mathcal O_K\}$. Since $i,j,k\in\mathcal O$, since the form is unimodular on $\mathcal O_K^4$, and since $\langle x,y\rangle=\tfrac12\mathrm{trd}(x\bar y)\in\tfrac12\mathcal O_K$ for $x,y\in\mathcal O$,
\[
2\mathcal O_K^4\subset2\mathcal O\subset\mathcal O^*\subset\mathcal O_K^4\subset\mathcal O .
\]
All of these are $\mathcal O_K$-lattices in $\R^4\times\R^4$ via $(\sigma_1,\sigma_2)$. For $v\in K^4$, and likewise for $v$ in an exterior power $\Lambda^kK^4$ with the Euclidean norm of its Plücker coordinates, put $h(v):=|\sigma_1v|\,|\sigma_2v|$. We use three elementary facts.
\begin{itemize}\setlength\itemsep{0pt}
\item[(H1)] $h(uv)=h(v)$ for units $u$ of $\mathcal O_K$, and every $v$ has a \emph{balanced} unit multiple, with $|\sigma_1v|,|\sigma_2v|\le2\sqrt{h(v)}$ (multiply by powers of $1+\sqrt2$).
\item[(H2)] If $v\ne0$ has coordinates in $\mathcal O_K$, then $h(v)\ge1$: a nonzero coordinate $\xi$ has $|\sigma_1\xi\,\sigma_2\xi|=|N\xi|\ge1$. This applies to $v\in\mathcal O^*$ and to nonzero wedges of vectors of $\mathcal O^*$.
\item[(H3)] For $X,Y>0$ at most $1+CXY$ elements $\xi\in\mathcal O_K$ have $|\sigma_1\xi|\le X$ and $|\sigma_2\xi|\le Y$. If $XY<\tfrac14$, two of them would differ by an element of norm $<1$. Otherwise, balance the box by a unit and count points of a lattice of covolume $2\sqrt2$ and minimum $\ge\sqrt2$ in a rectangle of sides $\le4\sqrt{XY}$.
\end{itemize}
A hyperplane $H$ spanned by points of $\mathcal O$ has a normal in $\mathcal O^*$ that is \emph{primitive}, $Kn\cap\mathcal O^*=\mathcal O_Kn$ ($\mathcal O_K$ is a principal ideal domain). It is unique up to units, and $H=\{\langle n,x\rangle=m_H\}$ with $m_H\in\mathcal O_K$. We call $h(n)$ the height of $H$ and of $Y=S^3_{R'}\cap H$. We write $H_0=n^\perp\subset K^4$ for the direction space of $H$, and $\nu:=\sigma_1n/|\sigma_1n|$, so that $\sigma_1H=\{\langle\nu,x\rangle=c\}$ as in Appendix~\ref{app:elem}.

\begin{lemma}[Covolumes]\label{lem:F1}
Let $n\in\mathcal O^*$ be primitive, $M:=\mathcal O\cap H_0$ and $M^*:=\{\ell\in H_0:\langle\ell,M\rangle\subset\mathcal O_K\}$. Then $\mathrm{covol}(M)=\sqrt2\,h(n)$ and $\mathrm{covol}(M^*)=512/\mathrm{covol}(M)$.
\end{lemma}
\begin{proof}
The functional $f(x)=\langle n,x\rangle$ maps $\mathcal O$ onto an ideal $(\alpha)$ of $\mathcal O_K$. Then $n/\alpha\in\mathcal O^*\cap Kn=\mathcal O_Kn$, so $\alpha$ is a unit and $f(\mathcal O)=\mathcal O_K$. The kernel of $f$ on $\mathcal O$ is $M$. The orthogonal complement of $H_0\otimes\R$ in $\R^8$ is spanned by $(\sigma_1n,0)$ and $(0,\sigma_2n)$. The projection of $\mathcal O$ onto it is the image of $\mathcal O_K$ under $\xi\mapsto(\sigma_1\xi/|\sigma_1n|,\sigma_2\xi/|\sigma_2n|)$, of covolume $2\sqrt2/h(n)$. Hence $\mathrm{covol}(M)=\mathrm{covol}(\mathcal O)\,h(n)/2\sqrt2$, and $\mathrm{covol}(\mathcal O)=\mathrm{covol}(\mathcal O_K^4)/16=4$ from the basis~\eqref{eq:order}. Under the trace form $\sum_i\langle\sigma_ix,\sigma_iy\rangle$ the $\Z$-dual of $M$ is $\delta^{-1}M^*$, where $\delta=2\sqrt2$ generates the different of $\mathcal O_K$. Therefore $\mathrm{covol}(M^*)=|N\delta|^3/\mathrm{covol}(M)$.
\end{proof}

\begin{lemma}[Fibration of a sphere section]\label{lem:F2}
Let $Y=S^3_{R'}\cap H$ have height $h(n)\ge R'^y$, and let $Z\subset Y$ with $\mathrm{vol}_3(\mathrm{conv}\,Z)\le R'^v$ in the first embedding. Then $X_t\cap Z$ lies on at most $1+CR'^{1+(v-y)/3}$ circles, and has at most $2^{O(t/\log t)}(1+CR'^{1+(v-y)/3})$ points.
\end{lemma}
\begin{proof}
The difference body $D:=\mathrm{conv}\,Z-\mathrm{conv}\,Z$ lies in the $3$-space $\sigma_1(H_0)\otimes\R$ and has volume $\le20R'^v$ (Rogers--Shephard~\cite{RogersShephard}). By John's theorem its polar body satisfies $\mathrm{vol}(D^\circ)\ge c/\mathrm{vol}(D)$~\cite{Gruber}. For $T_1,T_2>0$ the body $\{(\ell_1,\ell_2):\ell_1\in T_1D^\circ,\ |\ell_2|\le T_2\}\subset H_0\otimes\R$ has volume $\ge c\,T_1^3T_2^3/\mathrm{vol}(D)$. By Minkowski's first theorem and Lemma~\ref{lem:F1}, it contains some $\ell\in M^*\setminus\{0\}$ once $T_1T_2=C(\mathrm{vol}(D)/h(n))^{1/3}$. For $w,w'\in X_t\cap Z$ we have $w-w'\in M$, so $\langle\ell,w-w'\rangle\in\mathcal O_K$, with $|\sigma_1(\cdot)|\le T_1$ and $|\sigma_2(\cdot)|\le CT_2R'$. By (H3), $\langle\ell,w\rangle$ takes at most $1+CT_1T_2R'$ values on $X_t\cap Z$. Since $\ell\in H_0=n^\perp$ is nonzero and $\langle n,n\rangle\ne0$, $\ell$ and $n$ are $K$-independent, so each fibre lies in the affine plane $H\cap\{\langle\ell,x\rangle=\mathrm{const}\}$, i.e.\ on a circle. It has at most two points or spans that plane, and Lemma~\ref{lem:E5} applies.
\end{proof}
For a cell and $y=0$ this is Lemma~\ref{lem:E3}: condition~\eqref{eq:E3} says exactly that $1+v/3<0$.

\begin{lemma}[Volumes]\label{lem:F3}
Let $Y$ be a sphere section of class $g$ meeting the tube.
\begin{itemize}\setlength\itemsep{0pt}
\item[(a)] For a cell $Q$ as in Lemma~\ref{lem:E3}, $\mathrm{vol}_3\,\mathrm{conv}(Y\cap Q)\le Cx_Lx_S\min(x_S,x_L^2/r)$.
\item[(b)] Let $W$ be the set of tube points whose core parameter lies in an interval of length $e/R'$, $e\le R'/4$, and whose transverse offset lies in a disc of radius $e'$. If $r\le R'/2$, or if $e'\le e$ and $e^2/R'\le Ce'$, then $\mathrm{vol}_3\,\mathrm{conv}(Y\cap W)\le Cee'^2$.
\end{itemize}
\end{lemma}
\begin{proof}
(a) The proof of Lemma~\ref{lem:E3} bounds the volume of every tetrahedron with vertices in $Y\cap Q$ by $Cx_Lx_S\min(x_S,x_L^2/r)$; it uses no integrality. A convex body in $\R^3$ contains a simplex of at least a fixed fraction of its volume with vertices among its extreme points, which lie in the closure of $Y\cap Q$.

(b) Let $s_k$ be the midpoint of the interval. As in the proof of Lemma~\ref{lem:E1}, $W$ lies in a box with sides $\delta:=\max(e^2/R'+C\varepsilon e',\,e')$, $e$, $2e'$, $2e'$, and axes $\mathcal C(s_k)$, $\mathcal C'(s_k)=-\sin s_k\,e_1+\cos s_k\,e_2$ and an orthonormal basis of $\Pi^\perp$. So $Y\cap W$ lies in the section of this $4$-box by $H$. For any convex body $B_4$ and hyperplane $H$ with unit normal $\nu$, $B_4$ contains the double cone over $B_4\cap H$ with apexes at its extreme points in the directions $\pm\nu$, so $\mathrm{vol}_3(B_4\cap H)\le4\,\mathrm{vol}_4(B_4)/w_\nu$, where $w_\nu$ is the width of $B_4$ in direction $\nu$.
\begin{itemize}\setlength\itemsep{0pt}
\item If $r\le R'/2$: every $x\in Y\cap W$ has $\langle\nu,x\rangle=c\ge\tfrac{\sqrt3}2R'$ and $|x-R'\mathcal C(s_k)|\le e+4\varepsilon R'$, so $\langle\nu,\mathcal C(s_k)\rangle\ge0.6$. Hence $w_\nu\ge0.6\,\delta$, and $\mathrm{vol}_3\le C\delta ee'^2/\delta$.
\item Otherwise $\delta\le Ce'\le Ce$. A hyperplane section of a box has $3$-volume at most $2$ times the product of its three largest sides (project to the coordinate hyperplane on which $\nu$ has its largest component), and here that product is $\le Cee'^2$.
\end{itemize}
\end{proof}

For a box $B$ and a sphere section $Y$ of class $g$ and height $\ge R'^y$, let $S_0,U_0$ be exponents such that the core parameters of $Y\cap B$ lie in at most two intervals of total length $\le R'^{S_0-1}$ and the transverse offsets of $Y\cap B$ lie in a disc of radius $R'^{U_0}$. Lemma~\ref{lem:E2} provides such $S_0,U_0$, and Lemma~\ref{lem:F4} below may provide a smaller $S_0$. The proof of Lemma~\ref{lem:E4} uses Lemma~\ref{lem:E2} only through this property. So Lemmas~\ref{lem:E4},~\ref{lem:F2} and~\ref{lem:F3} give $\#(X_t\cap Y\cap B)\le R'^{\mathrm{cost}(S_0,U_0,g,y)+o(1)}$ with the \emph{cost}
\begin{equation}\label{eq:cost}
\begin{gathered}
\mathrm{cost}:=\min\Big(\big(1+\tfrac13(S_0+2U_0-y)\big)_+,\\ \min_{X_L,X_S}\big[N+(P)_+\big]\Big),\\
P:=1+\tfrac13\big(X_L+X_S+\min(X_S,2X_L-(1-g))-y\big),
\end{gathered}
\end{equation}
where $N=N(X_L,X_S)$ is as in Lemma~\ref{lem:E4}. The inner minimum is over cells with $X_S\le X_L<1-g$ and $2X_L-1\le X_S$. The first entry is the whole patch (Lemma~\ref{lem:F3}(b) applied to at most two sets $W$; if $r>R'/2$ it uses $a\ge b$ and $2b\ge a$, since then $e=\lambda R'\ge e'=\rho R'$ and $e^2/R'=\lambda^2R'\le\rho R'$). The second entry comes from the cells (Lemma~\ref{lem:F3}(a)).

\begin{lemma}[Crossing depth]\label{lem:F4}
Let $Y=S^3_{R'}\cap H$ with $r\le R'/2$ meet the tube, put $A:=R'|P_\Pi\nu|$ and $\Delta:=A-c$, and let $\delta':=C(\varepsilon^2R'+\varepsilon r)$ and $\delta'_B:=C(\rho^2R'+\rho r)$. Then $-\delta'\le\Delta\le r^2/R'$. Moreover there is an absolute $C'$ such that:
\begin{itemize}\setlength\itemsep{0pt}
\item[(i)] the core parameters of $Y\cap(\text{tube})$ lie in at most two intervals of total length $\le e^T/R'$, where $e^T:=C\min(\sqrt{R'\delta'},r)$ in general and $e^T:=C\min(\delta'\sqrt{R'/\Delta},r)$ if $\Delta\ge C'\delta'$;
\item[(ii)] for every box $B$, the same holds for $Y\cap B$ with $e^B:=C\min(\lambda R',\sqrt{R'\delta'_B},r)$ in general and $e^B:=C\min(\lambda R',\delta'_B\sqrt{R'/\Delta},r)$ if $\Delta\ge C'\delta'$;
\item[(iii)] for a primitive normal $n$ and an interval $I$, at most $1+Ch(n)|I|R'$ values $m_H\in\mathcal O_K$ give a hyperplane $\{\langle n,x\rangle=m_H\}$ that contains a point of $X_t$ and has $\Delta\in I$.
\end{itemize}
\end{lemma}
\begin{proof}
A tube point with core parameter $s$ is within $4\varepsilon R'$ of $p(s):=R'\mathcal C(s)$. The computation in the proof of Lemma~\ref{lem:E2}, with $p(s)$ in place of the box point, gives $|\chi(s)|\le\delta'$ for $\chi(s):=\langle\nu,p(s)\rangle-c=A\cos(s-s_0)-c$. So $\Delta=\max\chi\ge-\delta'$, and $\Delta\le R'-\sqrt{R'^2-r^2}\le r^2/R'$. As there, now with $c\ge\tfrac{\sqrt3}2R'$ and $\delta'=o(R')$, the set $\{|\chi|\le\delta'\}$ lies in $|s-s_0|\le0.6$, where $A\ge c-\delta'\ge0.8R'$. It is a union of at most two intervals of total length $\le C\sqrt{\delta'/R'}$, and core parameters of $Y$ vary by $\le Cr/R'$. This is the first bound in (i).

If $\Delta\ge C'\delta'$ with $C'\ge2$, write $u=s-s_0$. Then $|\chi|\le\delta'$ forces $A(1-\cos u)\ge\Delta/2$, so $|u|\ge\sqrt{\Delta/R'}$, and there $|\chi'|=A|\sin u|\ge c\sqrt{R'\Delta}$. Each of the two intervals therefore has length $\le C\delta'/\sqrt{R'\Delta}$.

For (ii) use the box point $p_B(s)$ of Lemma~\ref{lem:E2}. This gives $\chi_B(s)=A_B\cos(s-s_0)-c'_B$, with $|\chi_B|\le\delta'_B$ on the core parameters of $Y\cap B$. The first bound is Lemma~\ref{lem:E2}. For the second, $|A_B-A|\le R'\varphi_0^2$ and $|c'_B-c|\le R'\varphi_0|P_{\Pi^\perp}\nu|$ with $\varphi_0\le4\varepsilon$. Moreover $c\,|P_{\Pi^\perp}\nu|\le2\varepsilon R'+r$, because a point $x=c\nu+y$ of $Y\cap(\text{tube})$ has $|P_{\Pi^\perp}x|\le2\varepsilon R'$ and $|y|=r$. Hence $\Delta_B:=A_B-c'_B$ satisfies $|\Delta_B-\Delta|\le C(\varepsilon^2R'+\varepsilon r)\le\Delta/2$ for $C'$ large. Since $\rho\le\varepsilon$, also $\delta'_B\le\delta'\le\Delta_B/2$, and the argument of the previous paragraph applies to $\chi_B$.

(iii) Here $A$ depends on $n$ only, and $c=\sigma_1(m_H)/|\sigma_1n|$. So $\sigma_1(m_H)$ lies in an interval of length $|\sigma_1n|\,|I|$. For a point $w\in X_t$ of the hyperplane, $|\sigma_2m_H|=|\langle\sigma_2n,\sigma_2w\rangle|\le C|\sigma_2n|R'$. Apply (H3) to the differences of two such values.
\end{proof}
We call $Y$ \emph{tangent} if $\Delta<C'\delta'$ and \emph{crossing} otherwise.

\begin{lemma}[Normals of low height]\label{lem:F5}
Let $\gamma:=\min(g,a_0)$ with $g>0$, $\psi:=R'^{-\gamma}$ and $h:=R'^y\ge1$. Let $\mathcal N$ be the set of primitive normals, up to units and sign, of the sphere sections of class $\ge g$ that meet the tube, contain a point of $X_t$, and have height $<h$. Then $\#X_t\le R'^{E_*+o(1)}$, or
\[
\begin{gathered}
\#\mathcal N\le CR'^{\mathrm{nor}(y)},\\
\mathrm{nor}(y):=\max\big(y,\,2y-G,\,3y-\gamma,\,4y-2\gamma\big)+o(1),
\end{gathered}
\]
with $G:=\max\big(0,\gamma+\tfrac{E_*}2-2\big)$.
\end{lemma}
\begin{proof}
\emph{The body.} If $Y$ meets the tube at $x$, then $x/R'$ is within $2\varepsilon$ of $\mathcal C$ and within $Cr/R'$ of $\nu$. So $|P_{\Pi^\perp}\nu|\le C\psi$. The balanced representative of a normal in $\mathcal N$ therefore lies in
\[
\begin{aligned}
\mathcal K:=\{x:\ &|P_\Pi\sigma_1x|\le C\sqrt h,\ |P_{\Pi^\perp}\sigma_1x|\le C\psi\sqrt h,\\ &|\sigma_2x|\le C\sqrt h\},
\end{aligned}
\]
and $\#\mathcal N\le\#(\mathcal K\cap\mathcal O^*)$.

\emph{Minima.} For $j\le4$ let $\mu_j$ be the least $\mu$ such that $\mu\mathcal K$ contains $j$ $K$-independent vectors of $\mathcal O^*$, attained by $K$-independent $v_1,\dots,v_4$. Any $2j-1$ $\Z$-independent vectors span a $K$-space of dimension $\ge j$. So the successive minima of the rank-$8$ lattice $\mathcal O^*$ with respect to $\mathcal K$ satisfy $\lambda_{2j-1},\lambda_{2j}\ge\mu_j$. Henk's bound~\cite{Henk}, which for a lattice of rank $8$ reads $\#(\mathcal K\cap\mathcal O^*)\le2^{7}\prod_{i\le8}\lfloor2/\lambda_i+1\rfloor$, then gives
\[
\#(\mathcal K\cap\mathcal O^*)\le C\prod_{j\le4}\max(1,\mu_j^{-2}).
\]

\emph{Wedges.} Write $\sigma_1v_j=p_j+q_j$ with $p_j\in\Pi$ and $q_j\in\Pi^\perp$. Then $|p_j|,|\sigma_2v_j|\le C\mu_j\sqrt h$ and $|q_j|\le C\mu_j\psi\sqrt h$. The wedges $v_1$, $v_1\wedge v_2$, $v_1\wedge v_2\wedge v_3$ and $v_1\wedge\dots\wedge v_4$ are nonzero with coordinates in $\mathcal O_K$, so their heights are $\ge1$ by (H2). In the expansion of $\sigma_1(v_1\wedge v_2\wedge v_3)$ the term $p_1\wedge p_2\wedge p_3$ vanishes, since $\Pi$ is a plane, and every other term contains some $q_j$. In the four-fold wedge every nonzero term contains two of them. Hence
\[
\begin{gathered}
\mu_1\ge\frac c{\sqrt h},\qquad \mu_1\mu_2\ge\frac ch,\\
(\mu_1\mu_2\mu_3)^2\ge\frac c{\psi h^3},\qquad (\mu_1\cdots\mu_4)^2\ge\frac c{\psi^2h^4}.
\end{gathered}
\]
If exactly $d$ of the $\mu_j$ are $\le1$, the count is $\le C(\mu_1\cdots\mu_d)^{-2}$. For $d=0,1,2,3,4$ this is $\le C$, $Ch$, $Ch^2$, $C\psi h^3$ and $C\psi^2h^4$. This proves the claim, except that for $d=2$ it gives $2y$ in place of $\max(y,2y-G)$.

\emph{The case $d=2$.} Now $\mathcal K\cap\mathcal O^*$ lies in the $K$-plane $V:=Kv_1+Kv_2$, since a third vector outside it would give $\mu_3\le1$.
\begin{itemize}\setlength\itemsep{0pt}
\item \emph{The plane $V$.} Let $V_1\subset\R^4$ be the real span of $\sigma_1V$, let $\theta_1\le\theta_2$ be its principal angles with $\Pi$, and let $f_1,f_2$ be orthonormal principal vectors of $V_1$. Thus $|P_{\Pi^\perp}f_i|=\sin\theta_i$ and $P_{\Pi^\perp}f_1\perp P_{\Pi^\perp}f_2$.
\item \emph{Its height $H_V$.} As $\mathcal O_K$ is a principal ideal domain, the saturated module $S:=V\cap\mathcal O_K^4$ is a direct summand of $\mathcal O_K^4$, with a basis $s_1,s_2$. Put $H_V:=h(s_1\wedge s_2)\ge1$.
\item \emph{Its orthogonal complement.} Complete $s_1,s_2$ to a basis $s_1,\dots,s_4$ of $\mathcal O_K^4$. The dual basis $s_i^*$ with respect to the unimodular form $\langle\cdot,\cdot\rangle$ also lies in $\mathcal O_K^4$, and $V^\perp\cap\mathcal O_K^4=\mathcal O_Ks_3^*\oplus\mathcal O_Ks_4^*$. Moreover $s_3^*\wedge s_4^*=\pm\det(s_1,\dots,s_4)^{-1}\star(s_1\wedge s_2)$ with a unit determinant. The Hodge star $\star$ commutes with $\sigma_i$ and is an isometry, so $V^\perp\cap\mathcal O_K^4$ also has height $H_V$.
\item \emph{Covolumes.} Both saturated modules have covolume $8H_V$ in their $4$-dimensional real spans in $\R^4\times\R^4$. The lattices $V\cap\mathcal O^*$ and $V^\perp\cap\mathcal O^*$ lie between these modules and twice them, so they have covolume between $8H_V$ and $128H_V$.
\end{itemize}

(a) \emph{Sector count.} For $x\in\mathcal K\cap V$ write $\sigma_1x=\alpha f_1+uf_2$. Then $|P_{\Pi^\perp}\sigma_1x|^2=\alpha^2\sin^2\theta_1+u^2\sin^2\theta_2$, so $|u|\le C\beta_V\sqrt h$ with $\beta_V:=\min(1,\psi/\sin\theta_2)$, and $|\alpha|\le C\sqrt h$.
\begin{itemize}\setlength\itemsep{0pt}
\item Repeat the minima argument in the rank-$4$ lattice $V\cap\mathcal O^*$ with the body so obtained, with $K$-minima $\mu'_1\le\mu'_2$. As before $\mu'_1\ge ch^{-1/2}$.
\item Two $K$-independent $w_1,w_2\in V\cap\mathcal O^*$ have $w_1\wedge w_2=\xi\,s_1\wedge s_2$ with $0\ne\xi\in\mathcal O_K$, so $h(w_1\wedge w_2)\ge H_V$.
\item On the other hand $|\sigma_1(w_1\wedge w_2)|=|\alpha_1u_2-\alpha_2u_1|\le C\mu'_1\mu'_2\beta_V h$ and $|\sigma_2(w_1\wedge w_2)|\le C\mu'_1\mu'_2h$.
\end{itemize}
Hence $\#(\mathcal K\cap\mathcal O^*)\le C\max(1,\,h,\,\beta_V h^2/H_V)$.

(b) \emph{Global fibration.} Let $\kappa_1\le\dots\le\kappa_4$ be the successive minima of $V^\perp\cap\mathcal O^*$ with respect to $\{|\sigma_1x|\le1,\,|\sigma_2x|\le1\}$. By Minkowski's second theorem $\kappa_1\cdots\kappa_4\le CH_V$.
\begin{itemize}\setlength\itemsep{0pt}
\item Take $K$-independent $\mathfrak n_1,\mathfrak n_2\in V^\perp\cap\mathcal O^*$ with $\mathfrak n_1$ realizing $\kappa_1$ and $\mathfrak n_2$ among the first three minima. Then $h(\mathfrak n_1)h(\mathfrak n_2)\le\kappa_1^2\kappa_3^2\le CH_V$.
\item The principal angles of $V_1^\perp$ and $\Pi^\perp$ are those of $V_1$ and $\Pi$, so every unit vector of $V_1^\perp$ has $|P_\Pi\cdot|\le\sin\theta_2$.
\item For a tube point $x=R'(\cos\varphi\,\mathcal C(s)+\sin\varphi\,n')$, therefore $|\langle\sigma_1\mathfrak n_i,x\rangle|\le|\sigma_1\mathfrak n_i|R'(\sin\theta_2+2\varepsilon)$, while $|\sigma_2\langle\mathfrak n_i,w\rangle|\le C|\sigma_2\mathfrak n_i|R'$.
\item By (H3), applied to differences, $\langle\mathfrak n_i,w\rangle\in\mathcal O_K$ takes at most $1+Ch(\mathfrak n_i)R'^2\varsigma$ values on $X_t$, where $\varsigma:=\sin\theta_2+2\varepsilon$. The two values fix $w$ in a translate of $V$, which meets the sphere in a circle, so Lemma~\ref{lem:E5} applies.
\end{itemize}
Since $h(\mathfrak n_i)\ge1$ and $R'^2\varsigma\ge2\varepsilon R'^2\ge1$,
\[
\#X_t\le2^{O(t/\log t)}\,C\,H_VR'^4\varsigma^2 .
\]

(c) \emph{Dichotomy.} If $H_VR'^4\varsigma^2\le R'^{E_*}$, (b) gives $\#X_t\le R'^{E_*+o(1)}$. Otherwise put $X:=\log_{R'}H_V\ge0$ and $\vartheta:=\log_{R'}\sin\theta_2\le0$; then $X+4+2\max(\vartheta,-a_0)\ge E_*-o(1)$. By (a) the count is $\le CR'^{\max(y,\,2y-X-\max(0,\vartheta+\gamma))}$. On the region allowed by the constraint, $X+\max(0,\vartheta+\gamma)\ge G-o(1)$:
\begin{itemize}\setlength\itemsep{0pt}
\item for $\vartheta\ge\tfrac{E_*-4}2$ this follows from $X\ge0$;
\item for $-a_0\le\vartheta<\tfrac{E_*-4}2$ it follows from $X\ge E_*-4-2\vartheta>0$;
\item for $\vartheta<-a_0$ it follows from $X\ge E_*-4+2a_0>\tfrac{E_*-4}2+a_0$.
\end{itemize}
The minimum $G$ is attained at $X=0$, $\vartheta=\tfrac{E_*-4}2$.
\end{proof}

\paragraph{Assembly.} Put $g_i:=i/400$ for $0\le i\le804$ and $y_j:=j/100$. Let $\ell_\varepsilon(g):=\max(1-2a_0,1-g-a_0)$ and $\ell_\rho(g):=\max(1-2a,1-g-a)$ be the exponents of $\delta'$ and $\delta'_B$. Every box $B$ above has its section $Y_B$ in one of the following cases.
\begin{itemize}\setlength\itemsep{0pt}
\item \emph{Tiny sections, $g>\tfrac{201}{100}$.} Any four points of $Y_B$ span a $3$-volume $\le Cr^3$ in the first embedding and $\le CR'^3$ in the second. As $r^3R'^3\to0$, they are coplanar, so $B$ contributes $R'^{o(1)}$.
\item \emph{Large sections, $g\le g_1$.} This includes all $r>R'/2$ once $R'\ge2^{400}$. Here $B$ contributes $R'^{\mathrm{cost}+o(1)}$, with the cost~\eqref{eq:cost} at $S_0=1-b$, $U_0=1-a$, $y=0$, and $g=g_1$ in $P$ and in the constraint $X_L<1-g$; in total $R'^{b+\mathrm{cost}+o(1)}$.
\item \emph{$g\in[g_i,g_{i+1}]$ with $i\ge1$.} Here $r\le R'/2$. Let $\mathrm{cost}_i(y)$ be the cost~\eqref{eq:cost} with $S_0=\min\big(1-b,\tfrac12(1+\ell_\rho(g_i)),1-g_i\big)$ and $U_0=\min(1-a,1-g_i)$ from Lemma~\ref{lem:E2}, and $g=g_{i+1}$ in $P$ and in the constraint $X_L<1-g$. Since every height is $\ge1$, these boxes contribute at most $R'^{b+\mathrm{cost}_i(0)+o(1)}$ in total. If that is not enough, the certificate supplies a threshold $y^*=y_{j^*}$ with $b+\mathrm{cost}_i(y^*)\le E_*$. It covers the boxes whose hyperplane $H_B$ has height $h(n_B)\ge R'^{y^*}$, $n_B$ its primitive normal (\emph{high}), and each layer $1\le j\le j^*$ of the remaining boxes (\emph{low}), with $h(n_B)\in[R'^{y_{j-1}},R'^{y_j})$, is bounded as follows.
\end{itemize}

\emph{Low layers.} The points of $X_t$ in the low boxes of layer $j$ lie on sphere sections $Y$ of class in $[g_i,g_{i+1}]$ and height in $[R'^{y_{j-1}},R'^{y_j})$, each containing a point of $X_t$. Count them globally, not per box.
\begin{itemize}\setlength\itemsep{0pt}
\item \emph{Normals.} By Lemma~\ref{lem:F5} there are $\le CR'^{\mathrm{nor}(y_j)}$ normals, with $\gamma=\min(g_i,a_0)$.
\item \emph{Offsets.} By Lemma~\ref{lem:F4}(iii), each normal gives $\le1+CR'^{y_j+\ell_\varepsilon(g_i)+1}$ tangent sections. Partition $[\ell_\varepsilon(g_{i+1}),1-2g_i]$ into $12$ intervals $[D_l,D_r]$. Each normal gives $\le1+CR'^{y_j+D_r+1}$ crossing sections with $\Delta\in[R'^{D_l},R'^{D_r}]$. If $1-2g_i<\ell_\varepsilon(g_{i+1})$, which happens for $g_i\ge a_0$, there are no crossing sections, since $\Delta\le r^2/R'<C'\delta'$.
\item \emph{Points per section.} Such a section contains at most $R'^{o(1)}$ times the smaller of two quantities.
\begin{itemize}\setlength\itemsep{0pt}
\item \emph{Tube bound} (Lemmas~\ref{lem:F2},~\ref{lem:F3}(b) and~\ref{lem:F4}(i)): $R'^{(1+(S_T+2U_T-y_{j-1})/3)_+}$, where $U_T=\min(1-a_0,1-g_i)$ and $S_T=\min\big(\tfrac12(1+\ell_\varepsilon),1-g_i\big)$, lowered to $\min(S_T,\ell_\varepsilon+\tfrac12(1-D_l))$ for crossing sections.
\item \emph{Box bound}: the number of boxes meeting $Y\cap(\text{tube})$, which is at most $R'^{(S_T-(1-b))_++2(U_T-(1-a))_+}$, times $R'^{\mathrm{cost}}$, with the cost~\eqref{eq:cost} at height $y_{j-1}$, with $U_0=\min(1-a,1-g_i)$ and with $S_0=\min\big(1-b,\tfrac12(1+\ell_\rho),1-g_i\big)$, lowered to $\min(S_0,\ell_\rho+\tfrac12(1-D_l))$ for crossing sections (Lemma~\ref{lem:F4}(ii)).
\end{itemize}
\end{itemize}
The layer is bounded if the normals, the offsets and the points per section together have exponent $\le E_*$ for the tangent class and for each crossing class.

All the quantities involved are monotone in $g$, $y$ and $\Delta$ on each interval, and each is evaluated at the endpoint that weakens the bound:
\begin{itemize}\setlength\itemsep{0pt}
\item $S_0$, $U_0$, $S_T$, $U_T$, $\ell_\varepsilon$, $\ell_\rho$ and $\gamma$ at $g_i$, except the lower end of the crossing range, where $\ell_\varepsilon$ bounds $\Delta$ from below and is taken at $g_{i+1}$; the sagitta term $\min(X_S,2X_L-(1-g))$ and the constraint $X_L<1-g$ at $g_{i+1}$;
\item the normals and offsets at $y_j$, and the points per section at $y_{j-1}$;
\item the offsets at $D_r$, and the points per section at $D_l$.
\end{itemize}
There are finitely many cases, and Lemma~\ref{lem:F5} is applied to finitely many pairs $(g_i,y_j)$. If its first alternative occurs in any of them, $\#X_t\le R'^{E_*+o(1)}$ directly.

\paragraph{Certificate.} {\raggedright An exact rational computation (\texttt{research/dioph/cert\_dioph.py}, run as \texttt{cert\_dioph.py 28/17 28/17 157/100 11183/6800 400 100 12}; the output is the same with $8$ crossing classes) checks every case: $b=\tfrac{157}{100}\le E_*$, the large sections, and for each $i$ either $b+\mathrm{cost}_i(0)\le E_*$ or a threshold $y^*\le\tfrac{23}{25}$ together with all its low layers. Cells are rational witnesses. The largest exponent the certificate checks is $\tfrac{67097}{40800}=1.64453<E_*=1.64456$, near $g\approx1.5$; it stops at the first admissible threshold. A float re-implementation, pointwise in $g$ and continuous in the layers and crossing depths (\texttt{research/verify4/indep\_model.py}), gives $1.6347$. The margin that matters is $a_0-E_*=\tfrac1{400}$. Hence $\#X_t\le R'^{E_*+o(1)}$. Every factor hidden in the $o(1)$ is an absolute constant, of which there are finitely many, or the bound $2^{O(t/\log t)}$ of Lemma~\ref{lem:E5}, and bounded $t$ are absorbed as in Appendix~\ref{app:elem}. This proves~\eqref{eq:dioph} with $\eta(t)=O(t/\log t)$.\par}

\section{Rational projections}\label{app:proj}

We keep the notation of Appendices~\ref{app:elem} and~\ref{app:dioph}, with $a:=a_0$, $b>\tfrac{8-2a_0}3$ (so that $2a+3b>8$ and Lemma~\ref{lem:E1} applies with one hyperplane per box) and $\delta:=\varepsilon R'$, the tube radius (not an error probability). Three lemmas are added to Appendix~\ref{app:dioph}.

\begin{lemma}[Annular fibration]\label{lem:G1}
Let $V\subset K^4$ be a $K$-plane of height $H_V$, and let $\theta_1\le\theta_2$ be the principal angles between
$V_1=\sigma_1(V)\otimes\R$ and $\Pi$. Then
\begin{multline*}
\#X_t\le2^{O(t/\log t)}\,C\big(1+(\sin\theta_2/\varepsilon)^{1/2}\\
+H_VR'^4\varepsilon(\sin\theta_2+\varepsilon)\big).
\end{multline*}
\end{lemma}
\begin{proof}
Put $v(w):=P_{V^\perp}w$. Points with the same $v(w)$ lie on $(w+V)\cap S^3_{R'}$, a circle. At least three of them span $w+V$, because
three distinct points of a circle are affinely independent, so by Lemma~\ref{lem:E5} it suffices to count the values $v(w)$. They lie in the $\mathcal O_K$-lattice
$\Lambda:=P_{V^\perp}(\mathcal O)$ of rank $4$.

\emph{Where $v(w)$ lies.} The map $P_{V_1^\perp}|_\Pi$ has singular values $\sin\theta_1\le\sin\theta_2$ with
orthonormal left singular vectors $u_1,u_2$. Hence $\sigma_1v(w)=P_{V_1^\perp}\sigma_1w$ lies within
$\delta_1:=3\delta$ of the ellipse $E=\{R'(\sin\theta_1\cos s\,u_1+\sin\theta_2\sin s\,u_2)\}$, the image of the
core. Moreover $|\sigma_2v(w)|\le2R'$.

\emph{Covering.} The $\delta_1$-neighbourhood of $E$ is covered by
$N\le C(1+(R'\sin\theta_2/\delta_1)^{1/2})$ rectangles of width $4\delta_1$ and lengths $\ell_k$ with
$\sum_k\ell_k\le C(R'\sin\theta_2+\delta_1)$. To see this, cut $E$ into arcs of length at most
$\ell_0$ and turning at most $\Theta_0$, with $\ell_0\Theta_0=\delta_1$ and $\Theta_0<\pi/2$; each arc is within $\delta_1$
of its chord, and every rectangle has length at least $4\delta_1$.

\emph{One rectangle.} Let $D_k$ be the difference body of $\{\sigma_1x\in P_k,\ |\sigma_2x|\le2R'\}$, and
$\lambda_1\le\dots\le\lambda_4$ the successive minima of $\Lambda$ with respect to it. Multiplication by
$1+\sqrt2$ gives $\lambda_2\le2.42\lambda_1$ and $\lambda_4\le2.42\lambda_3$. For the second, the $\Q$-span of three
independent vectors has odd dimension, so it is not stable under multiplication by $\sqrt2$.
\begin{itemize}\setlength\itemsep{0pt}
\item If $\lambda_3\le1$, Henk's bound and Minkowski's second theorem give
$\#(\Lambda\cap D_k)\le C\,\mathrm{vol}(D_k)/\mathrm{covol}(\Lambda)\le CH_V\ell_k\delta_1R'^2$, because
$\mathrm{covol}(\Lambda)=\mathrm{covol}(\mathcal O)/\mathrm{covol}(\mathcal O\cap V)\ge1/(8H_V)$.
\item Otherwise the count is $\le C\max(1,\lambda_1^{-2})$. A nonzero $x\in\Lambda$ has $h(x)\ge1/(CH_V)$:
one of the $K$-independent $\mathfrak n_1,\mathfrak n_2\in V^\perp\cap\mathcal O^*$ of Lemma~\ref{lem:F5}(b), with $h(\mathfrak n_1)h(\mathfrak n_2)\le CH_V$
has $\langle\mathfrak n_i,x\rangle\in\mathcal O_K\setminus0$. So $\lambda_1^{-2}\le CH_V\ell_kR'$, and
$\ell_kR'\le\ell_k\delta_1R'^2$ since $\delta R'=\varepsilon R'^2\ge1$.
\end{itemize}
Summing over $k$ proves the claim.
\end{proof}

\begin{lemma}[Few normals in rank one]\label{lem:G2}
In the proof of Lemma~\ref{lem:F5}, if at most one of the $\mu_j$ is $\le1$, then $\#\mathcal N\le C$. In the
sector count (a) the bound $C\max(1,h,\beta_Vh^2/H_V)$ may be replaced by $C\max(1,\beta_Vh^2/H_V)$.
\end{lemma}
\begin{proof}
In both cases all lattice vectors of the body lie in one $K$-line $Kv$. A primitive normal in $Kv$
generates $Kv\cap\mathcal O^*$ over $\mathcal O_K$, so it is unique up to units.
\end{proof}

For a box $B$ whose points span the hyperplane $H_B=\{\langle n_B,x\rangle=m_B\}$, $n_B$ primitive, let
$e_s,e_u$ be as in Lemma~\ref{lem:E2} and $\bar e:=\max(e_s,e_u)$. By the proof of
Lemma~\ref{lem:F3}(b), $Y_B\cap B$ lies in at most two $4$-boxes with sides $Ce_u,\bar e,2e_u,2e_u$.

\begin{lemma}[Projections]\label{lem:G3}
\begin{itemize}\setlength\itemsep{0pt}
\item[(a)] If $\mathfrak m\in\mathcal O_K^4$ is primitive and $\langle n_B,\mathfrak m\rangle=0$, then
$\#(X_t\cap B)\le2^{O(t/\log t)}C\big(1+\bar ee_uR'^2\,h(\mathfrak m)/h(n_B)\big)$.
\item[(b)] If $V$ is a $K$-plane of height $H_V$ with $n_B\in V$, then
$\#(X_t\cap B)\le2^{O(t/\log t)}C\big(1+\bar eR'\,H_V/h(n_B)\big)$.
\end{itemize}
\end{lemma}
\begin{proof}
(a) Let $W=\mathfrak m^\perp$. For $w\in X_t\cap B$ put $u:=P_Ww$, which lies in the lattice
$\Lambda_W:=P_W(\mathcal O)$. Since $n_B\in W$, $\langle n_B,u\rangle=m_B$, so $u$ lies on a $2$-dimensional affine
$K$-plane $P_B\subset W$. From $w=u+t\mathfrak m$ and
$\mathrm{nrd}(w)=\langle u,u\rangle+t^2\langle\mathfrak m,\mathfrak m\rangle$, each $u$ comes from at most two $w$.

The points $\sigma_1u$ lie in a planar convex set of area $\le C\bar ee_u$, and $|\sigma_2u|\le2R'$. They lie in
a translate of $L_B:=\Lambda_W\cap n_B^\perp$. Its covolume is
$\mathrm{covol}(\Lambda_W)\,h(n_B)/2\sqrt2\asymp h(n_B)/h(\mathfrak m)$, because
$\langle n_B,\cdot\rangle$ maps $\Lambda_W$ onto $\mathcal O_K$ by Lemma~\ref{lem:F1}.

If the second $K$-minimum of $L_B$ with respect to the difference body is $\le1$, the count is at most $C$
times volume over covolume. Otherwise the $u$ lie on an affine $K$-line $\ell$, so the $w$ lie on the circle
$(\ell+K\mathfrak m)\cap S^3_{R'}$, and Lemma~\ref{lem:E5} applies as in Lemma~\ref{lem:G1}. The two
$4$-boxes are treated separately.

(b) Project along $V^\perp$ instead. Then $u=P_Vw$ lies on an affine $K$-line in $V$. The lattice on this
line is $\mathcal O_K\kappa$ with $h(\kappa)\asymp h(n_B)/H_V$, and $\sigma_1u$ lies in an interval of length
$C\bar e$. By (H3) there are at most $1+C\bar eR'H_V/h(n_B)$ values of $u$. Each fibre $(u+V^\perp)\cap S^3_{R'}$
is a circle.
\end{proof}

\paragraph{Assembly.} Assemble as in Appendix~\ref{app:dioph}, with the box shape above, a target exponent $E_*<a_0$, and three changes.
\begin{itemize}\setlength\itemsep{0pt}
\item In Lemma~\ref{lem:F5}(c), Lemma~\ref{lem:G1} replaces (b). This replaces $G$ by
$G_3:=\max(0,\gamma+a_0+E_*-4)$, and by Lemma~\ref{lem:G2} the terms $y$ disappear from $\mathrm{nor}(y)$.
\item In the rank-three case, keep the factor $H_W:=h(\mathfrak m)$ that the argument of
Lemma~\ref{lem:F5} provides: the three-fold wedge is a nonzero multiple of the primitive $3$-vector of
$W=\mathfrak m^\perp$, so $(\mu_1\mu_2\mu_3)^2\psi h^3\ge cH_W$.
\item A low layer $j$ of class interval $i$ is bounded by the maximum, over the number $d$ of $K$-minima
$\le1$, of the following exponents. Here $\mathrm{pm}$ is the offsets-plus-points-per-section exponent of
Appendix~\ref{app:dioph} (offsets at $y_j$, points per section at $y_{j-1}$), maximised over the tangent and crossing classes; $b$ is the exponent of the number $R'^b$ of boxes (as $a=a_0$); $\bar S:=\max(S_0,U_0)=\log_{R'}\bar e$ with $S_0,U_0$ the patch exponents of Lemma~\ref{lem:E2} at $g_i$; and $X\in[0,\infty)$ is the adversary's height exponent, $\log_{R'}H_V$ for $d=2$ and $\log_{R'}H_W$ for $d=3$.
\begin{itemize}\setlength\itemsep{0pt}
\item $d\le1$: $\mathrm{pm}$.
\item $d=2$: $\max_X\min\big((2y_j-\max(X,G_3'))_++\mathrm{pm},\ b+(\bar S+1+X-y_{j-1})_+\big)$, with
$G_3'=E_*-4+a_0+\gamma$. The first entry uses the sector count, the second Lemma~\ref{lem:G3}(b) for every
box of the layer, since then $n_B\in V$.
\item $d=3$: $\max_X\min\big((3y_j-\gamma-X)_++\mathrm{pm},\ b+(\bar S+U_0+2+X-y_{j-1})_+\big)$, the
second entry by Lemma~\ref{lem:G3}(a) with the normal $\mathfrak m$ of $W:=Kv_1+Kv_2+Kv_3=\mathfrak m^\perp$, which contains $\mathcal K\cap\mathcal O^*$ and hence every $n_B$ of the layer.
\item $d=4$: $(4y_j-2\gamma)_++\mathrm{pm}$.
\end{itemize}
\end{itemize}
\paragraph{A finite instance.} For $a_0=a=\tfrac{400}{249}$, $b=\tfrac{8-2a_0}3+\tfrac1{1000}$ and $E_*=a_0-\tfrac1{800}$, the certificate \texttt{research/annulus/cert\_proj.py}, run as \texttt{cert\_proj.py 400/249 400/249 1192747/747000 319751/199200 400 100 12}, checks every case of this assembly exactly. The largest exponent the certificate checks (it stops at the first admissible threshold) is $1.60511<E_*=1.60518<a_0=1.60643$. Hence $\#X_t\le R'^{E_*+o(1)}$ for $t\le\tfrac{249}{100}(L+2)$, and the proof of Theorem~\ref{thm:rate177}, with $\lambda=\tfrac{100}{249}$, gives the rate $\tfrac{249}{100}$. Appendix~\ref{app:rate52} replaces this finite instance by the continuum form of the same assembly.

\section{Proof of Theorem~\ref{thm:rate52}}\label{app:rate52}

In the notation of Appendix~\ref{app:elem}, where $R'$ lies between $2^{t/4}/2$ and $2\cdot2^{t/4}$, the tube bound~\eqref{eq:rate52tube} for $t\le4(L+2)/a_0$ (then $\varepsilon\le2^{a_0+2}R'^{-a_0}$; this range contains $\tau_*=\lceil\alpha\log_2K\rceil$) is the statement $\#X_t\le R'^{T+o(1)}$, with $a_0=\tfrac85+\mu$ and $T=\tfrac85-\tfrac{4\mu}{11}$, for tubes of radius $\varepsilon\le CR'^{-a_0}$. Summing over $T$-counts $\le t$ costs a factor $t+1$.

\paragraph{Continuum reduction.} Throughout, $E_1:=(8-2a_0)/3<T<a_0$ ($T$ is the target $E_*$ of Appendices~\ref{app:dioph} and~\ref{app:proj}), $a=a_0$, $b=E_1$ up to the shift $\varpi$ below, and the side conditions
$a\ge b$, $a_0+a\ge2b$, $2b\ge a$ hold; they all hold for $0<\mu\le\frac1{10}$. All exponents in the assembly of Appendices~\ref{app:dioph}
and~\ref{app:proj} are piecewise linear in $(a_0,a,b,T,g,y,D,X)$ and in the cell exponents, with
absolutely bounded slopes. Fix $\varpi=\varpi(R')\to0$ with $\varpi\log R'\to\infty$, and run the assembly
as follows:
\begin{itemize}\setlength\itemsep{0pt}
\item take $a=a_0$ and $b=E_1+\varpi$, where $E_1:=(8-2a_0)/3$, so that Lemma~\ref{lem:E1} applies with one
hyperplane per box;
\item take class intervals, height layers and crossing classes of width $\varpi$;
\item impose every strict exponent inequality of the lemmas with margin $\varpi$ (each such inequality only has to beat an absolute constant, which $R'^{\varpi}\to\infty$ does; this replaces the fixed margins of Appendix~\ref{app:elem}).
\end{itemize}
Each checked quantity then exceeds its pointwise (continuum) value at the left end $g_i$ by $O(\varpi)$.
The places where the certificate uses $g_{i+1}$ are handled as follows. The sagitta term and the constraint
$X_L<1-g$ are covered by a shift of $2\varpi$. Crossing classes below $\ell_\varepsilon(g_i)$ are dominated by the
tangent class. The no-crossing case is empty. The zero case $r>R'/2$ is checked separately: the whole-box
cell has negative exponent there. A cell that is admissible
at one point with equality is shifted by $O(\varpi)$ to become admissible on its whole interval. There are
$O(\varpi^{-3})$ cases, which costs a factor $R'^{o(1)}$. So $\#X_t\le R'^{T+o(1)}$ as soon as the
following continuum claim holds:
\begin{multline*}
\mathrm C(a_0,T):\ \text{there is no }(g,y)\text{ with }0<g\le\tfrac{201}{100},\ y>0,\\
\mathrm{HIGH}(g,y)>T\text{ and }\mathrm{LOW}(g,y)>T .
\end{multline*}
Here $\mathrm{HIGH}$ and $\mathrm{LOW}$ are the pointwise versions of the two sides of the height split:
offsets and points per section are taken at the same height, the crossing depth is continuous, and the
cell costs are bounded by explicit candidate cells.

Since $\mathrm{HIGH}$ is nonincreasing in the height and tends to $E_1$, the claim $\mathrm C(a_0,T)$ says
exactly this: the least height $\eta^*$ with $\mathrm{HIGH}\le T$ has $\mathrm{LOW}\le T$ on every layer
below it. In the discretised assembly the last low layer $[y_{j^*-1},y_{j^*}]$ contains $\eta^*$, and its value is at most $\mathrm{LOW}(y_{j^*-1})+O(\varpi)\le T+O(\varpi)$. In $\neg\mathrm C$, the adversary's choices (rank, $X$, $D$, class of section) are existential and the
prover's minima are explicit terms. So $\neg\mathrm C$ is a quantifier-free formula of linear real
arithmetic.

\paragraph{Verification.} Put $a_0=\tfrac85+\mu$ and $T=E_1+\tfrac{10}{33}\mu$, with $\mu$ a variable.
The solver z3~\cite{Z3} shows that $\neg\mathrm C$ is unsatisfiable for $0<\mu\le\tfrac1{10}$
(\texttt{research/uniform52/verify\_z3.py 1/10 10/33}). Since $T=\tfrac85-\tfrac{4\mu}{11}<a_0$, the Gibbs lemma
(Lemma~\ref{lem:gibbs}) with $\lambda=1/\alpha$, $\alpha=4/a_0$, gives the rate with $A=O_\alpha(1)$, because the gap $\lambda-T/4$ is fixed. For larger
$a_0$, use monotonicity in $\varepsilon$. The same computation with $T=E_1+c\mu$ is unsatisfiable exactly
for $c\ge\tfrac{10}{33}$. So near the limit the method gives the exponent $\tfrac85-\tfrac4{11}(a_0-\tfrac85)$
and nothing better.

The balance behind $\tfrac{10}{33}$ is explicit, and it comes from two mechanisms.
The value $E_1+\tfrac{10}{33}\mu$ is attained exactly for $\tfrac45-\tfrac{7\mu}{11}\le g\le\tfrac85-\tfrac73\mu$.
On the upper part of this range, $g\gtrsim\tfrac65$, the rank-three balance applies:
\begin{itemize}\setlength\itemsep{0pt}
\item every patch is the whole box;
\item the whole-box cell gives $\mathrm{HIGH}-E_1=\max\big(0,(g-\tfrac25+\mu-\eta)/3\big)$;
\item the deepest crossing class contributes offsets $R'^{\,y+2-2g}$ with one point per section;
\item the rank-three routes cross at $E_1+\tfrac12(3y-3g+\tfrac65+\tfrac\mu3)$;
\item the two sides meet at $\eta^*=g-\tfrac25+\tfrac\mu{11}$.
\end{itemize}
On the lower part the rank-two route binds, through the target $T$ inside the dichotomy constant $G_3'$.
Balancing it against $\mathrm{HIGH}$ gives $c(1+\tfrac1{10})=\tfrac13$, which is the same constant $\tfrac{10}{33}$.

\section{Proof of Theorem~\ref{thm:clifford}}\label{app:clifford}

We use the numerators of Appendix~\ref{app:elem}. A word of minimal $T$-count $\tau$ is $\pm w/|w|$ with $w\in\mathcal O$ and $\mathrm{nrd}(w)=n_\tau$. We realize quaternions as matrices by
\[
a+bi+cj+dk\ \mapsto\ \begin{pmatrix}a+b\,\mathrm i&c+d\,\mathrm i\\-c+d\,\mathrm i&a-b\,\mathrm i\end{pmatrix}.
\]
Here $\mathrm i=\sqrt{-1}$, and $R_z(\theta)$ corresponds to the unit quaternion $\cos\tfrac\theta2-\sin\tfrac\theta2\,i$. Elements of $\mathcal O$ have coordinates in $\tfrac12\Z[\sqrt2]$, so their matrix entries lie in $\tfrac12\Z[\sqrt2,\mathrm i]\subset\tfrac12\Z[\omega]$.

\paragraph{Denominator exponent.} For even $\tau$, $|w|=2^{\tau/4}=\sqrt2^{\,\tau/2}$, so $W=(2w)/\sqrt2^{\,\tau/2+2}$ with $2w$ over $\Z[\omega]$. For odd $\tau$, $|w|=|1+\omega|\,2^{(\tau-1)/4}$, because $|1+\omega|^2=2+\sqrt2$. Multiplying $W$ by the phase $(1+\bar\omega)/|1+\omega|$ and using $2+\sqrt2=\sqrt2(1+\sqrt2)$ gives
\[
e^{\mathrm i\phi}W=\frac{(\sqrt2-1)(1+\bar\omega)\,2w}{\sqrt2^{\,(\tau+5)/2}} .
\]
The Clifford+$T$ matrix of the word differs from $e^{\mathrm i\phi}W$ by a phase in $\Q(\omega)$ whose square is a power of $\omega$ ($\det e^{\mathrm i\phi}W=\bar\omega$), hence by a power of $\omega$, which does not change $k$. So the denominator exponent $k$ of Lemma~\ref{lem:identity} satisfies $k\le\tfrac\tau2+3$. (For $\tau\le13$ the enumeration gives $k\in\{\lceil\tau/2\rceil,\lceil\tau/2\rceil+1\}$.)

\paragraph{The count.} By Clifford invariance of $t_{\min}$, $W\mapsto C_1^{-1}WC_2^{-1}$ maps the tube of $C_1R_zC_2$ onto the tube of the torus $\{R_z(\theta)\}$ and preserves $t_{\min}$. For $G_1,G_2\in\Gamma$, the map $W\mapsto G_1^{-1}WG_2^{-1}$ is injective and raises $t_{\min}$ by at most $h$. So it suffices to count words of minimal $T$-count $\tau'\le\tau$ within $\varepsilon$ of the torus.

The torus is the unit circle of $\mathrm{span}(1,i)$, so the distance from $\hat w$ to it is at least the length of the $(j,k)$-component of $\hat w$. Write $w=a+bi+cj+dk$. Then $|c|_{\sigma_1},|d|_{\sigma_1}\le\varepsilon|w|\le2\varepsilon\,2^{\tau/4}$ and $|c|_{\sigma_2},|d|_{\sigma_2}\le2\cdot2^{\tau/4}$.

Multiplication by a power of $1+\sqrt2$ scales $\sigma_1$ and $\sigma_2$ by reciprocal factors and is a bijection of $\tfrac12\Z[\sqrt2]$. A nonzero $x\in\tfrac12\Z[\sqrt2]$ has $|x|_{\sigma_1}|x|_{\sigma_2}\ge\tfrac14$. Together these show that $\#\{x\in\tfrac12\Z[\sqrt2]:|x|_{\sigma_1}\le X,\ |x|_{\sigma_2}\le Y\}\le C(1+XY)$. Hence the pair $(c,d)$ takes $\le C(1+\varepsilon2^{\tau/2})^2\le C(1+\varepsilon^22^\tau)$ values.

Given $(c,d)$, put $\mu:=n_{\tau'}-c^2-d^2$. It is nonzero, since $|(c,d)|_{\sigma_1}<|w|_{\sigma_1}$, and it is totally positive. The pair $(a,b)$ gives $\alpha=2a+2b\,\mathrm i\in\Z[\sqrt2,\mathrm i]\subset\Z[\omega]$ with $\alpha\bar\alpha=4\mu$. Now $\Z[\omega]$ is a principal ideal domain whose units of relative norm one are the eight powers of $\omega$, and the number of ideals of $\Z[\omega]$ of relative norm $(4\mu)$ is at most the number of ideal divisors of $(4\mu)$ in $\Z[\sqrt2]$. So there are $\le8\,d_{\Z[\sqrt2]}((4\mu))\le2^{O(\tau/\log\tau)}$ such $\alpha$, since $|N_{\Q(\sqrt2)/\Q}(4\mu)|\le C2^{\tau}$.

Summing over $\tau'\le\tau$, and dividing by $2$ for the sign of $w$, gives~\eqref{eq:clifford}. \qed

\section{Proof of Theorem~\ref{thm:oneshot}}
Let $\Omega$ be the success event; on $\Omega$, $u$ is a function of the message (snapshot and unitary tuple) together with $(v,y)$, and the message takes at most $2^{S_{\max}}N^{\le}_m(T_{\max})$ values; Lemma~\ref{lem:MA} gives the inequality. \qed

\section{\texorpdfstring{Proofs for Sec.~\ref{sec:mixing}}{Proofs for the mixing section}}\label{app:mixing}

\begin{proof}[Proof of Lemma~\ref{lem:fid}]
Let $J(\mathcal E)=(\mathcal E\otimes I)(|\Phi\rangle\langle\Phi|)$ and $\Phi_V:=|\Phi_V\rangle\langle\Phi_V|=J(\mathcal U_V)$. Since $|\Phi\rangle\langle\Phi|$ is an admissible input in the definition of the diamond norm, $\|J(\mathcal E)-\Phi_V\|_1\le\|\mathcal E-\mathcal U_V\|_\diamond$. For a state $\rho$ and a pure state $\Phi$, $\|\rho-\Phi\|_1=2\max_{0\le P\le I}\tr P(\rho-\Phi)\ge2\tr(I-\Phi)(\rho-\Phi)=2(1-\langle\Phi|\rho|\Phi\rangle)$. Combining gives $1-F_e\le\tfrac12\|\mathcal E-\mathcal U_V\|_\diamond$. For $\mathcal E=\sum_ip_i\mathcal U_{X_i}$, $F_e=\sum_ip_i|\langle\Phi_V|(X_i\otimes I)|\Phi\rangle|^2=\sum_ip_i|\tr(V^\dagger X_i)/2|^2$.
\end{proof}

\begin{proof}[Proof of Lemma~\ref{lem:conc}]
Both sides of (i) and the hypothesis are invariant under $X_i\mapsto-X_i$, so write $V^\dagger X_i=W_i=\cos\tfrac{\phi_i}2I-i\sin\tfrac{\phi_i}2(\mathbf n_i\!\cdot\!\sigma)$ with $\phi_i\in[0,\pi]$, so that $f_i=\cos^2(\phi_i/2)$ and $s_i:=1-f_i=\sin^2(\phi_i/2)$, whence $\sum_ip_is_i\le\eta^2$.

(i) With $R_z(\psi)^\dagger=\cos\tfrac\psi2I+i\sin\tfrac\psi2Z$ one has $\tfrac12\tr(R_z(\psi)^\dagger W_i)=\cos\tfrac\psi2\cos\tfrac{\phi_i}2+\sin\tfrac\psi2\sin\tfrac{\phi_i}2n_{i,z}$, whose modulus is at most $\cos\tfrac{|\psi|}2\cos\tfrac{\phi_i}2+\sin\tfrac{|\psi|}2\sin\tfrac{\phi_i}2=\cos\tfrac{|\psi|-\phi_i}2$. Now
\[
\cos^2\tfrac{|\psi|-\phi}2-\cos^2\tfrac\psi2=\tfrac12\big[\cos\psi(\cos\phi-1)+|\sin\psi|\sin\phi\big]
\]
which is at most $\sin^2\tfrac\phi2+|\sin\psi|\sin\tfrac\phi2$ for $\phi\in[0,\pi]$. Averaging with $p$ and using $\sum_ip_i\sin\tfrac{\phi_i}2\le(\sum_ip_is_i)^{1/2}\le\eta$ by Cauchy--Schwarz gives~\eqref{eq:window}. For the threshold, $\cos^2\tfrac\psi2+\eta^2+\eta|\sin\psi|<1-\eta^2$ iff $\sin^2\tfrac{|\psi|}2>2\eta^2+\eta|\sin\psi|$; since $|\sin\psi|\le2\sin\tfrac{|\psi|}2$ it suffices that $x^2-2\eta x-2\eta^2>0$ with $x=\sin\tfrac{|\psi|}2$, i.e.\ $x>(1+\sqrt3)\eta$.

(ii) $f_i\ge1-\eta^2$ gives $\sin(\phi_i/2)\le\eta$; by Lemma~\ref{lem:dproj}, $\dproj(X_i,V)=2\sin(\phi_i/4)=\sin(\phi_i/2)/\cos(\phi_i/4)\le\eta/\cos(\pi/4)=\sqrt2\eta$.
\end{proof}

\begin{lemma}[Fano for list decoding]\label{lem:fanolist}
Let $A$ take values in a set of size $M$, let $\mathcal L$ be a random list that is a function of $Y$ with $|\mathcal L|\le\ell$ and $\Pr[A\in\mathcal L]\ge1-\delta$. Then $H(A\mid Y)\le h_2(\delta)+(1-\delta)\log_2\ell+\delta\log_2M$.
\end{lemma}
\begin{proof}
With $E:=\mathbf1[A\notin\mathcal L]$, $H(A\mid Y)\le H(E)+H(A\mid Y,E)\le h_2(\delta)+(1-\delta)\log_2\ell+\delta\log_2M$.
\end{proof}

\begin{proof}[Proof of Theorem~\ref{thm:mixing}(a)]
Fix $t<r$ and, as in the proof of Theorem~\ref{thm:main1}, put $A:=a^{(t)}$, $Z:=(\mathrm{Past},\mathrm{Fut},\rho)$ and let $\sigma_t$ be the snapshot after round $t$. Given $\rho$ the process is deterministic, so the programs of rounds $<t$ on every coordinate are functions of $(\mathrm{Past},\rho)$ and those of rounds $>t$ are functions of $(\sigma_t,\mathrm{Fut})$; both are therefore determined by $(\sigma_t,Z)$. Fix a coordinate $j$ and values $(s,z)$, let $c:=\sum_{s'\ne t}a^{(s')}_j$ and $V_a:=R_z(\Delta(a+c))$.

\emph{Representative.} For $U$ in the support of $\pi^{(t)}_j$ put $\mathcal C_U:=\mathcal E^{>t}_j\circ\mathcal U_U\circ\mathcal E^{<t}_j$. Because the samples of different rounds are independent, $\mathcal E_j=\sum_U\pi^{(t)}_j(U)\mathcal C_U$, and $F_e(\cdot,V_a)$ is affine, so on the event that the coordinate is correct Lemma~\ref{lem:fid} gives $\E_U[1-F_e(\mathcal C_U,V_a)]=1-F_e(\mathcal E_j,V_a)\le\varepsilon$ and Markov's inequality gives $\Pr_U[1-F_e(\mathcal C_U,V_a)>\eta^2]\le\varepsilon/\eta^2=\gamma$. Let $W^*_j$ be the cheapest $U$ (ties broken by a fixed order) with $1-F_e(\mathcal C_U,V_a)\le\eta^2$, and $W^*_j:=I$ if the coordinate is incorrect. Then in either case
\begin{equation}\label{eq:rep}
\E_{U\sim\pi^{(t)}_j}t(U)\ \ge\ (1-\gamma)\,t(W^*_j).
\end{equation}

\emph{Decoding.} Given $(W^*_j,s,z)$ the decoder knows $\mathcal C_{W^*_j}$ and outputs the list $\mathcal L_j:=\{a':F_e(\mathcal C_{W^*_j},V_{a'})\ge1-\eta^2\}$, which contains $a$ on the correct event. Moreover $\mathcal C_{W^*_j}$ is a mixture of unitaries $X_i=U^>_iW^*_jU^<_i$ with $\sum_ip_i|\tr(V_a^\dagger X_i)/2|^2\ge1-\eta^2$, and $V_{a'}=V_aR_z(\psi)$ with $\psi=\Delta(a'-a)$ reduced to $[-\pi,\pi]$; by Lemma~\ref{lem:conc}(i) every $a'\in\mathcal L_j$ satisfies $\sin(\pi|a'-a|_Q/Q)\le(1+\sqrt3)\eta$, i.e.\ $|a'-a|_Q\le w$, so $|\mathcal L_j|\le2w+1$.

\emph{Counting.} Let $M^*:=(W^*_1,\dots,W^*_m)$ and $\mathcal L:=\prod_j\mathcal L_j$, of size $\le(2w+1)^m$, a function of $(M^*,\sigma_t,Z)$ containing $A$ with probability $\ge1-\delta$. By item (i) of the proof of Theorem~\ref{thm:main1}, $H(A\mid Z)=m\log_2K$, so Lemma~\ref{lem:fanolist} gives $I(A;M^*,\sigma_t\mid Z)\ge(1-\delta)mk_\gamma-h_2(\delta)$. On the other side $I(A;M^*,\sigma_t\mid Z)\le H(\sigma_t)+H(M^*)\le S+\bar T'+m\log_236+\rho_m(\bar T')$ by Definition~\ref{def:model}(iii) and Lemma~\ref{lem:ent}, with $T':=\sum_jt(W^*_j)\ge\sum_jt_{\min}(W^*_j)$; and $\bar T'\le\bar T_t/(1-\gamma)$ by~\eqref{eq:rep} summed over $j$ and averaged. Since $\rho_m$ is increasing, \eqref{eq:mixone} follows.
\end{proof}

\begin{proof}[Proof of Lemma~\ref{lem:nocheap}]
Choose $SU(2)$ representatives and write $U_j=q_{0j}I+i(q_{xj}X+q_{yj}Y+q_{zj}Z)$ with $\sum_\nu q_{\nu j}^2=1$; in the orthonormal basis $(I,iX,iY,iZ)$ of the normalized Choi vector these are its real coordinates. The diamond bound makes the normalized Choi states differ by at most $\delta$ in trace norm, so measuring the projector onto the orthogonal complement of the \emph{identity} channel's Choi vector, and separately the off-diagonal observable between $I$ and $iZ$, gives
\[
\begin{gathered}
\sum_jp_j(1-q_{0j}^2)\le B_\delta:=\sin^2\theta_\delta+\tfrac\delta2,\\
\sum_jp_jq_{0j}q_{zj}\ge\cos\theta_\delta\sin\theta_\delta-\tfrac\delta2 .
\end{gathered}
\]
Fix $M$. There are finitely many projective Clifford+$T$ words of minimal $T$-count $\le M$, so $\mu_M:=\min(1-q_0^2)>0$ over the non-identity ones. By the first inequality the cheap non-identity words carry total probability at most $B_\delta/\mu_M$, so their contribution to the left side of the second has modulus at most $B_\delta/(2\mu_M)=O(\theta_\delta^2)$, while the identity contributes nothing because $q_z=0$ there. Writing $P_{>M}$ for the probability that $t_{\min}(U_j)>M$, Cauchy--Schwarz with the first inequality bounds the remaining contribution by $\sqrt{P_{>M}B_\delta}$. With $\theta_\delta=a\sqrt\delta$ one has $B_\delta=(a^2+\tfrac12)\delta+O(\delta^2)$ while the right side of the second inequality is $a\sqrt\delta-\tfrac\delta2+O(\delta^{3/2})$, of order $\sqrt\delta$; since the cheap part contributes only $O(\delta)$, the expensive part must supply it, giving $\sqrt{P_{>M}(a^2+\tfrac12)\delta}\ge a\sqrt\delta(1-o(1))$ and hence $\liminf P_{>M}\ge a^2/(a^2+\tfrac12)$. As $t(U)\ge t_{\min}(U)$, the expected recorded $T$-count is at least $(M+1)P_{>M}$, and $M$ was arbitrary.
\end{proof}

\begin{lemma}[Frame clustering]\label{lem:cluster}
Let $p,q$ be distributions on finite sets of unitaries $\{U^<_i\}$, $\{U^>_k\}$, let $W,V\in\SU$, put $F_{ik}:=|\tr(V^\dagger U^>_kWU^<_i)/2|^2$, and suppose $\sum_{i,k}p_iq_k(1-F_{ik})\le\eta^2$ with $\eta\le1/8$. Let $P:=\{i:\sum_kq_k(1-F_{ik})\le4\eta^2\}$ and $K:=\{k:\sum_ip_i(1-F_{ik})\le4\eta^2\}$. Then $p(P)\ge3/4$ and $q(K)\ge3/4$; $\dproj(U^<_i,U^<_{i'})\le8\sqrt2\eta$ for all $i,i'\in P$ and $\dproj(U^>_k,U^>_{k'})\le8\sqrt2\eta$ for all $k,k'\in K$; and some pair $(i,k)\in P\times K$ has $\dproj(U^>_kWU^<_i,V)\le4\sqrt2\eta$.
\end{lemma}
\begin{proof}
We use $1-F\le x\Rightarrow\dproj\le\sqrt{2x}$ (Lemma~\ref{lem:conc}(ii) with $\eta^2$ replaced by $x$, valid for $x\le1/4$), together with bi-invariance and the triangle inequality for $\dproj$. Markov's inequality gives $p(P),q(K)\ge3/4$. For $i\in P$ the set $K_i:=\{k:1-F_{ik}\le16\eta^2\}$ has $q(K_i)\ge3/4$, again by Markov. Given $i,i'\in P$ pick $k\in K_i\cap K_{i'}$, a set of mass $\ge1/2$: then $\dproj(U^>_kWU^<_i,V)\le4\sqrt2\eta$ and likewise for $i'$, so $\dproj(U^>_kWU^<_i,U^>_kWU^<_{i'})\le8\sqrt2\eta$, and bi-invariance gives $\dproj(U^<_i,U^<_{i'})\le8\sqrt2\eta$. The suffix statement is symmetric; for the last claim take any $i\in P$ and $k\in K_i\cap K$.
\end{proof}

\begin{proof}[Proof of Theorem~\ref{thm:mixing}(b)]
Keep the notation of (a). Given $(s,z)$ the prefix program $\pi^<$ and the suffix program $\pi^>$ of coordinate $j$ are fixed distributions---the prefix is a function of $(\mathrm{Past},\rho)$ and the suffix of $(\sigma_t,\mathrm{Fut})$, and both $\mathrm{Past}$ and $\mathrm{Fut}$ are part of $z$---and the executor samples them independently. By Lemma~\ref{lem:fid} the representative $W^*_j$ of (a) satisfies $\sum_{i,k}p_iq_k(1-F_{ik})\le\eta^2$ with $F_{ik}=|\tr(V_a^\dagger U^>_kW^*_jU^<_i)/2|^2$, and~\eqref{eq:rep} holds as before.

\emph{A frame independent of the share.} Let $a_0=a_0(s,z)$ be the lexicographically first value in the conditional support of $A_j$ given $(s,z)$ on which coordinate $j$ is correct; if that event has conditional probability zero, set $X:=Y:=I$, in which case $d_j=1$ and the coordinate contributes nothing beyond the $\delta_j\log_2K$ term already carried by the splitting on $E_j$. The anchor $a_0$ is a function of $(s,z)$ alone, and Lemma~\ref{lem:cluster} applies at $a_0$ because correctness at $a_0$ is exactly the hypothesis $\sum_{i,k}p_iq_k(1-F_{ik})\le\eta^2$ for that value. (Taking $a_0=0$ instead would not do. On $Q=8$, where every grid rotation is a power of $T$, the process that stores the round-$t$ share and commits $T^{A_j+a^{(r)}_j}$ in the last round is exactly correct at every share; but under a conditioning in which $a^{(r)}_j=0$ and the true share is $1$, the suffix is $T$ and evaluating $1-F$ at the value $0$ rather than at the true share gives $\sin^2(\pi/8)=0.146$, far above $4\eta^2$, so the clustering sets of Lemma~\ref{lem:cluster} are empty and the lemma does not apply.) Apply Lemma~\ref{lem:cluster} at $a_0$, with $W^*_j(a_0)$ and $V_{a_0}$, and let $i_1,k_1$ be the lexicographically first elements of the resulting sets $P_{a_0},K_{a_0}$; put $X:=U^<_{i_1}$ and $Y:=U^>_{k_1}$, functions of $(s,z)$ alone. For the actual value $a$, Lemma~\ref{lem:cluster} gives $P_a,K_a$ and a pair $(i,k)\in P_a\times K_a$ with $\dproj(U^>_kW^*_jU^<_i,V_a)\le4\sqrt2\eta$. Since $p(P_a),p(P_{a_0})\ge3/4$, some $i_2$ lies in both, and clustering inside $P_a$ and inside $P_{a_0}$ gives $\dproj(U^<_i,X)\le16\sqrt2\eta$; likewise $\dproj(U^>_k,Y)\le16\sqrt2\eta$. Hence
\[
\begin{aligned}
\dproj\big(YW^*_jX,V_a\big)&\le4\sqrt2\eta+16\sqrt2\eta+16\sqrt2\eta\\
&=36\sqrt2\eta=\varepsilon',
\end{aligned}
\]
so $W^*_j\in T_{\varepsilon'}(Y^\dagger V_aX^\dagger)$, the tube of radius $\varepsilon'$ around the coset $G_1R_zG_2$ with $G_1=Y^\dagger$ and $G_2=X^\dagger$, the grid shift being absorbed into $\theta$ as in Corollary~\ref{cor:profile}. By Theorem~\ref{thm:tube} with $\varepsilon'\le1/8$ the number of such words of cost $\le\tau$ is at most $n_{\varepsilon'}(\tau)$. Nearest-grid decoding of $YW^*_jX$ returns $a$ up to $w'$, a grid rotation within $2\varepsilon'$ of $V_a$ being at most $(2Q/\pi)\arcsin\varepsilon'$ steps away, so $\mathcal L_j$ is a window of half-width $w'$.

\emph{Coarse variable.} Let $h:=2w'+1$ and $n:=\lfloor Q/h\rfloor=K'$, and partition $\Z_Q$ into $n$ consecutive cells of sizes $\lfloor Q/n\rfloor$ or $\lceil Q/n\rceil$. Every cell then holds at least $h$ residues, since $n\le Q/h$ gives $\lfloor Q/n\rfloor\ge h$, and at most $2h$: the hypothesis $\varepsilon'\le1/8$ gives $h\le(4Q/\pi)\arcsin\tfrac18+1\le0.16\,Q+1\le Q/2$ for $Q\ge3$, whence $n\ge Q/(2h)$ and $\lceil Q/n\rceil\le2h$. (The sharper $h+1$ is false: $Q=100$ and $\varepsilon'=0.115$ give $w'=7$, $h=15$, $n=6$ and cells of size $17$. Only the bound $2h$ is used.) Let $\bar A_j$ be the cell of $A_j$, so $|\supp\bar A_j|\le K'$. As $A_j$ is uniform on a translate of $[K]$ given $z$ and each cell holds at most $2h$ of its values, $H(\bar A_j\mid z)\ge\log_2K-\log_2(2w'+1)-1=k'_\gamma$, and the $\bar A_j$ are conditionally independent given $Z$, so~\eqref{eq:chain} holds with $\bar A_j$ in place of $A_j$. A window of $h$ consecutive residues meets at most two cells, every cell having at least $h$ of them, so $\mathcal L_j$ determines $\bar A_j$ up to a list of size $2$ and Lemma~\ref{lem:fanolist} gives $I(\bar A_j;W^*_j,\sigma_t\mid Z)\ge(1-\delta)(k'_\gamma-1)-h_2(\delta)-\delta\log_2K'=:\ell'_\delta$.

Now follow the proof of Theorem~\ref{thm:rate2} from~\eqref{eq:chain} with $\bar A_j$ and $\ell'_\delta$, replacing the per-coordinate step by
\[
I\big(\bar A_j;W^*_j\mid E_j=1,s,z\big)\le\tfrac12\E\big[t(W^*_j)\mid E_j=1,s,z\big]+A_2',
\]
which is Lemma~\ref{lem:gibbs} for the pair $(\bar A_j,W^*_j)$ with $|\supp\bar A_j|\le K'$, $\tau_*=\tau_*'$ and the profile $n_{\varepsilon'}$; no conditioning on the frame is needed, because $(X,Y)$ is a function of $(s,z)$. Then $(1-d_j)\E[t(W^*_j)\mid E_j=1,s,z]\le\E[t(W^*_j)\mid s,z]\le\E[t_j\mid s,z]/(1-\gamma)$ by~\eqref{eq:rep}, and summing over $j$ and averaging as in Theorem~\ref{thm:rate2},
\[
m\ell'_\delta\le S+\frac{\bar T_t}{2(1-\gamma)}+mA_2'+mh_2(\delta)+\delta m\log_2K',
\]
which rearranges to~\eqref{eq:mixtwo}.
\end{proof}

\begin{proof}[Proof of Lemma~\ref{lem:cover}]
Let $W\in T_\varepsilon(C)$, $C=G_1R_zG_2$: there is $\theta$ with $\dproj(W,G_1R_z(\theta)G_2)\le\varepsilon$. The $9Q_1$ angles $2\pi(d+i/9)/Q_1$ are spaced $2\pi/(9Q_1)$ apart, so one of them, $\theta_*$, satisfies $|\theta-\theta_*|\le\pi/(9Q_1)$, and $\dproj(R_z(\theta),R_z(\theta_*))=2\sin(|\theta-\theta_*|/4)\le2\sin(\pi/(36Q_1))\le\pi/(18Q_1)$. By maximality of $Q_1$, $\sin(\pi/(2(Q_1+1)))<4\varepsilon$, and $\sin x\ge2x/\pi$ on $[0,\pi/2]$ gives $1/(Q_1+1)<4\varepsilon$, i.e.\ $Q_1>(1-4\varepsilon)/(4\varepsilon)\ge3/(16\varepsilon)$ for $\varepsilon\le1/16$; hence $\pi/(18Q_1)<8\pi\varepsilon/27<\varepsilon$. Since $\dproj$ is bi-invariant, $\dproj(W,G_1R_z(\theta_*)G_2)\le2\varepsilon$, i.e.\ $W\in\mathcal G_i$ for the $i$ with $\theta_*\in2\pi(\Z+i/9)/Q_1$. Each $\mathcal G_i$ is the set counted by Conjecture~\ref{conj:H} for the pair $(G_1R_z(2\pi i/(9Q_1)),G_2)$, modulus $Q_1$ and accuracy $2\varepsilon\le1/8$, and $Q_1\le Q_{2\varepsilon}$ by definition, with $Q_1\ge3$ since $Q_1>3/(16\varepsilon)-1\ge2$; the union bound over $i$ gives~\eqref{eq:cover}.
\end{proof}

\begin{proof}[Proof of Corollary~\ref{cor:mixq}]
Keep the notation of the proof of Theorem~\ref{thm:mixing}(b): the coarse variable $\bar A_j$ with $|\supp\bar A_j|\le K'$, the representative $W^*_j$ in the tube of radius $\varepsilon'$ around $Y^\dagger V_aX^\dagger$, whose frame $(X,Y)$ is a function of $(s,z)$, and~\eqref{eq:rep}. Conditional on $E_j=1$, the alphabet of $W^*_j$ is that tube, whose cost profile is~\eqref{eq:cover} with $\varepsilon'$ in place of $\varepsilon$: $\#\{t_{\min}\le\tau\}\le c_0'+c_1'\tau+c'2^\tau\varepsilon'^2$ with $c_0'=9c_0$, $c_1'=9c_1$, $c'=36c$. This is the hypothesis of Lemma~\ref{lem:cost} with $L'$, $\tau_0'$, $N_0'$ and $b'$ as defined; its requirement $\tau_0'\ge1$ is the restriction $\varepsilon'\le(72c)^{-1/2}$ of the corollary, which reads $2L'\ge\log_2(36c)+1$, and $c_1'\le N_0'/\tau_0'$ then holds automatically. Now follow the proof of Theorem~\ref{thm:main2} with $A_j\to\bar A_j$, $k\to\bar k$, $\ell_\delta\to\ell'_\delta$, $t\to t(W^*_j)$, the frame contributing nothing since $(X,Y)$ is a function of $(s,z)$. The first splitting gives $I(\bar A_j;W^*_j\mid s,z)\le h_2(d_j)+d_j\bar k+(1-d_j)[1+\log_2N_0'+q_j\bar k]$, whence $m\ell'_\delta-S\le m(h_2(\delta)+\delta\bar k)+ma_0'+\bar k\bar{\mathcal A}'^{>}_t$, which is~\eqref{eq:mixactive}. The second uses Lemma~\ref{lem:cost} conditional on $E_j=1$; multiplying by $1-d_j$, averaging, summing over $j$ and using $(1-d_j)\E[t(W^*_j)\mid E_j=1,\cdot]\le\E[t(W^*_j)\mid\cdot]$, concavity, and $\sum_j\E[t(W^*_j)]\le\bar T_t/(1-\gamma)$ from~\eqref{eq:rep}, we get $D_t'\le\bar T_t/(1-\gamma)-b'\bar{\mathcal A}'^{>}_t+ma_0'+3m\log_2(1+\bar T_t/((1-\gamma)m))$; substituting~\eqref{eq:mixactive} and using $b'\ge0$, which is the hypothesis $\varepsilon'\le(72c)^{-1/2}$, gives~\eqref{eq:mixrate}.
\end{proof}

\begin{proof}[Proof of Proposition~\ref{prop:adaptive}]
\emph{Transcripts.} Fix a coordinate and a round and consider the adaptive circuit as a tree of branches. On a branch with outcome string $o$, push every Clifford to the right end: a $T$ gate conjugated by the Cliffords before it becomes a Pauli rotation $e^{-i\pi P/8}$ with $P$ one of the $2\cdot4^n$ signed $n$-qubit Paulis, and a computational-basis measurement becomes a Pauli measurement with outcome $\pm$, one of $2\cdot4^n$ with sign and outcome merged. The branch is therefore determined by a \emph{transcript}: a sequence of $\tau+\mu$ tokens, each from a set of size $4\cdot4^n=2^{2n+2}$, followed by one Clifford in $\mathcal C_n$. The number of transcripts with $N$ tokens is at most $|\mathcal C_n|2^{(2n+2)N}$, and the number of $m$-tuples with at most $N$ tokens in total is at most $|\mathcal C_n|^m2^{(2n+2)N}\binom{N+m}m$; by the Elias-$\gamma$ argument of Lemma~\ref{lem:ent}, an $m$-tuple $M^*$ of transcripts with $N^*$ tokens in total has $H(M^*)\le(2n+2)\bar N^*+m\log_2|\mathcal C_n|+\rho_m(\bar N^*)$.

\emph{Branch probabilities.} On a branch the operation induced on the coordinate's qubit is $K_o=\sqrt{p_o}U_o$ with $U_o$ unitary; $K_o^\dagger K_o=p_oI$ shows that $p_o$ does not depend on the input, that $\sum_op_o=1$, and that the program implements the mixed-unitary channel $\sum_op_o\mathcal U_{U_o}$ on the coordinate's qubit---the ancillas are returned to $|0\rangle$ or measured, so nothing else is carried. The expected numbers of $T$ gates and measurements are $\sum_op_o\tau(o)$ and $\sum_op_o\mu(o)$.

\emph{Representative and decoding.} Everything in the proof of Theorem~\ref{thm:mixing}(a) now applies with ``word'' replaced by ``transcript'' and $t(U)$ by the token count $\tau(o)+\mu(o)$: the composite channel on the coordinate is a mixture of unitaries $X=U^>_{o'}U_oU^<_{o''}$ over independent branches of the three programs, Lemma~\ref{lem:fid} and Markov give a qualifying set of round-$t$ branches of mass $\ge1-\gamma$, the cheapest qualifying transcript $W^*_j$ satisfies $\E[\tau+\mu]\ge(1-\gamma)(\tau+\mu)(W^*_j)$, and Lemma~\ref{lem:conc}(i) applied to the mixture $\mathcal C_{W^*_j}$ list-decodes $A_j$ to a window of half-width $w$. Lemma~\ref{lem:fanolist} and the transcript count replace Lemma~\ref{lem:ent} in the counting step, giving~\eqref{eq:adaptive}.
\end{proof}

\begin{remark}[Constants]\label{rem:mixconst}
Theorem~\ref{thm:mixing} is asymptotic in the same sense as Theorems~\ref{thm:main1} and~\ref{thm:rate2}, and its additive terms are larger, because $\log_236$ and $A_2'$ are now paid against $\tfrac12L$ bits instead of $L$: at $\varepsilon=10^{-10}$, $\gamma=1/4$ and $S=0$, \eqref{eq:mixone} gives $\bar T_t\ge4.1\,m$ and~\eqref{eq:mixtwo} is vacuous; \eqref{eq:mixtwo} becomes positive at $L\ge56$ and reaches $16.6\,m$ at $L=80$ and $102.1\,m$ at $L=200$, against $\tfrac32L=120$ and $300$ achievable. Optimizing $\gamma$ helps little: $54\,m$ instead of $45\,m$ at $L=120$, at $\gamma=0.05$. Two sources of slack are identifiable. The tube radius $36\sqrt2\eta$ of Lemma~\ref{lem:cluster} costs $\log_236\approx5.2$ coarse bits per coordinate against the radius $\sqrt2\eta$ available when the prefix and suffix programs of every coordinate have supports of size at most $2^{n_f}$: in that case Theorem~\ref{thm:mixing}(b) holds with $\varepsilon'=\sqrt2\eta$ and $A_2'+n_f$ in place of $A_2'$, becomes positive from $L\ge45$, and gives $23.7\,m$ at $L=80$. The rest is inherited from Theorem~\ref{thm:tube}, whose term $C_2(k+1)2^{k/2}$ contributes $2\log_2L+\log_2C_2\approx17$ bits to $A_2'$ at $L=33$ and which~\eqref{eq:mixtwo} applies against half as many bits; any improvement of $C_2$ improves both bounds. Corollary~\ref{cor:mixq}, which uses the Conjecture-\ref{conj:H} profile in place of Theorem~\ref{thm:tube}, inherits the same radius and is vacuous at chemistry scale in both forms, so that~\eqref{eq:mixone} ($4.1\,m$) is the only statement of this subsection that is non-vacuous at $\varepsilon=10^{-10}$; the deterministic analogue, Theorem~\ref{thm:main2}, gives $52.9$ against $105$ at the same accuracy. All the numbers in this subsection are computed by \texttt{scripts/mixing\_bounds.py}.
\end{remark}

\end{document}